\pdfoutput=1
\documentclass[11pt]{article}
\usepackage[T1]{fontenc}
\usepackage{geometry} 
\usepackage[english]{babel}
\usepackage{cite}
\usepackage{amsmath}
\usepackage{amsthm}
\usepackage{amsfonts} 
\usepackage{amssymb} 
\usepackage{xspace}
\usepackage{mathtools} 
\usepackage{braket} 
\usepackage[linesnumbered,vlined]{algorithm2e}
\usepackage{xcolor} 
\usepackage{footnote} 
\usepackage{environ} 
\usepackage{suffix} 
\usepackage{apptools}
\usepackage[hidelinks]{hyperref} 
\usepackage{float} 
\usepackage[capitalise]{cleveref} 
\usepackage{comment}
\usepackage[disable,colorinlistoftodos,prependcaption,textsize=tiny]{todonotes}
\usepackage{enumerate}
\usepackage{algorithmicx}
\usepackage[noend]{algpseudocode}
\usepackage[utf8]{inputenc}
\usepackage{graphicx}
\usepackage{svg}
\usepackage{booktabs} 
\usepackage{tabularx} 
\usepackage[flushleft]{threeparttable} 
\usepackage{multirow}

\newcommand{\phil}[1]{\todo[backgroundcolor=green!25]{Phil: #1}}
\newcommand{\yijun}[1]{\todo[backgroundcolor=blue!25]{Yijun: #1}}
\newcommand{\oren}[1]{\todo[backgroundcolor=red!25]{Oren: #1}}

\makesavenoteenv{tabular}
\makesavenoteenv{table}

\makeatletter
\newtheorem*{rep@theorem}{\rep@title}
\newcommand{\newreptheorem}[2]{%
	\newenvironment{rep#1}[1]{%
		\def\rep@title{#2 \ref{##1}}%
		\begin{rep@theorem}}%
		{\end{rep@theorem}}}
\makeatother

\usepackage{thmtools} 
\usepackage{thm-restate} 

\def\eps{\varepsilon}

\newcommand{\size}[1]{\ensuremath{\left|#1\right|}}

\DeclareMathOperator{\dist}{dist}

\DeclareMathOperator{\poly}{poly}

\crefname{claim}{Claim}{Claims}
\crefname{property}{Property}{Properties}
\crefname{algocf}{Algorithm}{Algorithms}
\Crefname{algocf}{Algorithm}{Algorithms}

\newtheorem{theorem}{Theorem}
\newreptheorem{theorem}{Theorem}
\newtheorem{lemma}{Lemma}[section]

\newtheorem{definition}[lemma]{Definition}

\newtheorem{corollary}[lemma]{Corollary}
\newtheorem{observation}[lemma]{Observation}

\newtheorem{remark}[lemma]{Remark}
\newtheorem*{remark*}{Remark}

\theoremstyle{remark}

\theoremstyle{definition}

\newcommand{\E}{\mathbf{E}}
\newcommand{\Prob}{\mathbf{P}}

\DeclarePairedDelimiter{\ceil}{\lceil}{\rceil}
\newcommand{\ip}[1]{\left}

\newcommand{\hybrid}{\ensuremath{\mathsf{Hybrid}}\xspace}

\newcommand{\nccshort}{\ensuremath{\mathsf{NCC}}\xspace}
\newcommand{\ncczero}{\ensuremath{\nccshort_0}\xspace}
\newcommand{\congest}{\ensuremath{\mathsf{CONGEST}}\xspace}
\newcommand{\local}{\ensuremath{\mathsf{LOCAL}}\xspace}
\newcommand{\clique}{\ensuremath{\mathsf{Congested\ Clique}}\xspace}

\newcommand{\bccshort}{\ensuremath{\mathsf{BCC}}\xspace}

\newcommand{\cwc}{\ensuremath{\mathsf{CWC}}\xspace}
\newcommand{\hybridzero}{\ensuremath{\mathsf{Hybrid}_0}\xspace}

\newcommand{\kdis}{$k$\textsc{-dissemination}\xspace}
\newcommand{\klrout}[1][k,\ell]{$(#1)$\textsc{-routing}\xspace}
\newcommand{\kagg}{$k$\textsc{-aggregation}\xspace}

\newcommand\bigOa{\ensuremath{{O}}}
\newcommand\bigO[1]{\ensuremath{{O}(#1)}}
\WithSuffix\newcommand\bigO*[1]{\ensuremath{{O}\left(#1\right)}}
\newcommand\tildeBigO[1]{\ensuremath{{\widetilde{{O}}}(#1)}}
\WithSuffix\newcommand\tildeBigO*[1]{\ensuremath{{\widetilde{{O}}}\left(#1\right)}}
\newcommand\littleO[1]{\ensuremath{{o}(#1)}}
\WithSuffix\newcommand\littleO*[1]{\ensuremath{{o}\left(#1\right)}}
\newcommand\tildeLittleO[1]{\ensuremath{{\widetilde{{o}}}(#1)}}
\WithSuffix\newcommand\tildeLittleO*[1]{\ensuremath{{\widetilde{{o}}}\left(#1\right)}}
\newcommand\bigOmega[1]{\ensuremath{{\Omega}(#1)}}
\WithSuffix\newcommand\bigOmega*[1]{\ensuremath{{\Omega}\left(#1\right)}}
\newcommand\tildeBigOmega[1]{\ensuremath{{\widetilde{{\Omega}}}(#1)}}
\WithSuffix\newcommand\tildeBigOmega*[1]{\ensuremath{{\widetilde{{\Omega}}}\left(#1\right)}}
\newcommand\littleOmega[1]{\ensuremath{{\omega}(#1)}}
\WithSuffix\newcommand\littleOmega*[1]{\ensuremath{{\omega}\left(#1\right)}}
\newcommand\tildeLittleOmega[1]{\ensuremath{{\widetilde{{\omega}}}(#1)}}
\WithSuffix\newcommand\tildeLittleOmega*[1]{\ensuremath{{\widetilde{{\omega}}{\left(#1\right)}}}}
\newcommand\bigTheta[1]{\ensuremath{{\Theta}(#1)}}
\WithSuffix\newcommand\bigTheta*[1]{\ensuremath{{\Theta}\left(#1\right)}}
\newcommand\tildeTheta[1]{\ensuremath{{\widetilde{{\Theta}}}(#1)}}
\WithSuffix\newcommand\tildeTheta*[1]{\ensuremath{{\widetilde{{\Theta}}\left(#1\right)}}}

\makeatletter
\newcommand*{\whp}{%
    \@ifnextchar{.}%
        {w.h.p}%
        {w.h.p.\@\xspace}%
}
\makeatother

\newcommand{\polylog}{\poly\log}

\RestyleAlgo{ruled} 
\LinesNumbered
\let\oldnl\nl
\newcommand{\nonl}{\renewcommand{\nl}{\let\nl\oldnl}}

\SetKwFunction{AssignHelpers}{Assign-Helpers}
\SetKwFunction{TransmitSkeleton}{Transmit-Skeleton}
\SetKwFunction{LearnNeighborhood}{Learn-Neighborhood}
\SetKwFunction{LearnHelpers}{Learn-Helpers}
\SetKwFunction{ReassignSkeletons}{Reassign-Skeletons}

\NewEnviron{killcontents}{}

\newcommand{\tk}[1][k]{\mathcal{NQ}_{#1}}
\newcommand{\tn}[1][n]{\mathcal{NQ}_{#1}}

\newcommand{\hop}{\operatorname{hop}}
\DeclareMathOperator*{\argmax}{arg\,max}
\DeclareMathOperator*{\argmin}{arg\,min}

\newcommand{\tilO}{\smash{\ensuremath{\widetilde{O}}}}
\newcommand{\tilOm}{\smash{\ensuremath{\widetilde{\Omega}}}}
\newcommand{\tilT}{\smash{\ensuremath{\widetilde{\Theta}}}}
\newcommand{\tild}{\smash{\ensuremath{\widetilde{d}}}}

\newcommand{\hybridpar}[2]{\ensuremath{\mathsf{Hybrid}(#1,#2)}}

\newcommand{\LOCAL}{\ensuremath{\mathsf{LOCAL}}\xspace}
\newcommand{\CONGEST}{\ensuremath{\mathsf{CONGEST}}\xspace}

\newcommand{\NCC}{\ensuremath{\mathsf{NCC}}\xspace}

\newcommand{\CC}{\ensuremath{\mathsf{Congested\, Clique}}\xspace}

\newcommand{\MPC}{\ensuremath{\mathsf{MPC}}\xspace}

\newcommand{\calA}{\mathcal{A}}
\newcommand{\calB}{\mathcal{B}}

\newcommand{\calH}{\mathcal{H}}

\newcommand{\calS}{\mathcal{S}}
\newcommand{\calT}{\mathcal{T}}

\newcommand{\euler}{$\mathsf{Eulerian}$-$\mathsf{Orientation}$\xspace}
\newcommand{\oracle}{$\mathcal{O}^{\operatorname{Euler}}$\xspace}
\newcommand{\minor}{$\mathsf{Minor}$-$\mathsf{Aggregation}$\xspace}

\newcommand{\NQ}[1][k]{\ensuremath{\mathcal{NQ}_{#1}}\xspace}

\newcommand{\ID}{\ensuremath{\text{\normalfont ID}}\xspace}

\newcommand{\p}{\ensuremath{\!+\!}}
\newcommand{\m}{\ensuremath{\!-\!}}

\title{Universally Optimal Information Dissemination and Shortest Paths in the HYBRID Distributed Model\footnote{This work combines the three preprints \href{https://arxiv.org/abs/2304.06317}{arXiv:2304.06317}, \href{https://arxiv.org/abs/2304.07107}{arXiv:2304.07107}, and \href{https://arxiv.org/abs/2306.05977}{arXiv:2306.05977} together with some improved results.}}

\author{
Yi-Jun Chang 
\\
\small{National University of Singapore}
\\
\small{\texttt{cyijun@nus.edu.sg}}
\and
Oren Hecht
\\
\small{Technion}
\\
\small{\texttt{hecht.oren@campus.technion.ac.il}}
\and
Dean Leitersdorf
\\
\small{National University of Singapore}
\\
\small{\texttt{dean.leitersdorf@gmail.com}}
\and
Philipp Schneider
\\
\small{University of Bern}
\\
\small{\texttt{philipp.schneider2@unibe.ch}}
}

\date{}

\begin{document}
\begin{titlepage}
\maketitle
\thispagestyle{empty}

\begin{abstract}
    Distributed computing deals with the challenge of solving a problem whose input is distributed among the computational nodes of a network and each node has to compute its part of the problem output. The main complexity measure of a distributed algorithm is the number of required communication rounds, which depends highly on the communication model that is considered.

    
    In most modern networks, nodes have access to various modes of communication each with different characteristics.
    In this work we consider the \hybrid model of distributed computing, introduced recently by Augustine, Hinnenthal, Kuhn, Scheideler, and Schneider (SODA 2020), where nodes have access to two different communication modes: high-bandwidth local communication along the edges of the graph and low-bandwidth all-to-all communication,  
    capturing the non-uniform nature of modern communication networks. It is noteworthy that the \hybrid model in its most general form covers most of the classical distributed models as marginal cases.

    Prior work in \hybrid has focused on showing \emph{existentially} optimal algorithms, meaning there exists a pathological family of instances on which no algorithm can do better. This neglects the fact that such worst-case instances often do not appear or can be actively avoided in practice. In this work, we focus on the notion of \emph{universal optimality}, first raised by Garay, Kutten, and Peleg~(FOCS 1993).
    Roughly speaking, a universally optimal algorithm is one that, given any input graph, runs as fast as the best algorithm designed specifically for that graph.    
    
    We show the \emph{first} universally optimal algorithms in \hybrid, up to polylogarithmic factors. We present universally optimal solutions for fundamental tools that solve information dissemination problems, such as broadcasting and unicasting multiple messages in a \hybrid network.    
    Furthermore, we demonstrate these tools for information dissemination can be used to obtain universally optimal solutions for various shortest paths problems in \hybrid. 
    
    A major conceptual contribution of this work is the conception of a new graph parameter called \emph{neighborhood quality} that captures the inherent complexity of many fundamental graph problems in the \hybrid model.

    We also develop new existentially optimal shortest paths algorithms in \hybrid, which are utilized as key subroutines in our universally optimal algorithms and are of independent interest. Our new approximation algorithms for $k$-source shortest paths
  match the existing $\tildeBigOmega{\sqrt{k}}$ lower bound \emph{for all} $k$. Previously, the lower bound was only known to be tight when \smash{$k \in \tildeBigOmega{n^{2/3}}$}.
\end{abstract}

\end{titlepage}

\thispagestyle{empty}
{\small
\tableofcontents
}
\thispagestyle{empty}

\newpage

\setcounter{page}{1}

\section{Introduction}

In most contemporary networks, nodes can interface with multiple, diverse communication infrastructures and use different \emph{modes of communication} to improve communication characteristics like bandwidth and latency. Such \emph{hybrid} networks appear in many real-world applications. For instance, cell phones typically combine high-bandwidth short-ranged communication (e.g.,~communication between nearby smartphones using Bluetooth, WiFi Direct, or LTE Direct) with global communication through a lower-bandwidth cellular infrastructure~\cite{Kar2018}. As another example, organizations augment their networks with communication over the Internet using virtual private networks (VPNs)~\cite{Rossberg2011}. Furthermore, data centers combine limited wireless communication with high-bandwidth short-ranged wired communication~\cite{Halperin2011, Farrington2010}.

Theoretic research on distributed problems concentrated mostly on distributed models where nodes can only use a single, uniform {mode of communication} to exchange data among each other. These classic models come in two different ``flavors''. The first is based on  \textit{local communication} in a graph in synchronous rounds, where adjacent nodes may exchange a message each round (for instance \LOCAL, \CONGEST).
The second flavor captures \textit{global communication}, where any pair of nodes can exchange messages but only at a very limited bandwidth (e.g.\ \CC, or \MPC). 
In a sense, theoretical research lags behind industry-led efforts and applied research into so-called heterogeneous networks, which aim to realize gains in capacity and coverage offered by combining global communication (e.g., via different tiers of cellular networks) with local links \cite{Yeh2011}.

This began to change relatively recently with the introduction of the so-called \hybrid model \cite{augustine2020shortest}, which combines two modes of communication with fundamentally different characteristics. First, a graph-based local mode that allows neighbors to communicate by exchanging relatively large messages, representing high bandwidth, short-range communication where the main issue is data \emph{locality}, for instance, WiFi. Second, a global mode where any pair of nodes can communicate in principle but only at a very limited bandwidth, which reflects the issue of \emph{congestion} of communication via shared infrastructure, like cell towers. From a practical standpoint, this model aims to reflect better the heterogeneous nature of contemporary networks. From a theoretical point of view, the model asks the natural question of what can be achieved if limited global communication is added to a high bandwidth local communication mode. 

The dissemination of information in the network and the computation of shortest paths of a given input graph are fundamental tasks in real-world networks. Information dissemination, such as broadcasting or unicasting multiple messages or aggregating values distributed across the network, is interesting on its own as an end goal. In particular, such algorithms can be applied to announcing a failure, a change of policy, or other control messages in a data center. Furthermore, algorithms for information dissemination often serve as subroutines for solving other problems. Computing shortest paths allows us to gather information about the topology of the network. It is often used as a subroutine for related tasks like computing or updating tables for IP routing, which forms the backbone of the modern-day Internet.
A series of works on general graphs in  \hybrid \cite{augustine2020shortest, Kuhn2020, CensorHillel2021} have narrowed down the computational complexity of these problems (see Tables \ref{tab:info_dissem},\ref{tab:apsp},\ref{tab:klsp},\ref{tab:sssp} and \cref{fig:kssp_complexity_overview}). 

However, all prior work concentrated on \emph{existentially} optimal algorithms, meaning that the associated algorithms are essentially optimized to deal with pathological families of worst-case instances on which no algorithm can do better!
The focus on {existential} optimality comes with a significant drawback, since such worst-case graphs are often artificial constructions that might not appear or can be actively avoided in practical contexts.
In fact, previous work has shown that up to exponentially faster solutions for information dissemination and shortest paths algorithms are possible on specific classes of graphs in \hybrid or weaker models \cite{Augustine2019, coy2022routing, Coy2021Near, Feldmann2020}. In general, many networks of interest can offer drastically faster algorithms compared to the existential lower bounds.
Therefore, a worthy goal is to design \emph{universally optimal} algorithms. Loosely speaking, a universally optimal algorithm runs as fast as possible on \emph{any} graph, not just worst-case graphs. The concept of universal optimality  was first theorized in the distributed setting by Garay, Kutten, and
Peleg~\cite{Garay1998}. In particular, they asked the following question.

\begin{center}
\begin{minipage}{0.85\textwidth}
\textit{``The interesting question that arises is, therefore, whether it is possible to identify the inherent graph parameters associated with the distributed complexity of various fundamental network problems and to develop universally optimal algorithms for
them.''}
\end{minipage}
\end{center}

In this work, we consider solving fundamental distributed tasks such as information dissemination and shortest paths computation in the \hybrid model of distributed computing in a \emph{universally optimal} fashion -- i.e., we aim to design algorithms that are competitive even against the best algorithm designed specifically for the underlying graph topology.

\subsection{Our Contribution}

We give a condensed summary of our results with a comparison to the current state of the art. An overview that also explains our technical contributions is given in \cref{sec:technicaloverview}.
Before we proceed, we provide an informal definition of the \hybrid model. We are given an initial input graph $G = (V, E)$, and proceed in synchronous time steps called rounds. In each round, any two adjacent nodes in $G$ can exchange an arbitrarily large message, through the \emph{local network}. Further, any pair of nodes can exchange a $O(\log n)$ bit message through the \emph{global network} but each node can only send or receive $O(\log n)$ such messages per round. In \hybridzero nodes can only send global messages to nodes whose identifiers they know, and initially they know only those of their neighbors in $G$. Precise definitions are given in Section \ref{sec:models}.

A main conceptual contribution of this work is a new graph parameter called \emph{neighborhood quality} that captures the inherent complexity of many fundamental graph problems in the \hybrid model. Informally, $\NQ$ is the minimum distance $t$ such that the ball of radius $t$ around $v$ is sufficiently large to allow any node $v$ to exchange $\tildeBigOmega{k}$ bits of information with other nodes in $O(t)$ rounds via the global network.
In a sense, the parameter $\NQ$ dictates how effectively nodes can locally collaborate to solve a global distributed problem with some ``workload'' $k$ such as information dissemination and shortest paths. 
We believe that this could have significant implications for distributed paradigms such as \emph{edge computing} where the goal is to handle as much workload as possible locally in collaboration with nearby nodes using the local network, while minimizing the (more costly) global communication.\phil{I wanted to make this connection to edge computing since I believe it is important, but I don't know if it makes sense to others.} \oren{It makes sense to me, I think it's fine :)}
To compare the following results with the state of the art, it suffices to know that $1 \leq \NQ \leq \sqrt{k}$ (as shown in \cref{sec:neigh_qual}).

\paragraph{Universally optimal information dissemination.} We show the \emph{first universally optimal} algorithms for information dissemination in \hybrid. We prove that $\NQ$ constitutes a universal upper bound for broadcasting $k$ messages, aggregating $k$ messages, and unicasting an individual message from each of $k$ source nodes to each of $\ell \leq k$ target nodes (for a total of $k \cdot \ell$ distinct messages), for the precise problem definitions see \cref{sec:problems}. Furthermore, we prove that on the specific local communication graph $G$, no algorithm can solve these information dissemination problems in less than $\tildeBigOmega{\tk}$ rounds, even if the algorithm is tailor-made for the topology of $G$. \cref{tab:info_dissem} summarizes our results and prior results.

\bigskip

\begin{table}[h]
    \centering
    \setlength{\tabcolsep}{1mm}
    \begin{tabularx}{\textwidth}{@{\hskip 1mm}c *6{>{\centering\arraybackslash}X}@{}}
        \toprule
        reference & problem & bound & model & randomness & notes\\
        \midrule
        \cite{augustine2020shortest} & \multirow{3}{*}{\shortstack[c]{broadcast and\\ aggregation}} & $\tilO\big(\!\sqrt{k} \p {\ell}\big)$ & \hybrid & randomized & ex.\ opt.${}^*$ \\
        \cref{thm:optimal_dissemination} & & $\tilO(\NQ)$ & \hybridzero & deterministic & univ.\ opt.\ \\
        \cref{thm:broadcast_unicast_LB} & & $\tilOm(\NQ)$ & \hybrid & randomized & univ.\ opt.\ \\
        \midrule
        \cite{Kuhn2020} & \multirow{3}{*}{unicast} & $\tilO(\!\sqrt{k}
    \p k \ell/n)$ & \hybrid & randomized & $\ell \leq k$${}^\dagger$ \\ 
        \cref{thm:routing} & & $\tilO(\NQ)$ & \hybrid & randomized &  $\ell \leq k^\ddagger$\\ 
        \cref{thm:broadcast_unicast_LB} &  & $\tilOm(\NQ)$ & \hybrid & randomized & univ.\ opt.${}^\ddagger$ \\ 
        \bottomrule
    \end{tabularx}
    \begin{threeparttable}
        \begin{tablenotes}
            \footnotesize
            \item $*$ Where $\ell$ is the maximum initial number of messages per node, matching lower bounds $\tilOm\big(\!\sqrt{k}\big)$ and $\tilOm\big({\ell}\big)$ \cite{Schneider2023}.
            \item $\dagger$ Existentially optimal in $k$ \cite{Schneider2023}. Further conditions apply regarding distribution of source and target sets.
            \item $\ddagger$ Universally optimal. Further conditions apply regarding size or distribution of source and target sets, see Thm.\ \ref{thm:routing}.            
	\end{tablenotes}
    \end{threeparttable}
    \caption{Results comparison for information dissemination.}
    \label{tab:info_dissem}
\end{table}

\paragraph{Universally optimal shortest paths computation.} We consider various shortest paths problems, where the most general is the $(k,\ell)$-SP problem, where $\ell$ target nodes must learn their distance to $k$ source nodes. Furthermore, we cover the special cases of $k$-SSP $:= (k,n)$-SP, APSP := $n$-SSP and SSSP := $1$-SSP. In the approximate version of the problem, nodes must learn a distance that is precise up to a stretch factor. See Section \ref{sec:problems} for formal definitions.

In the \hybrid model, many shortest paths problems are settled with existentially optimal bounds (with an important gap remaining, see Figure \ref{fig:kssp_complexity_overview}).
Using our tools for information dissemination, we can obtain the first universally optimal solutions for shortest paths problems, that are, in general, faster than the existentially optimal ones. Tables \ref{tab:apsp} and \ref{tab:klsp} show the known and our results.

\bigskip

\begin{table}[h]
    \centering
    \setlength{\tabcolsep}{1mm}
    \begin{tabularx}{\textwidth}{@{\hskip 1mm}c *8{>{\centering\arraybackslash}X}@{}}
        \toprule
        reference & bound & stretch & model & randomness & weights & notes \\
        \midrule
        \cite{Kuhn2020} &  $\tilO\big(\!\sqrt{n}\big)$ & exact & \hybrid & randomized & weighted & exist.\ opt.\  \\
        \cite{anagnostides2021deterministic} & $\tilO\big(\!\sqrt{n}\big)$ & $\bigOa\big(\frac{\log n}{\log \log n}\big)$ & \hybrid & deterministic & weighted & exist.\ opt.\   \\
        \cite{anagnostides2021deterministic} &  $\tilO\big(\!\sqrt{n}\big)$ & $1 \p \eps$ & \hybrid & deterministic & unweighted & exist.\ opt.\  \\
        \cite{augustine2020shortest} & $\tilOm\big(\!\sqrt{n}\big)$ & $\bigOa\big(\!\sqrt{n}\big)$ & \hybrid & randomized & unweighted & exist.\ opt.\ \\
        \midrule
        \cref{thm:unweighted_apsp_approx} &  $\tilO(\tn)$ & $1 \p \eps$ & \hybridzero & deterministic & unweighted & univ.\ opt.\ \\    
        \cref{thm:firstWeightedAPSPApprox} &  $\tilO(\tn)$ & \smash{$\bigOa\big(\frac{\log n}{\log \log n}\big)$} & \hybridzero & deterministic & weighted &  univ.\ opt.\ \\ \cref{theorem:secondWeightedAPSPApprox} &  \smash{$\tilO\big({n^{1/4}\tn^{1/2}}\big)$} & 3 & \hybridzero & randomized & weighted & exist.\ opt.\ \\ \cref{thm:apsp_hybridzero_unweighted} &  $\tilOm(\tn)$ & any & \hybridzero & randomized & unweighted &  univ.\ opt.\ \\     \cref{thm:klsp_LB},\ref{thm:klsp_LB_sources_known}&  $\tilOm(\tn)$ & $\poly n$ & \hybrid & randomized & weighted &  univ.\ opt.\ \\
        \bottomrule
    \end{tabularx}
    \caption{Results comparison for the all pairs shortest paths (APSP) problem.}
    \label{tab:apsp}
\end{table}

\yijun{Question: why is $\tilO\big({n^{1/4}\tn^{1/2}}\big)$ existentially optimal? I see. It is at most $\sqrt{n}$}
\phil{we could also say that it is univ. opt. with competitiveness factor $n^{1/4}/\tn^{1/2}$.}
\phil{added another line to the table for the weighted lower bounds in \hybrid}

\bigskip

\begin{table}[h]
    \centering
    \setlength{\tabcolsep}{1mm}
    \begin{tabularx}{\textwidth}{@{\hskip 1mm}c *9{>{\centering\arraybackslash}X}@{}}
        \toprule
        reference & problem & bound & stretch & model & randomness & weights & notes\\
        \midrule
        \cite{Kuhn2020} & $(k,1)$-SP & $\tilOm\big(\!\sqrt{k}\big)$ & $\bigOa\big(\!\sqrt{n}\big)$ & \hybrid & randomized & unweighted & exist.\ opt.\  \\  
        \cref{thm:klsp_UB} & $(k,\ell)$-SP & $\tilO(\NQ)$ & $1 \p \eps$ & \hybrid & randomized & weighted & $\ell \leq k^*$\\
        Thm.\ \ref{thm:klsp_LB},\ref{thm:klsp_LB_sources_known} & $(k,1)$-SP & $\tilOm(\NQ)$ & $\poly n$ & \hybrid & randomized & weighted & univ.\ opt.${}^*$  \\
        \bottomrule
    \end{tabularx}
\begin{threeparttable}
        \begin{tablenotes}
            \footnotesize
            \item $*$ Universally optimal. Further conditions apply regarding size or distribution of source and target sets, see Thm.\ \ref{thm:klsp_UB}.         
	\end{tablenotes}
    \end{threeparttable}
    \caption{Results comparison for the $(k,\ell)$-shortest paths ($(k,\ell)$-SP) problem.}
    \label{tab:klsp}
\end{table}

\oren{Should we highlight again that $\NQ\leq{}\sqrt{k}$? I feel like we said it one time in a hidden spot and it could be missed}
\phil{I don't think so due to redundancy. The place above where we say it is also relatively prominent and we even suggest that the reader should remember this fact for interpreting/comparing the results.}

\paragraph{Existentially optimal shortest paths computation.} Last but not least, we also consider existentially optimal algorithms for the SSSP and the $k$-SSP problem. We use this as a tool in our universally optimal shortest paths algorithms, but our solutions also close some of the remaining gaps for shortest path problems in the \hybrid model. Table \ref{tab:sssp} gives an overview of the state of the art and our results for the SSSP problem and Figure \ref{fig:kssp_complexity_overview} gives a visual representation for $k$-SSP.

\bigskip

\begin{table}[h]
    \centering
    \setlength{\tabcolsep}{1mm}
    \begin{tabularx}{\textwidth}{@{\hskip 1mm}c *6{>{\centering\arraybackslash}X}@{}}
        \toprule
        reference & bound & stretch & model & randomness  \\
        \midrule
        \cite{anagnostides2021deterministic} & \smash{$\tilO(n^{1/2})$} & \smash{$\frac{\log n}{\log\log n}$}  &  \hybrid & deterministic  \\         
        
        \cite{censor2021sparsity} & $\tilO(n^{5/17})$ & $1 \p \eps$ & \hybrid & randomized  \\ 
                
        \cite{augustine2020shortest} & $\tilO(n^{\eps})$ & $(1/\eps)^{\bigOa(1/\eps)}$ & \hybrid & randomized  \\ 
        
        \cref{thm:almost_shortest_sssp} & $\tilO(1)$ & $1 \p \eps$ & \hybridzero & deterministic  \\   
        \bottomrule 
    \end{tabularx}
    \caption{Results comparison for the single-source shortest paths (SSSP) problem.}
    \label{tab:sssp}
\end{table}
\phil{rearranged the rows a little (sorted by running time) because I think it highlights the better all the parameters where we have improvement.}

\begin{figure}[h]
    \centering
    \includegraphics[width=0.73\textwidth]{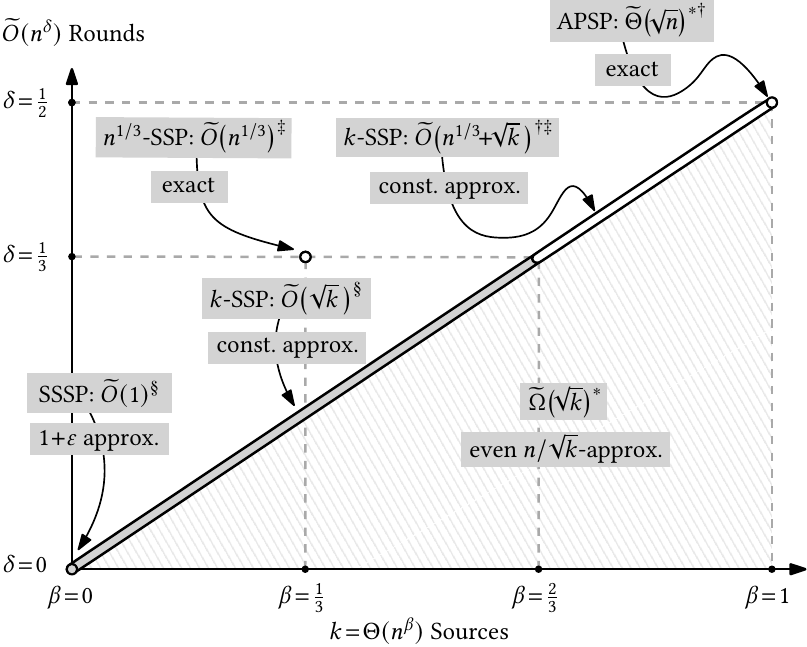}
    \caption{(Existential) complexity landscape of the $k$-SSP problem with the number of sources on the horizontal and the round complexity on the vertical axis. Circles or bars denote known upper bounds (ours in gray). The gray shaded area denotes the lower bound. References are as follows $*$: \cite{augustine2020shortest}, $\dagger$: \cite{Kuhn2020}, $\ddagger$: \cite{CensorHillel2021}, \S: this work.}		\label{fig:kssp_complexity_overview}
\end{figure}

\subsection{Notations}
Unless otherwise stated, we consider undirected and connected graphs $G = (V, E, \omega)$, with $n = |V|$, $m = |E|$, and an edge-weight function $\omega$ such that all the edge weights $\omega(e)$ are polynomial in $n$. If the graph is unweighted, then $\omega \equiv 1$. The weighted distance between two nodes $v,w\in{V}$ is denoted by $d(v,w)$. The hop distance between $v,w\in{V}$ is denoted by $\hop(v,w)$ and is the unweighted distance between the two nodes. We also extend this definition to node sets:  $\hop(A,B) = \min_{v \in A, \, w \in B} \hop(v,w)$.
The diameter of a graph is denoted by $D$ and is defined by $\max_{v, w \in V} \hop(v,w)$. Denote by $d^h(v,w)$ the weight of a shortest path between $v$ and $w$ among all $v$-$w$ paths of at most $h$ hops. 

Let 
$\calB_t(v)=\{w\in{}V \, | \, \hop(v,w)\leq{}t\}$
be the ball of radius $t$ centered at $v$. Given a node set $S \subseteq V$, we write $\calB_t(S) = \bigcup_{v \in S} \calB_t(v)$.
Given a set of nodes $C\subseteq V$, the \emph{weak} diameter of $C$ is defined as $\max_{u,v\in{}C}\hop(u,v)$, where the hop distance is measured in the original graph $G$. The \emph{strong} diameter of $C$ is the diameter of the subgraph of $G$ induced by $C$. For any positive integer $x$, let $[x] = \{1, 2, \ldots, x\}$. Unless otherwise specified, all logarithms $\log x := \log_2 x$ are of base-$2$.

We say that an event occurs \emph{with high probability} (\whp) if it happens with probability at least $1 - 1/n^c$ for some large constant $c>0$, where $n$ is the number of nodes in the underlying local communication network. We use the notations $\tildeBigO{\cdot}$, $\tildeBigOmega{\cdot}$, and $\tildeTheta{\cdot}$ to hide any polylogarithmic factors. For example, $\tilO(1) = \polylog n$.

\subsection{Models}
\label{sec:models}

We now formally define the \hybrid and \hybridzero models, introduced recently by Augustine, Hinnenthal, Kuhn, Scheideler, and Schneider~\cite{augustine2020shortest}. In these models, the local communication network is abstracted as a graph $G=(V,E)$, where each node $v \in V$ corresponds to a computing device and each edge $e = \{u,v\} \in E$ corresponds to a local communication link between $u$ and $v$.
Communication happens in synchronous rounds. In each round, nodes can perform arbitrary local computations, following which they communicate with each other using the following two communication modes.

\begin{description}
    \item[Unlimited local communication:] Local communication is modeled with the standard \local model of distributed computing of~\cite{linial1992locality}, where for any $e = \{v, u\} \in E$, nodes $v$ and $u$ can communicate any number of bits over $e$ in a round.
    \item[Limited global communication:] Global communication is modeled with the \emph{node-capacitated congested clique} (\nccshort) model~\cite{Augustine2019}, where every node can exchange $O(\log{n})$-bit messages with up to any $O(\log{n})$ nodes in $G$. In other words, it is required that each node is the sender and receiver of at most $O(\log{n})$ messages of $O(\log{n})$ bits per round.
\end{description}
There is some subtlety in the model definition regarding how we handle the situation where the number of messages sent to a node $v$ exceeds the bound $O(\log{n})$. For example, one reasonable option is to assume that there is an adversary that drops an arbitrary number of messages and leaves only $O(\log{n})$ messages delivered to $v$. This subtlety in the model definition is not an issue for us because, for all our algorithms, it is guaranteed that the bound $O(\log{n})$ is satisfied deterministically or \whp.

\paragraph{Identifiers.}
We assume that each node $v \in V$ is initially equipped with an $O(\log n)$-bit distinct identifier $\ID(v)$.
We make a distinction based on the assumption about the range of identifiers. In the \hybrid model, the set of all identifiers is exactly $[n]$. In the \hybridzero model, the range of identifiers is $[n^c]$ for some constant $c \geq 1$. At the start of an algorithm, a node only knows its identifier and the identifiers of its neighbors. Consequently, the global communication is over \ncczero~\cite{augustine2021distributed} instead of \nccshort. 
Compared with \hybrid, a main challenge in designing algorithms in \hybridzero is that a node might not know which identifiers are used in the graph and thus can only send messages to nodes whose identifiers it knows. For example, suppose a node wants to send a message to a uniformly random node. In \hybrid, this task can be done by selecting a number $s \in [n]$ uniformly at random and then sending a message to the node whose identifier is $s$. Such an approach does not apply to \hybridzero.
 If an algorithm works in \hybridzero, then it works in \hybrid too, and if a lower bound holds in \hybrid, then it holds in \hybridzero too.

    


\paragraph{Parameterization.}
It is possible to parameterize the \hybrid model by two parameters: $\lambda$ is the maximum message size for local edges and $\gamma$ is the maximum number of bits that each node can communicate per round via global communication. Therefore, the standard \hybrid model can be seen as $\hybrid(\lambda,\gamma)$ with $\lambda=\infty$ and $\gamma \in O(\log^2{n})$, as we have unlimited-bandwidth local communication and allow each node to send $O(\log n)$-bit messages to $O(\log n)$ nodes via global communication. As the number of messages is not constrained by $O(\log n)$ in $\hybrid(\infty, O(\log^2{n}))$, this model has slightly more flexibility compared with \hybrid, but they are equivalent up to $\tilO(1)$ factors in the round complexity.  
Many standard distributed models can be seen as specific cases of $\hybrid(\lambda,\gamma)$ or $\hybridzero(\lambda,\gamma)$, as follows, where $\approx$ indicates equivalence up to $\tilO(1)$ factors.
\begin{align*}
\clique & \approx \hybrid(0,O(n\log{n})) &&&
    \local &= \hybridzero(\infty, 0) \\
\NCC & \approx \hybrid(0,O(\log^2 {n})) &&&
    \congest &= \hybridzero(O(\log n), 0)\\
\ncczero & \approx \hybridzero(0,O(\log^2 {n}))
\end{align*}


\subsection{Problems}
\label{sec:problems}

We provide formal definitions for the distributed tasks considered in this work. We start with the fundamental distributed task of broadcasting $k$ messages to the entire network, which is interesting on its own as an end goal (e.g., to broadcast a network
update or notification of failure) and interesting as a basic building block for solving other problems (e.g., for computing paths and distances between nodes).


\begin{definition}[\kdis]
\label{def:kdis}
    Given any set of messages $M$ of $\bigO{\log n}$ bits, where $k = |M|$ and each $m \in M$ is originally known to only one node in the graph, the \kdis problem requires that all messages $M$ become known to every node in the graph.
\end{definition}

The \kdis problem is one of the most basic \hybrid communication primitives. In the paper defining the \hybrid model~\cite{augustine2020shortest}, it was shown that \kdis problem can be solved \whp in $\tildeBigO{\sqrt{k}+\ell}$ rounds, where $\ell$ is the maximum number of messages that a node can hold at the beginning. Recently, a deterministic algorithm operating in $\tildeBigO{\sqrt{k}}$ rounds was shown in~\cite{anagnostides2021deterministic}, but their algorithm requires $k\geq{}n$. In our paper, we consider the most general version of the problem, with no bounds on $k$ and no limitation on the original distribution of the messages in the graph. The number of messages initially at a node can be any number in $0, 1, \ldots, k$.

Next, we consider a related distributed task called \kagg. In the following definition, we say that a function $F: X\times X \rightarrow X$ is an aggregation function if $F$ is associative and commutative. Examples of aggregation functions include minimum, maximum and sum.

\begin{definition}[\kagg]
    \label{def:kagg}
    Let $F: X\times X \rightarrow X$ be an aggregation function with $|X| \in n^{O(1)}$. Assume each node $v$ originally holds $k$ values $f_1(v), \dots, f_k(v)$. The \kagg problem requires that each node learn all the values $F(f_i(v_1),\dots,f_i(v_n))$, for all $i \in [k]$. 
\end{definition}

Next, we consider the \emph{unicast} problem where multiple, distinct messages are routed between specific source and target nodes.

\begin{definition}[\klrout]
\label{def:klrout}
Let $S \subseteq V$ be a set of source nodes. Let $T \subseteq V$ be a set of target nodes. Each source $s  \in S$ has an individual message intended for each target $t \in T$ that is labelled with $\ID(t)$. 
The \klrout problem requires every target node $t \in T$ to know all the $|S|$ messages that are addressed to $t$. We consider four different scenarios.
\begin{description}
    \item[Arbitrary sources and arbitrary targets:] $S$ is a set of $k$ arbitrary nodes, and $T$ is a set of $\ell$ arbitrary nodes
    \item[Arbitrary sources and random targets:] $S$ is a set of $k$ arbitrary nodes, and $T$ is selected by letting each node in $V$ join $T$ independently with probability $\ell/n$.
    \item[Random sources and arbitrary targets:] $S$ is selected by letting each node in $V$ join $S$ independently with probability $k/n$, and $T$ is a set of $\ell$ arbitrary nodes
    \item[Random sources and random targets:] $S$ is selected by letting each node in $V$ join $S$ independently with probability $k/n$, and $T$ is selected by letting each node in $V$ join $T$ independently with probability $\ell/n$.
\end{description}
\end{definition}

In different scenarios, the roles of the two parameters $k$ and $\ell$ are slightly different. In the arbitrary setting, these parameters specify the size of the corresponding node sets. In the randomized setting, these parameters specify the \emph{expected} size of the corresponding node sets. Furthermore, in our randomized algorithms, a random set of sources $S$ or targets $T$ is fixed after it was sampled and cannot be re-sampled. Randomized sources or targets still help tremendously for our upper bounds. By contrast if resampling would be allowed the corresponding lower bounds would not hold. In fact, much better running times are possible if we could re-sample source or target nodes in case we are not satisfied with the result.

\paragraph{Shortest paths.} Let $S \subseteq V$ be a set of source nodes. Let $T \subseteq V$ be a set of target nodes. The $(k,\ell)$-{Shortest Paths} ($(k,\ell)$-SP) problem requires every target $t \in T$ to learn $d_G(s,t)$ for all sources $s \in S$ and match the distance label $d_G(s,t)$ to the identifier $\ID(s)$ of the corresponding source node $s \in S$.\phil{that assignment is important for the lower bound proof in unweighted graphs and \hybridzero} In the $\alpha$-approximate version of the problem for \emph{stretch} $\alpha \geq 1$, every target node $t$ has to compute a distance estimate $\widetilde{d}_G(s,t)$ such that $d_G(s,t)\leq \widetilde{d}_G(s,t)\leq \alpha d_G(s,t)$ for all source nodes $s \in S$.

Same as \klrout, here we also consider four different variants of the $(k,\ell)$-SP problem.
For example, in the setting where the sources are arbitrary and the targets are random, $S$ is an arbitrary set of $k$ nodes, and $T$ is selected by letting each node join $T$ with probability $\ell/n$ independently.
Many natural variants of the shortest paths problems can be stated as special cases of $(k,\ell)$-SP with $\ell = n$ targets $T = V$.
\begin{align*}
   \text{$(1,n)$-SP} &= \text{single-source shortest paths (SSSP)} \\
   \text{$(k,n)$-SP} &= \text{$k$-source shortest paths ($k$-SSP)}\\
   \text{$(n,n)$-SP} &= \text{all-pairs shortest paths (APSP)}
\end{align*}



We assume that the local communication network and the input graph for the considered graph problem are the same. This is a standard assumption for distributed models with graph-based communication such as \LOCAL and \CONGEST and reflects the desire to gather topological information about the local network.


\paragraph{Universal optimality.}


Our concept of universal optimality adheres closely to \cite{haeupler2021universally} which bases itself on a description by \cite{Garay1998}. Consider a problem instance of some problem $\Pi$, which is defined by certain problem inputs. We split such an instance $(S,I) \in \Pi$ into a fixed part $S$ and a parametric input $I$. The intuition behind this partition is that we want to design universal algorithms that are (provably) competitive against algorithms that have the advantage that all nodes are initially granted knowledge of $S$, but not $I$. Therefore, we have to define how powerful we make the algorithms that we compete against by deciding what part of the problem instances of $\Pi$ we assign to $S$ and what we assign to $I$.

In our problems of information dissemination and distance computation in \hybrid, we always assign the local graph $G$ to $S$. More concretely, in case of an information dissemination instance, e.g., $(S,I) \in $ \kdis (see \cref{sec:problems}), the fixed part $S$ contains $G$ but also the starting locations of all $k$ messages, whereas the contents of the messages are part of $I$. 
In the case of a distance problem like $(k, \ell)$-SP, consider an instance $(S,I)\in(k, \ell)$-SP. Then $S$ contains (at least) the graph $G$, whereas $I$ contains (at most) an assignment of identifiers $\ID:V \to [n]$, the weight function $\omega$, and the set of source and target nodes. 
However our universal algorithms can often be made competitive even against algorithms that have additional knowledge like, e.g., the weight function $\omega$, which means that $\omega$ is part of $S$, too.

In order to highlight why the split of an instance of $\Pi$ into $S$ and $I$ is important, consider an even stronger concept than universal optimality called \textit{instance optimality} (also theorized by \cite{Garay1998}). An algorithm is instance optimal if it would be competitive with any algorithm optimized for $S$ \textit{and} the input $I$.
Unfortunately, the concept of instance optimality is not very interesting in the \hybrid model. This is due to the fact that in the \hybrid model, nodes can detect extremely fast (in $\bigOa(\log n)$ rounds to be precise) whether they all live in a specific problem instance $(S,I) \in \Pi$ by using the global network and if so output a hard-coded solution, else run the trivial $D$-round algorithm. 
So unless a graph problem $\Pi$ has $\tilO(1)$ complexity in \hybrid in general, instance optimality with competitiveness $\tilO(1)$ is unattainable for any algorithm that is \textit{oblivious} to $(S,I)$. 
Hence, designing algorithms that are competitive with the algorithm where nodes know $S$ but not $I$ is a more fruitful concept. 
Formally, we define universal optimality as follows.

\begin{definition}[Universal Optimality, see \cite{Garay1998, haeupler2021universally}]
	\label{def:uni_opt}
	Let $\calA$ be an algorithm that correctly computes the solution to some distributed problem $\,\Pi$ with probability at least $p>0$ in some computational model $\mathcal M$ and takes $T_{\calA}(S,I)$ rounds for $(S,I) \in \Pi$. Then $\calA$ is called a \emph{universally optimal} model $\mathcal M$ algorithm with competitiveness $C$ (omitted for $C \in \tilO(1)$) if the following holds. Let $\Pi_S := \{I \mid (S,I) \in \Pi \}$. For all fixed inputs $S$ and for all algorithms $\calA'$ that solve $(S,I) \in \Pi$ with probability at least $p$ in time $T_{\calA'}(S,I)$, we have $$T_{\calA}(S,I) \leq C \cdot \max\limits_{I \in \Pi_S} T_{\calA'}(S,I).$$
\end{definition}

In order to satisfy the definition above and show universality of an algorithm $\calA$ we conduct the following steps. First, we prove that our algorithm $\calA$ achieves a certain time $T$ without nodes having knowledge of $(S,I)$. Second, we prove a lower bound of $\tilOm(T)$ that holds even when nodes are given initial knowledge of $S$ (but not $I$).

\subsection{Further Related Work}

The \hybrid model in its current form was introduced in \cite{augustine2020shortest}. Since then, most research in \hybrid focused on information dissemination, shortest paths computations and closely related problems such as diameter calculation or computation of routing schemes. 
These classes of problems are well suited for \hybrid due to the following reasons: (a) they are ``global'' in the sense that using only \local it takes $D$ rounds to solve them, (b) they require significant information exchange, i.e., $\tilOm(n)$ rounds are required if only \NCC would be used \cite{augustine2020shortest, Kuhn2020, Kuhn2022}, and (c) they are not too demanding on the computation nodes have to conduct locally (as would be the case for global, NP-hard problems, e.g., \emph{optimal} vertex coloring).

The stated goal of prior work was typically to show that the combination of \local and \NCC in the \hybrid model offers significant speedup over using either \local or \NCC. 
For an overview of the current state of the art for shortest paths computation, see~\cref{tab:info_dissem,tab:apsp,tab:klsp,tab:sssp,fig:kssp_complexity_overview}. 
Computation of the graph diameter in \hybrid was considered in \cite{Kuhn2020, anagnostides2021deterministic, CensorHillel2021}, the computation of routing schemes in \cite{Kuhn2022, Coy2021Near, coy2022routing} and distance problems in classes of sparse graphs in \cite{Feldmann2020, Coy2021Near, coy2022routing}.

Several other distributed models are of a hybrid nature, such as \cite{gmyr2017distributed, Afek1990,baDISC}. One model related to \hybrid is the Computing with Cloud (\cwc) model introduced by \cite{afek2021distributed}, which considers a network of computational nodes, together with usually one passive storage cloud node. They explore how to efficiently run a joint computation, utilizing the shared cloud storage and subject to different capacity restrictions. We were inspired by their work on how to analyze neighborhoods of nodes to define an optimal parameter for a given graph which limits communication.

There is a series of exciting research on universally optimal algorithms~\cite{haeupler2021universally,Haeupler2021low,ghaffari_dfs,ghaffari_mixing_time,haeupler2016near, haeupler2018faster, haeupler2018minor, haeupler2022hop, kitamura2021low, rozhovn2022undirected, zuzic2018towards, Ghaffari2022, zuzic2022universally,ghafarriOPTBroadcast} that  addressed the research question suggested in the quote of Garay, Kutten, and
Peleg~\cite{Garay1998} in an interesting setting. Specifically, they identified a graph parameter called \emph{shortcut quality} $\mathcal{SQ}(G)$ that accurately captures the inherent complexity of many fundamental problems, including minimum spanning tree (MST), approximate SSSP, and approximate minimum cut in the \congest model of distributed computing. For these problems, $\tildeBigOmega{\mathcal{SQ}(G)}$ is a lower bound even for algorithms that are specifically designed for the network $G$ and know its graph topology~\cite{haeupler2021universally}. In the known-topology setting, universally optimal algorithms that attain the matching upper bound $\tildeBigO{\mathcal{SQ}(G)}$ have been shown~\cite{haeupler2021universally}. Even in the more challenging unknown-topology setting, a weaker upper bound of the form $\poly(\mathcal{SQ}(G)) \cdot n^{o(1)}$ can be attained~\cite{haeupler2022hop}.

Elkin~\cite{elkin2006faster} considered universally optimal algorithms in the \local model. He focused on a class of MST algorithms that maintain a set of edges in each round that eventually converge to the correct solution. He defined a graph parameter $\mu(G, \omega)$ called \emph{MST-radius}, which depends on both the graph topology $G$ and the edge weights $\omega$, and he showed that this parameter captures the universal complexity for the considered class of algorithms.

\subsection{Roadmap}
In \cref{sec:technicaloverview}, we formally state all our main results and discuss the key technical ideas that underlie their proofs.
In \cref{sec:neigh_qual}, we define the neighborhood quality of a graph and analyze its properties to compare it to other graph parameters.
In \cref{sec:optimal_broadcast}, we show universally optimal algorithms for \kdis and \kagg.
In \cref{sec:optimal_unicast},  we show universally optimal algorithms for \klrout.
In \cref{sec:upper_bound}, we utilize our broadcast and unicast algorithms to give universally optimal algorithms for approximating distances and cuts.
In \cref{sec:lower_bound}, we present lower bounds to demonstrate the universal optimality of our algorithms.

Our universally optimal shortest paths algorithms utilize existentially optimal SSSP and $k$-SSP algorithms as subroutines. In \cref{sec:sssp_logtime}, we show that a $(1 \p \eps)$-approximation of SSSP can be computed in $\tilO(1/\eps^2)$ rounds in \hybridzero deterministically. In \cref{sec:kssp}, we extend this SSSP algorithm to give existentially optimal algorithms for $k$-SSP.

In \cref{apx:generalnotations}, we present a few basic probabilistic tools that we use throughout this paper.
In \cref{sec:graph_families}, we analyze the neighborhood quality for some graph families to give a taste as to how $\tk(G)$ relates to other graph parameters such as $n$ and $D$ and to show that our \emph{universally} optimal algorithm improves significantly over the existing state-of-the-art \emph{existentially} optimal algorithms in these graph families. In \cref{sec:node_comm}, we review \emph{the node communication problem}, which is an abstraction for the problem where a part of the network has to learn some information from a distant part of the network. We use this tool in our lower bound proofs.

\phil{although I really like the roadmap, I think it would be a candidate for removal if we are desperate for space for the PODC submission.}
\phil{If we keep it the references might need to be updated.}
\yijun{The roadmap can also appear at the beginning of the appendix}
\phil{Yes I think that's how we should do it.}


\section{Technical Overview}\label{sec:technicaloverview}

In this section, we formally state all our main results and provide some additional consequences that these results have. The goal is to provide intuition and discuss the key technical ideas underlying their proofs, along with the corresponding challenges.

We begin by defining the \emph{neighborhood quality} of a graph. Roughly speaking, it describes the size of the neighborhood each node has available in order to collaboratively handle a ``workload'' of size $k$. Given a graph $G = (V, E)$, a value $k$, and a node $v$, let
    $\tk(v) = \min \{\{t \mid \size{\calB_t(v)} \geq k/t\} \cup{} \{D\}\}$ and
    and let $\tk(G) = \max_{v \in V} \tk(v)$.
For the most part, we drop the $G$ in $\NQ(G)$ and write $\tk$ instead. 
For this technical overview, we only briefly mention that the quality of neighborhoods behaves inversely in \NQ, i.e., neighborhoods have higher quality if $\NQ$ is small, whereas $1 \leq \NQ \leq \sqrt{k}$. \cref{sec:neigh_qual} is dedicated to an in depth analysis of the parameter $\NQ$.

\subsection{Universally Optimal Information Dissemination}

Before we start summarizing our lower and upper bounds of $\tilde\Theta(\tk)$ for solving \kdis, \kagg, and \klrout for certain ranges of parameters $(k,\ell)$ with random target nodes in \hybrid, we observe a couple of interesting properties.

First, this bound shows that in \hybrid, the complexity of these information dissemination tasks depends \emph{only on the graph topology}.\phil{in some sense it also depends on k, but I think it's ok since $\NQ$ depends only on the graph for any fixed k} In particular, the complexity of \kdis does not depend on the original distribution of the tokens to broadcast. 

Intuitively, this makes sense as in \hybrid within $t$ rounds every node in $G$ can communicate with every neighbor it has within $t$ hops. At the same time, in those $t$ rounds each of those nodes can receive $\tilO(t)$ messages through the global network. Thus the size of the $t$-hop neighborhood of a node in $G$ dictates an upper bound on the amount of information that node can receive in $t$ rounds by combining the global and local networks. 


\paragraph{Multi-message broadcast upper bound.} 
We consider \kdis and \kagg (see \cref{def:kdis} and \cref{def:kagg}) in \cref{sec:optimal_broadcast}, for which we show universally optimal algorithms that are \emph{deterministic}.

\begin{restatable}[Broadcast]{theorem}{kdisUP}
\label{thm:optimal_dissemination}
The \kdis problem can be solved
    in $\tildeBigO{\tk}$ rounds {deterministically} in the \hybridzero model.
\end{restatable}

\begin{restatable}[Aggregation]{theorem}{kaggUB}
\label{thm:kaggUB}
The \kagg problem can be solved in $\tildeBigO{\tk}$ rounds {deterministically} in the \hybridzero model.
\end{restatable}

One interesting consequence of \cref{thm:optimal_dissemination} is that it offers a simple $\tildeBigO{\tn}$-round preprocessing step such that after the step any \hybrid algorithm can be emulated in \hybridzero without any slowdown. We just run \cref{thm:optimal_dissemination} with $k=n$ to broadcast all the identifiers. After that, effectively we may assume that the range of identifiers is $[n]$ as every node $v$ knows the rank of $\ID(v)$ in the list of all identifiers.

Since the algorithms for \kdis and \kagg are based on the same framework, in the technical overview we only focus on \kdis.
To show \cref{thm:optimal_dissemination}, we partition the graph into \emph{clusters}, each with diameter $\tildeBigO{\tk}$ and with $\Theta(k/\tk)$ nodes. We desire to ensure that all the $k$ messages arrive at each cluster, which will allow each node to ultimately learn all the $k$ messages by using the local edges to receive all the information its cluster has.

To do so, we begin by building a virtual tree of all the clusters. While trivial in the \hybrid model, it is rather challenging in \hybridzero, as we must construct a tree that spans the entire graph, has $\tildeBigO{1}$ depth and maximum degree, and every two nodes in the tree know the identifiers of each other even though they may be distant in the original graph. To do so, we build upon certain overlay construction techniques from~\cite{gmyr2017distributed}.

Then, between any parent $P$ and child $C$ clusters in the binary tree, we ensure that every node $P$ knows the identifier of exactly one node in $C$ and vice versa, so that they may communicate through the global network. 
Reaching this state requires great care, as sending the identifiers of all the nodes in one cluster to all the nodes in another cluster must be done over the bandwidth restricted global network, since the clusters might be physically distant from each other.

\oren{I think the above paragraph is phrased a bit awkwardly maybe - I agree with its idea following our reviews, but I felt it's repeating the same point many times. Maybe it will be better to distribute it throughout the technical section to not give the impression we "try to hard to show our work is complex"?}
\phil{I shortened this a little (original text commented out)}


Once nodes in clusters $P$ and $C$ can communicate with each other, we can converge-cast all $k$ messages in parallel in the cluster tree. To prevent congestion in the global network, we perform load-balancing steps within each cluster to ensure that each of its nodes is responsible for sending roughly the same number of messages.

\phil{Did some more shortening, mainly removing the redundancy that Oren was mentioning, original text in the comments.}

\paragraph{Multi-message unicast upper bound.}
We proceed with the algorithmic results for the \klrout problem in \cref{sec:optimal_unicast}, see \cref{def:klrout} for the definition.

\begin{restatable}[Unicast]{theorem}{UnicastUB}
	\label{thm:routing}
The \klrout problem can be solved w.h.p.\ in the \hybrid model with the following round complexities.
    \begin{enumerate}[(1)]
        \item $\tilO(\NQ)$ rounds, for $\ell \leq \NQ$, with arbitrary sources and random targets.
        \item $\tilO(\mathcal{NQ}_\ell)$ rounds, for $k \leq \mathcal{NQ}_\ell$, with random sources and arbitrary targets.
        \item $\tilO(\max\{\NQ,\mathcal{NQ}_\ell\})$ rounds, for  $k \cdot \ell  \leq \NQ \cdot n$, with random sources and random targets.
    \end{enumerate}
\end{restatable}

We point out a few observations about the above results. First, the solutions for \klrout are strictly better than broadcasting all $k \ell$ messages in the network, which takes $\tilT(\NQ[k\ell])$ rounds, and in general $\NQ \ll \NQ[k \ell]$. Compared to broadcasting, this allows us to handle many more distinct messages in $\tilO(\NQ)$ rounds. Second, the requirement $k \cdot \ell  \leq \NQ \cdot n$ in case (3) of \cref{thm:routing} is at the theoretical limit (up to a $\tilO(1)$ factor) as this is the most traffic that the global network can handle in total in $\tilO(\NQ)$ rounds and it is not too hard to see that in large diameter networks, a constant fraction of all messages need to be transmitted via the global network eventually.

Third, in case (3) of \cref{thm:routing} the restriction that each source node has at most one message for each target can be lifted, we only require that each node sends at most $k$ and receives at most $\ell$ messages. If $k$ and $\ell$ are roughly equal, then this can be seen from a reduction, as our base routing algorithm can simulate the so-called \CC model among all source and target nodes, where any pair of source or target nodes can exchange a message each round. It is known due to \cite{Lenzen2013} that it only takes a constant number of rounds in the \CC to deliver $k\ell$ messages with the eased condition. In another reduction step, we can even drop the assumption that $k,\ell$ are roughly equal (cf.\ \cref{lem:consolidate_sources}).

Let us briefly discuss our algorithmic solution. First, we compute for each source and each target node an adaptive set of helpers. This concept is a generalization of the {helper sets} in~\cite{Kuhn2020}, where we take into account the quality of the neighborhoods expressed by $\NQ$. It means that each source and each target node is assigned a set of nearby helper nodes, that they can rely on almost exclusively to increase their communication bandwidth. It also explains the condition of randomized source or target nodes, since this ensures that they are ``well spread out'' w.h.p., allowing the appointment of helpers that can be used (almost) exclusively.

On a high level, the messages are routed from sources to their respective helpers over the local network. From there, the global network is used to send messages from the helpers of the source nodes to the helpers of the target nodes. Finally, the target nodes can collect their assigned messages from their respective helpers via the local network. This process poses two main challenges. 

First, the sets of helper nodes of sources and targets do not mutually know the identifiers of each other, and thus cannot directly forward the messages via the global network. Reducing the information that needs to be exchanged in order to ``connect'' helpers of sources and targets requires quite some technical work, by routing via intermediate nodes that are determined using universal hashing. Second, computing these helper sets is not possible in case $k$ or $\ell$ is very large (e.g., $k=n, \ell = \tk$). We show instead that the general case can be reduced to $\tilO(1)$ routing instances with ``good'' parameter ranges $\ell, k \leq \sqrt{n \cdot \NQ}$.

\paragraph{Information dissemination lower bounds.} To prove the universal optimality of our upper bounds for information dissemination we match them with corresponding lower bounds in \cref{sec:lower_bound}, which hold even for algorithms that have access to the graph topology.


\begin{restatable}[Information dissemination lower bounds]{theorem}{broadcastunicastLB}
\label{thm:kdisLB}    \label{thm:broadcast_unicast_LB}
    The \kdis, \kagg, and \klrout problems with arbitrary target nodes take $\tildeBigOmega{\tk}$ rounds in expectation in the \hybrid model. This holds for randomized algorithms with constant success probability $p>0$. For the case of the \klrout problem with random target nodes, the same lower bound holds for algorithms with a probability of failure of less than \smash{$\frac{(1-p)\ell}{n}$} for any constant $0<p\leq 1$. For \kdis the lower bound holds for any distribution of messages.
\phil{The success probability can be decreased to $(1-\frac{(1-p)\ell}{n})^{\NQ}$ but the formula is not so nice and I'm not sure how to simplify or meaningfully upper bound it. It's anyway already tight since our upper bounds have higher success probability.} 
\end{restatable}

We start by discussing the results of \cref{thm:broadcast_unicast_LB}. First, the $\tildeBigOmega{\tk}$ bound holds even for randomized algorithms and for the weaker \hybrid model, showing that $\tilT(\NQ)$ is a universal bound for \kdis and \kagg. Second, in case of \klrout with randomly sampled target nodes, the lower bound holds for any algorithm that fails with probability less than $\frac{(1-p)\ell}{n}$. Compare this to our upper bound, which holds w.h.p., i.e., the failure probability is at most $1/n^c$ for \emph{any} constant $c>0$. Thus, our upper bounds for \klrout are almost tight up to $\tilO(1)$ factors, which narrows down the complexity \klrout to $\tilT(\NQ)$ for the parameter ranges and conditions given in \cref{thm:routing}.

The general idea of proving lower bounds in the hybrid model is as follows. The aim is to create an ``information gap'' in the network, which is formalized by a random variable whose outcome is known by one subset of nodes $A$ but unknown to some other subset $B$ that is at a relatively large distance in the local network. This problem has been named the ``node communication problem'' (see \cref{def:node_comm_problem}), and a reduction from Shannon's source coding theorem \cite{Shannon1948} shows a lower bound that depends on the entropy of $X$, the hop distance between $A$ and $B$ and, crucially, the size of the neighborhoods of $A$ and $B$, see \cite{Kuhn2022, Schneider2023}.

We show how the node communication problem can be adapted into our setting for universal lower bounds. In particular, we prove in \cref{lem:receive_lower_bound} that there is a node for which learning a set of $k$ ``information tokens'' that are distributed arbitrarily in the network, requires at least $\tilOm(\NQ)$ rounds.
Besides the arbitrary token distribution, the challenge for the lower bound is that we cannot customize our graph to represent some worst-case instance as it was done in prior work, but can only rely on the properties inherent to the graph parameter $\NQ$.
This \cref{lem:receive_lower_bound} forms the foundation for our reductions to the \kdis and \klrout problem. Another reduction from \kdis to \kagg, extends the lower bound of \tilOm(\NQ) to the latter.

\paragraph{Application.}
Combining the upper bound and lower bound for \kdis has interesting consequences for the simulation of the \emph{broadcast congested clique} (\bccshort)~\cite{drucker2014on}, which is a distributed model where in each round each node can broadcast a message of $O(\log n)$ bits to the entire network. Using broadcast as the only communication primitive, it was shown that many fundamental problems can be solved efficiently in \bccshort~\cite{chen2019broadcast, becker2015hierarchy, jurdzinski2018connectivity, becker2020impact, montealegre2016brief, holzer2015approximation}. 
An immediate corollary of our \kdis algorithm is a  \bccshort simulation in \hybrid, which requires broadcasting $k=n$ tokens spread uniformly across $G$, which can be done using  $\tildeBigO{\tn}$ rounds by our broadcast algorithm. This is optimal due to our  $\tildeBigOmega{\tk}$ lower bound for \kdis. 

\begin{corollary}[Implied by \cref{thm:optimal_dissemination,thm:broadcast_unicast_LB}]
    There is a universal lower bound of $\tildeBigOmega{\tn}$ rounds for simulating one round of \bccshort in \hybrid that holds for randomized algorithms with constant success probability, and there exists a deterministic algorithm which does so in $\tildeBigO{\tn}$ rounds in  \hybridzero.
\end{corollary}

\subsection{Universally Optimal Shortest Paths}

We show a variety of applications for our information dissemination tools for graph problems, in particular, we obtain solutions whose running time can be expressed in the graph parameter $\NQ$. Furthermore, we give matching lower bounds for these graph problems, which proves their universal optimality. 
\phil{It was honest and not wrong to say that the upper bounds of this subsection are straightforwad, but I think there is a danger that this is generalized to all results in this subsection, including lower bounds and maybe some reviewers will carry that even over to the SSSP solutions in the subsequent subsection. I rewrote the inital sentence a little.} 

\paragraph{{\boldmath $(k,\ell)$}-SP.} First of all, using our multi-message unicast algorithm, we obtain the following $(k,\ell)$-SP algorithm. A key technical ingredient needed to establish the following theorem is an efficient algorithm for $k$-SSP, which we will discuss later.

\begin{restatable}[$(k,\ell)$-SP]{theorem}{uniKLSP}
	\label{thm:klsp_UB}
    Let $\eps > 0$ be an arbitrary constant. The $(k,\ell)$-SP problem can approximated with stretch $(1 \p \eps)$ w.h.p.\ in the \hybrid model in $\tilO\left(\NQ\right)$ rounds if one of the conditions is met.
    \begin{itemize}
        \item  The set of source nodes is an arbitrary set of $k$ nodes, the target nodes are sampled with probability \smash{$\frac{\ell}{n}$}, and $\ell \leq \NQ$.
        \item  The source and target nodes are sampled with probability \smash{$\frac{\ell}{n}$} and \smash{$\frac{k}{n}$}, respectively, with $\ell \leq \NQ^2$ and $\ell \cdot k \leq \NQ \cdot n$.
    \end{itemize} 
\end{restatable}

We start by discussing these results. Note that the best existentially optimal algorithms for $k$-SSP all require  $\tilT\big(\!\sqrt{k}\big)$ rounds (see \cite{Kuhn2020} with a remaining gap for $k \leq n^{2/3}$ that we also close in this work). This is inherent, as there exists a worst-case graph where a single target node learning its distance to $k$ source nodes $(k,1)$-SP has a lower bound of $\tilOm\big(\!\sqrt{k}\big)$. Note that our upper bound of $\tilO(\NQ)$ is, in general, faster.

We emphasize that the parameter ranges and conditions on randomization of source and target nodes required in \cref{thm:klsp_UB}, are a direct consequence of the fact that we have to communicate $k\cdot \ell$ distance labels from sources to targets. Therefore, the same restrictions on the overall communication capacity of the global network and the requirement that sources and targets are ``well spread out'' in the local network that we discussed for \cref{thm:routing} apply here as well.

The first result of \cref{thm:klsp_UB} is a consequence of our existentially optimal shortest path algorithm for a $(1\p\eps)$ approximation of SSSP in \tilO(1) rounds (\cref{thm:almost_shortest_sssp}), which we can use to solve SSSP sequentially for each target node as a ``source'' in $\tilO(\NQ)$ rounds. Afterwards, each source $s$ knows the approximate distance $\tild(s,t)$ to each target. To solve the $(k,\ell)$-SP problem, we need to ``reverse'' this situation. Therefore, each source node $s$ creates a message to each target node $t$ containing $\tild(s,t)$ which we deliver with an instance of \klrout using \cref{thm:routing} case~(1).

The second result works similarly, although we have to be more efficient in solving the $\ell$-SSP problem for the set of target nodes as sources. For this we employ our existentially optimal $k$-SSP algorithm with stretch $(1\p \eps)$ in $\tilO\big(\!\sqrt \ell\big) \subseteq \tilO(\NQ)$ rounds (\cref{thm:k-ssp}, also discussed later). Again, this gives a solution ``in reverse'' where the source nodes know the distances that the target nodes need to learn, which corresponds to \klrout that can be solved using \cref{thm:routing} case~(3).\phil{I just explained the whole thing, given that it is only application of other theorems. I think it's good in the sense that people see how routing and existentially optimal SP can be used, but it can also be shortened.}

\oren{I think it looks good, and agree that if we have to shorten the paper, we can shorten maybe the paragraph about randomization? I think we emphasized this point before.}

\paragraph{APSP.} 
In previous works, exact weighted APSP was settled with existential lower and upper bounds of $\tildeBigO{\sqrt{n}}$ rounds, due to \cite{augustine2020shortest, Kuhn2020}. Their algorithms are randomized, while the best known deterministic algorithm of \cite{anagnostides2021deterministic} runs in $\tildeTheta{\sqrt{n}}$ rounds to produce an $O(\log{n}/\log{\log{n}})$-approximation. For unweighted APSP, \cite{anagnostides2021deterministic} showed a $(1+\eps)$-approximation, deterministically, in $\tildeBigO{\sqrt{n}}$ rounds as well.

We show several algorithms for APSP that are based on our universally optimal broadcasting protocol and we stress that these work in \hybridzero, where identifiers can be arbitrary $O(\log n)$ strings and most results are deterministic. Most of our algorithms below run in $\tildeBigO{\tn}$ rounds, and thus are never slower than the $\tildeTheta{\!\sqrt{n}}$ round algorithms of \cite{augustine2020shortest, Kuhn2020}.

The most immediate consequence of the broadcast protocol is that in sparse graphs one can efficiently learn the entire graph and thus solve any graph problem locally. That is, in graphs with $\tildeBigO{n}$ edges, we apply \cref{thm:optimal_dissemination} with $k=n$ at most $\tildeBigO{1}$ times, resulting in $\tildeBigO{\tn}$ rounds, note that the corresponding existential bound in \cite{Anagnostides2021} takes $\tilO\big(\!\sqrt{n}\big)$ rounds.

\phil{I went over this first paragraph and tried to concentrate all the general remarks about the APSP results and related work in to this paragraph and also reduce redundancy. The argument that the given algorithms are universally optimal will be given later when the matching lower bounds are presented.}

\begin{corollary}
    \label{cor:APSP_sparse_graph}
    Given a sparse, weighted graph $G=(V,E,\omega)$ with $|E|=\tildeBigO{n}$, there is a deterministic algorithm that solves any graph problem in $\tildeBigO{\tn}$ rounds in \hybridzero, including exact weighted APSP.
\end{corollary}

\oren{Should we put this as the first theorem/corollary of this APSP subsection? It does look a bit simplified and immediate, since it's the first thing that jumps to the eyes of the reader. I would suggest leaving this corollary as a sentence }
\phil{Changed the structure accordingly. I think it is good to have it as a corollary. I was wondering if we should generalize it for any number of edges with a running time of $\tildeBigO{\tn \cdot |E|/n}$? It is also deterministic, right? It's important to mention as \cite{Anagnostides2021} has a similar existeintially optimal argument that is deterministic.}

To deal in graphs with any number of edges and still solve APSP in $\tilO(\NQ)$ rounds, we combine the universally optimal broadcast from \cref{thm:optimal_dissemination} with additional techniques of computing spanners, skeleton graphs and our existentially optimal SSSP solution.

We show a $(1+\eps)$ approximation of unweighted APSP which runs in $\tildeBigO{\tn/\eps^2}$ rounds deterministically. The best known algorithm for exactly computing or approximating APSP for general graphs to practically any factor takes $\tilOm\big(\sqrt{n}\big)$ rounds~\cite{Kuhn2020}. As $\tn \leq \sqrt{n}$, we are always at least as fast, and are faster when $\tn \in o(\sqrt{n})$.

\begin{restatable}[Unweighted APSP]{theorem}{unweightedAPSPApprox}
\label{thm:unweighted_apsp_approx}
    For any $\eps \in (0, 1)$, there is a  deterministic algorithm that computes a $(1+\eps)$-approximation of APSP in unweighted graphs in \hybridzero in $\widetilde{O}(\tn/\eps^2)$ rounds.
\end{restatable}

A key technical ingredient of the proof of \cref{thm:unweighted_apsp_approx} is our $\widetilde{O}(1/\eps^2)$-round deterministic $(1+\eps)$-approximate SSSP algorithm, which we will discuss later. Given the $(1+\eps)$-approximate SSSP algorithm, the idea behind the unweighted APSP algorithm is as follows. We first compute a $\tildeBigO{\tn}$-weak diameter clustering consisting of at most $\tn$ clusters, c.f., \cref{lem:clustering_with_weak_diam}. Then, we execute the $(1+\eps)$-approximate SSSP algorithm from cluster centers. Each node then determines the distance to its closest cluster center locally, and all these distances are broadcast. Finally, each node can approximate its distance to all others well enough using the $(1+\eps)$ approximations to cluster centers and the local distances that have been broadcast. 

\oren{I guess we can remove most of the above paragraph on the unweighted algorithm, as it's not to complex and based on existing subroutines. (Both to shorten and to not appear too simple)}

\phil{I have no strong opinion here, in fact I can see your argument but also that this small explanation here is actually nice (made some slight changes to it).}

We now show several results for approximating APSP in weighted graphs.
We begin with the following algorithm that computes an $(\eps \cdot \log n)$-approximation, and is based on the result of \cite{rozhovn2020polylogarithmic} for constructing a graph spanner, which is a sparse subgraph that preserves approximations of distances, and then broadcasting that entire spanner.

\begin{restatable}[Deterministic weighted APSP]{theorem}{firstWeightedAPSPApprox}
\label{thm:firstWeightedAPSPApprox}
    For any $\eps$ such that $0 < \eps \in O(1)$, there is a deterministic algorithm that given a weighted graph $G=(V,E,\omega)$, computes a $(1 + \eps \cdot \log{n})$-approximation for APSP in $\widetilde{O}(2^{1/\eps}\tn)$ rounds in the \hybridzero model.
\end{restatable}

We use \cref{thm:firstWeightedAPSPApprox} to show a result which is comparable to the best known deterministic approximation, by \cite{anagnostides2021deterministic}, deterministically achieving the same approximation ratio of $O(\log{n}/\log{\log{n}})$, but in $\tildeBigO{\tn}$ instead of $\tildeBigO{\sqrt{n}}$ rounds. By running \cref{thm:firstWeightedAPSPApprox} with $\eps=1/\log{\log{n}}$, we obtain the following corollary.

\begin{corollary}
    \label{cor:APSP_det_weigh}
    An $O(\log{n}/\log{\log{n}})$-approximation for weighted APSP can be obtained in \hybridzero in $\widetilde{O}\left(\tn\right)$ rounds deterministically.
\end{corollary}

Finally, we show the following result for approximating APSP. This runs slightly slower than $\tildeBigO{\tn}$ rounds, yet, shows a much better approximation ratio. For instance, for a $3$-approximation of weighted APSP, it achieves a round complexity of \smash{$\tildeBigO{n^{1/4}\tn^{1/2}} \subseteq \tildeBigO{\sqrt{n}}$}, which is, in general, much faster since $1 \leq \tn \leq \sqrt{n}$. This result is based on the skeleton graphs technique, first observed by \cite{Ullman1991}.

\begin{restatable}[Randomized weighted APSP]{theorem}{secondWeightedAPSPApprox}
\label{theorem:secondWeightedAPSPApprox}
    For any integer $\alpha \geq 1$, there is an algorithm that computes a $(4\alpha - 1)$-approximation for APSP in weighted graphs in $\tildeBigO{n^{1/(3\alpha + 1)} \left(\tn\right)^{2/(3 + 1/\alpha)} + \tn}$ rounds in the \hybridzero model \whp.
\end{restatable}

\paragraph{Approximating cuts.}
Our broadcast algorithm can also be used in combination with cut-sparsifiers by first computing such a sparsifier and then broadcasting it to the entire graph. This allows us to obtain approximation algorithms whose round complexity is $\tilO(\tk)$, where $k$ is the size of the sparsifier, plus the round complexity for constructing the sparsifier.

Similar to~\cite{anagnostides2021deterministic}, we leverage cut-sparsifiers~\cite{spielman2004nearly} and their efficient implementation in the \congest 
 model~\cite{koutis2016simple} to approximate the size of all cuts and solve several cut problems. Our algorithms run in $\tildeBigO{\tn/\eps + 1/\eps^2}$ rounds in \hybridzero, which is always at least as fast as all existing algorithms, and faster when $\tn \in o(\sqrt{n})$. The idea behind our result is to execute the \congest algorithm of~\cite{koutis2016simple}, in $\tildeBigO{1/\eps^2}$ rounds, to create a subgraph with $\tildeBigO{n/\eps^2}$ edges which approximate all cuts, and then we broadcast that subgraph in $\tildeBigO{\tk[n/\eps^2]} \subseteq \tildeBigO{\tn/\eps}$ rounds using~\cref{thm:optimal_dissemination}.

\begin{restatable}[Cuts]{theorem}{minCutApprox}
\label{theorem:minCutApprox}
    For any $\eps \in (0, 1)$, there is an algorithm in \hybridzero that runs in $\tildeBigO{\tn/\eps + 1/\eps^2}$ rounds \whp, after which each node can locally compute a $(1+\eps)$-approximation to all cut sizes in the graph. This provides $(1+\eps)$-approximations for many problems including minimum cut, minimum $s$-$t$ cut, sparsest cut, and maximum cut.
\end{restatable}

\paragraph{Shortest paths lower bounds.} We show lower bounds for shortest paths computations. The first lower bound works for unweighted graphs in \hybridzero, i.e., we assume that the set of identifiers for the source nodes is unknown. Therefore, to solve $k$-SSP, each node must learn the identifiers of these $k$ source nodes. Roughly, this can be seen as an application of \kdis, so the lower bound of $\tildeBigOmega{\tk}$ rounds discussed before applies.

\begin{restatable}[Unweighted $k$-SSP lower bound]{theorem}{unweightedLB}
    \label{thm:apsp_hybridzero_unweighted}
    The unweighted $k$-SSP problem with random source nodes takes $\tildeBigOmega{\tk}$ rounds in expectation in the \hybridzero model, even if all nodes initially know the input graph $G = (V,E)$. This holds for any approximation ratio and for randomized algorithms with constant success probability $p>0$.
\end{restatable}

This lower bound is almost tight for $k=n$ due to our various solutions for APSP in \cref{cor:APSP_sparse_graph}, \cref{thm:apsp_hybridzero_unweighted}, \cref{cor:APSP_det_weigh} in $\tilO(\NQ)$ rounds thus proving the universal optimality of the algorithms associated with Theorems \ref{thm:unweighted_apsp_approx}, \ref{thm:firstWeightedAPSPApprox} and \ref{theorem:secondWeightedAPSPApprox}.

Note that the lower bound of
\cref{thm:apsp_hybridzero_unweighted} works for unweighted graphs, but assumes that the choice of identifiers are unknown, which corresponds to the \hybridzero model. 
By contrast, the second lower bound applies for the weighted problem in the stronger \hybrid model, where the set of all identifiers is exactly $[n]$ and is known to all nodes.

\begin{restatable}[Weighted $(k,\ell)$-SP lower bound]{theorem}{weightedLB}
    \label{thm:klsp_LB}
    Solving the $(k,\ell)$-SP problem with arbitrary target nodes takes $\tildeBigOmega{\tk}$ rounds in expectation in the \hybrid model, which holds for randomized algorithms with constant success probability $p>0$, even for a polynomial stretch and even if all nodes initially know the \emph{weighted} local graph $G = (V,E,\omega)$. In the case of random target nodes, the same lower bound holds for algorithms with a probability of failure of less than \smash{$\frac{(1-p)\ell}{n}$} for any constant $0<p\leq1$.
\end{restatable}

Note that this result shows the universal optimality of results in \cref{thm:klsp_UB}, which has a much lower probability of failure even for $\ell = 1$. Moreover, this result holds for the stronger \hybrid model even if each node is given the \emph{weighted} graph and a separate set of identifiers of all source nodes as initial input. 

In this lower bound, we exploit that the mapping $\ID: V \to [n]$ from the source nodes in the graph to their identifiers is unknown (cf., \cref{sec:problems}). This allows us to find a node $v$ and two relatively large sets of nodes that are at a large \emph{hop} distance to $v$, where the nodes in one of the sets have a much larger \emph{weighted} distance to $v$ than those in the other. Sources are then randomly assigned to either one of the two sets.
By learning distances to the source nodes, $v$ also learns which set each source is located in. 

This is akin to learning the state of some random variable, thus we can apply a lower bound for the node communication problem, and we can prove a $\tildeBigOmega{\tk}$ lower bound.
One of the challenges is that for a universal lower bound, we cannot modify the topology of the given graph, so we have to rely on the graph property $\NQ$ and set edge weights to identify the node sets required in this proof.

The subsequent theorem shows that we can also give the nodes the exact location of the sources in the graph $G$ (by granting the mapping $\ID: V \to [n]$ and the set of identifiers of all sources as input) and instead withholding the weight function $\omega$. It is interesting to point out that if we would grant both the weighted graph and the location of source nodes as input, then the solution can computed locally. This implies that it is strictly necessary to initially withhold either the location of sources (\cref{thm:klsp_LB}) or edge weights (\cref{thm:klsp_LB_sources_known}) to obtain meaningful lower bounds.

\begin{restatable}[Weighted $(k,\ell)$-SP lower bound with known sources]{theorem}{weightedLBSourcesKnown}
    \label{thm:klsp_LB_sources_known}
    Solving the $(k,\ell)$-SP problem with arbitrary target nodes takes $\tildeBigOmega{\tk}$ rounds in expectation in the \hybrid model. This holds for randomized algorithms with constant success probability $p>0$ and even if all nodes initially know the unweighted local graph $G = (V,E)$ the set of source identifiers $\{\ID(s) \mid s \in S\}$ and the assignment of identifiers $\ID : V \to [n]$. In the case of random target nodes, the same lower bound holds for algorithms with a probability of failure of less than \smash{$\frac{(1-p)\ell}{n}$}, for any constant $0 < p\leq 1$.
\end{restatable}

Note that the lower bounds in \cref{thm:klsp_LB} and \ref{thm:klsp_LB_sources_known} for the $(k,\ell)$-SP problem are tight up to $\polylog n$ factors for the parameter ranges given in the corresponding upper bounds in \cref{thm:klsp_UB}. Finally, since APSP clearly solves the $(n,\ell)$-SP problem, \cref{thm:klsp_LB} and \ref{thm:klsp_LB_sources_known} also provide a lower bound for the APSP problem of $\tildeBigOmega{\tn}$ on weighted graphs that holds for the \hybrid model, i.e., where nodes initially \textit{know} the set of IDs. Therefore, Theorem \ref{thm:firstWeightedAPSPApprox} proves that the lower bounds in \cref{thm:klsp_LB} and \ref{thm:klsp_LB_sources_known} are almost tight for the case $k=n$.\phil{added this sentence to clarify the implications of the two above theorems omn APSP.}

\subsection{Existentially Optimal Shortest Paths}

We now discuss our new existentially optimal shortest paths algorithms, which are utilized as subroutines in our universally optimal algorithms and are of independent interest on their own.

\begin{restatable}[Existentially optimal SSSP]{theorem}{optSSSP}
\label{thm:almost_shortest_sssp}
A $(1 \p \eps)$-approximation of SSSP can be computed in $\tilO(1/\eps^2)$  rounds deterministically in the \hybridzero model.
\end{restatable}

This above result closes an existing gap for SSSP in \hybrid, as the fastest known solutions are $\tilO(n^{5/17})$ by \cite{censor2021sparsity} for a stretch of $1 \p \eps$, or $\tilO(n^{\eps})$ for a (large) constant stretch $(1/\eps)^{\bigOa(1/\eps)}$ due to \cite{augustine2020shortest}, or \smash{$\frac{\log n}{\log\log n}$} in \smash{$\tilO(n^{1/2})$} for a deterministic solution by \cite{anagnostides2021deterministic}. Note that our result is deterministic and holds even in \hybridzero, thus it beats the aforementioned results in all categories.

We give a brief overview of the proof.
It was shown in~\cite{rozhovn2022undirected} that a $(1+\eps)$-approximation of SSSP can be computed deterministically in the $\mathsf{PRAM}$ model with linear work and $\tilO(1)$ depth. 
Their main machinery is the simulation of an interface model, called the \minor model.
Specifically, they showed that a $(1 \p \eps)$-approximation of SSSP on $G$ can be computed with a total of $\tilO\big(1/\eps^2\big)$ rounds of \minor model and calls to the oracle \oracle. The oracle \oracle returns a solution for the so-called \euler problem (\cref{def:euler_oracle}) on a certain subgraph $H$ of the original graph (and a small number of certain virtual nodes).

We prove \cref{thm:almost_shortest_sssp} by giving efficient implementation of both  \minor and \oracle in the \hybridzero model. We stress that the main technical work lies in the provision of the oracle \oracle, i.e., the efficient solution of the \euler problem.
First, for the implementation of \minor we use an \emph{overlay network} construction~\cite{gmyr2017distributed}, we show that the \minor model can be implemented efficiently in the \hybridzero model. 

Second, Our strategy for implementing \oracle is to reduce the problem to the case where the arboricity of the graph is $\tilO(1)$ using the power of the local communication in the \hybridzero model (note that local communication is also strictly necessary for this simulation see \cref{rem:minor_aggregation_simulation}). Specifically, we compute a \emph{network decomposition}~\cite{rozhovn2020polylogarithmic} of the power graph $G^2$, and then we go through each color class one by one to try to orient as many edges as possible for the small-diameter clusters in the color class. In the end, we are left with a small-arboricity graph. 

Once we know that the graph has a small arboricity, we may assign each edge to a node using the \emph{forests decomposition} algorithm of Barenboim and Elkin~\cite{BE10} so that each node has a small load. After that, a desired orientation can be computed by first decomposing each node into degree-2 nodes and then orienting each cycle individually. We conclude that an Eulerian orientation of any graph can be computed in $\tilO(1)$ rounds deterministically in the \hybridzero model, which is a result that might be of independent interest.

We conclude the technical overview with the following theorem.

\begin{restatable}[Existentially optimal $k$-SSP]{theorem}{kssp}
	\label{thm:k-ssp}
	Given $\eps > 0$, $k$-SSP can be approximated w.h.p.
	\begin{itemize}
		\item in \hybrid in \smash{$\tilO\big(\!\sqrt{k}\cdot\tfrac{1}{\eps^2}\big)$} rounds with stretch $1 \p \eps$, if the sources are sampled with probability \smash{$\frac{k}{n}$},
		\item  in \hybridpar{\infty}{\gamma} in $\tilO\big(\!\sqrt{k/\gamma}\cdot\tfrac{1}{\eps^2}\big)$ rounds with stretch $3 \p \eps$, for $k$ arbitrary sources,
		\item  in \hybridpar{\infty}{\gamma} in $\tilO\big(\tfrac{1}{\eps^2}\big)$ rounds with stretch $1 \p \eps$, for $k \leq \gamma$ arbitrary sources.
	\end{itemize}
\end{restatable}

This result matches the lower bound of $\tilOm\big(\!\sqrt{k}\big)$ in the \hybrid model (due to \cite{Kuhn2020}) and the generalization $\tilOm\big(\!\sqrt{k/\gamma}\big)$ in the \hybridpar{\infty}{\gamma} model \cite{Schneider2023}, which hold for randomized algorithms with large stretch $\Omega(\!\sqrt{n}\big)$. This was not known up to our result, as all previous upper bounds were only able to match the lower bound for \smash{$k \geq n^{2/3}$} (see Figure \ref{fig:kssp_complexity_overview}). It is also interesting that the global capacity $\gamma$ does not only simply scale the running time, but can also be used to obtain near optimal solutions for up to $\gamma$ sources in just $\tilO(1)$ rounds.

Our technique is to schedule our SSSP algorithm from \cref{thm:almost_shortest_sssp} to achieve a round complexity for solving $k$-SSP that is better than simply repeating SSSP for $k$ times. To do so, we combine our SSSP algorithm with a framework to efficiently schedule multiple algorithms in parallel on a {skeleton graph}~\cite{Ullman1991} by utilizing {helper sets}~\cite{Kuhn2020}.\footnote{Note that these techniques are also used to prove other results in this paper. Our APSP algorithms use skeleton graphs. Our multi-message unicast algorithm uses an adaptive version of helper sets.}
Using the scheduling framework, we can run multiple instances of SSSP on a skeleton graph efficiently. The distance estimates in the skeleton graph can be turned into good distance estimates in the original graph, thereby allowing us to solve the $k$-SSP problem.

\section{The Graph Parameter Neighborhood Quality}
\label{sec:neigh_qual}

We start by presenting a fundamental graph parameter that describes the complexity of various types of problems in \hybrid and \hybridzero. Leaning on the nomenclature used by prior work on universally optimal distributed algorithms~\cite{Ghaffari2022,haeupler2021universally,haeupler2022hop},
we call our graph parameter the \emph{neighborhood quality} $\NQ$. Informally, $\NQ$ is the minimum distance $t$ such that the ball of radius $t$ around $v$ is sufficiently large to allow $v$ to exchange $\tildeBigOmega{k}$ bits of information, measured in terms of Shannon entropy~\cite{Shannon1948}, with other nodes in $O(t)$ rounds.
Intuitively, the higher quality the neighborhoods have, the smaller the parameter $\NQ$ is. 

Any distributed task can be trivially solved in diameter rounds using the unlimited-bandwidth local network even if nodes are required to learn a huge amount of information, so $\NQ$ has to be, in effect, upper bounded by $D$, which roughly means that neighborhoods play a bigger role on graphs with large diameter, on which global problems become interesting. 

For any node $v$, denote by $\calB_t(v)$ the ball of radius $t$ around $v$ -- that is, the set of all nodes that can reach $v$ with a path of at most $t$ edges, including $v$ itself.

\begin{definition}[Neighborhood Quality]
\label{def:bcast_quality}
\label{def:neigh_qual}
    Given a graph $G = (V, E)$, a number $k > 0$, and a node $v \in V$, we define
    \[
    \NQ(v) = \min \left(\left\{t : \size{\calB_t(v)} \geq k/t\right\} \cup{} \left\{D\right\}\right) \ \ \ \ \text{and} \ \ \ \  \tk(G) = \max_{u \in V} \tk(u).
    \]
\end{definition}

When $G$ is clear from context, we write $\tk$ instead of $\tk(G)$. Intuitively, if the diameter of the graph is not too small, then $\NQ$ strikes a balance between the radius and the size of the neighborhood of any node relative to $k$. This results directly from the definition and is reflected in the following useful observation that we will refer to occasionally.

\begin{observation}
	\label{lem:neigh_qual_min_neigh_size}
	Let $G=(V,E)$ be a graph and let $k \in[n]$. If $\NQ < D$, then $|\calB_{\NQ}(v)| \geq k/\NQ$ and equivalently $\NQ \geq k/|\calB_{\NQ}(v)|$ for any node $v \in V$.
\end{observation}

Observe that $\tk$ can be seen as a property of the set of \emph{power graphs} of $G$. A power graph $G^t$ of $G$ has the same node set as $G$ and has an edge $e = \{v, u\}$ if there is a path between $v$ and $u$ in $G$ with at most $t$ edges. The neighborhood quality $\tk$ can be alternatively defined as the minimum value $t$ such that the minimum degree in $G^t$ is at least $k/t$. Intuitively, this makes sense as in \hybrid or \hybridzero within $t$ rounds every node in $G$ can communicate with every neighbor it has in $G^t$. At the same time, in those $t$ rounds each of those nodes can receive $O(t \cdot \log n)$ messages through the global network. Thus the minimum degree of a node in $G^t$ dictates an upper bound on the amount of information a node can exchange in $t$ rounds by combining the global and local networks.

As a warm-up, we show that the parameter $\NQ$ can be computed deterministically in the \hybridzero model in $\tilO\left(\NQ\right)$ rounds. Consequently, for any algorithm design task whose goal is to attain the round complexity $\tilO\left(\NQ\right)$, we can freely assume that $\NQ$ is global knowledge and $\NQ(v)$ is known to each node $v$.

\begin{lemma}[Computing $\NQ$]
	\label{lem:compute_neigh_qual}
 Given an integer $k$, there is a deterministic $\tilO\left(\NQ\right)$-round algorithm that lets each node compute the precise values of $\NQ$ and $\NQ(v)$ in the \hybridzero model. 
\end{lemma}

\begin{proof}
    The high-level idea is to let each node locally compute $\NQ(v)$ and then all nodes can compute the maximum of those values to obtain $\NQ(G)$ using \cref{lem:hybrid_zero_aggregate}.  The first step works by each node exploring the local network to increasing depth $t=1,2,3, \ldots$ and locally computing $|\calB_t(v)|$ in parallel. An exploration up to depth $d$ costs $d$ rounds in the local communication network.
	To ensure that we do not look too deep, i.e., beyond $\NQ$, after each step of exploration, we compute \smash{$N_t := \min_{v \in V} |\calB_t(v)|$} in $\tilO(1)$ rounds using \cref{lem:hybrid_zero_aggregate}. We stop once we reach a value $t'$ with $N_{t'} \geq k/t'$, in which case we have found $\NQ := t'$ by \cref{def:neigh_qual}. If the stopping condition is not met even after the entire network is explored, then all nodes know that $\NQ = D$. In any case,  	
	the overall round complexity is $\tilO\left(\NQ\right)$.
\end{proof}

\subsection{Graph Clustering Depending on \texorpdfstring{\NQ}{Neighborhood Quality}}


An important tool that we require frequently, is a clustering with guarantees pertaining to the weak diameter and the size expressed in terms of \NQ. We first need the following definition.

\begin{definition}
\label{def:ruling_sets}
    An $(\alpha,\beta)$-ruling set for $G=(V,E)$ is a subset $W\subseteq{}V$, such that for every $v\in{}V$ there is a $w\in{}W$ with $\hop(v,w)\leq{}\beta$ and for any $w_1,w_2\in{}W$, $w_1\neq{}w_2$, we have $\hop(w_1,w_2)\geq{}\alpha$.
\end{definition}


We use ruling sets to design a clustering algorithm that partitions the node set $V$ into clusters with weak diameter $\tildeBigO{\tk}$ and size $\Theta(k/\tk)$ in $\tildeBigO{\tk}$ rounds deterministically. This clustering algorithm plays a crucial role in our broadcast and unicast algorithms.

\begin{lemma}[Clustering]
\label{lem:clustering_with_weak_diam}
    For any integer $k$, there is a deterministic $\tildeBigO{\tk}$-round algorithm that partitions the node set $V$ into clusters satisfying the following requirements in the \hybridzero model.
    \begin{itemize}
        \item The weak diameter of each cluster $C$ is at most $4\tk\ceil{\log{n}}$.
        \item The size of each cluster $C$ is at least $k/\tk$ and at most $2k/\tk$.
        \item Each cluster $C$ has a leader $r(C)$.
        \item Let $R$ be the set of all cluster leaders. Each node knows whether it is in $R$ and knows to which cluster it belongs.
    \end{itemize}
\end{lemma}
\begin{proof}
In this proof, we utilize the following result from~\cite{Kuhn2018a}. Let $\mu$ be any positive integer. A $(\mu+1, \mu\ceil{\log{n}})$-ruling set can be computed deterministically in $O(\mu\log{n})$ rounds in \congest~\cite[Theorem 1.1]{Kuhn2018a}, so the algorithm also applies to \hybridzero. We remark that before the work~\cite{Kuhn2018a}, it was already known that an $(\alpha, \beta)$-ruling set with $\beta \in  O(\alpha \log n)$ can be computed in $O(\alpha \log n)$ rounds deterministically in the \local model~\cite{Awerbuch1989}.

    We compute $\tk$ in $\tildeBigO{\tk}$ rounds by \cref{lem:compute_neigh_qual}. We choose $\mu=2\tk$ and use the algorithm of~\cite{Kuhn2018a} to compute a $(2\tk+1, 2\tk\ceil{\log{n}})$-ruling set in $O(\tk\log{n})$ rounds
    We denote the set of rulers by $R$ -- i.e., $R$ is the set of nodes in the ruling set. For $2\tk\ceil{\log{n}}$ rounds, each node learns its $2\tk\ceil{\log{n}}$ neighborhood and the rulers in it, through the local network. For every $v\in{V}$, let $r(v)$ be the closest ruler by hop distance, with ties broken by minimum identifier. By \cref{def:ruling_sets}, $r(v)$ must be in its $2\tk\ceil{\log{n}}$ neighborhood. By exploring this neighborhood, each node $v$ finds $r(v)$.


Every node $v \in V\setminus R$ joins the cluster of its closest ruler $r(v)$. Observe that any cluster $C$ contains exactly one ruler, so we set this ruler as the cluster leader $r(C)$ and set the cluster identifier of $C$ as the identifier of $r(C)$. \cref{def:ruling_sets} guarantees that the strong diameter of each such cluster is at most $2\beta=4\tk\ceil{\log{n}}$. Thus, for $4\tk\ceil{\log{n}}$ rounds, each node $v$ floods $r(v)$ through the local network, so for every cluster $C$, any $v\in{}C$ knows all the nodes in $C$. 

Let $C$ be a cluster. As for every $r_1\neq{}r_2\in{}R$, $\hop(r_1,r_2)\geq{}\alpha=2\tk+1$, it holds that $\calB_{\tk}(r(C))\subseteq{}C$ -- that is, every node in $u \in \calB_{\tk}(r(C))$ joins $C$, as the closest ruler to $u$ is $r(C)$. By \cref{lem:neigh_qual_min_neigh_size}, $|C|\geq{}|\calB_{\tk}(r(C))|\geq{}k/\tk$. Thus, every cluster has a minimum size of $k/\tk$.

Now, to make sure that our clusters are not too big, each cluster $C$ with $|C|>2k/\tk$ splits deterministically into multiple clusters, until each cluster $C'$ satisfies the condition $k/\tk\leq{}|C'|\leq{}2k/\tk$. This can be computed locally for each cluster without communication.
After this process, we get at most $n\tk/k$ disjoint clusters, each with weak diameter at most $4\tk\ceil{\log{n}}$, and of size $k/\tk\leq{}|C|\leq{}2k/\tk$. The diameter guarantee is weak and not strong due to splitting.
\end{proof}

\subsection{Properties of \texorpdfstring{$\tk$}{Neighborhood Quality}}

Next, we analyze the properties of $\tk$ to compare it to other graph parameters, such as the number of nodes in a graph and its diameter. This allows us to compare our results to existing previous works.

\begin{restatable}{lemma}{boundingTk}
\label{lem:boundingTk}
$\sqrt{\frac{Dk}{3n}} < \tk\leq{}\min{}\left\{D,\sqrt{k}\right\}$.
\end{restatable}
\begin{proof}
    From \cref{lem:neigh_qual_min_neigh_size}, we know that for all $v\in{}V$ it holds that $\left|\calB_{\tk}(v)\right|\geq{}k/\tk$, so $\frac{n}{\left|\calB_{\tk}(v)\right|}\leq{}\frac{n\tk}{k}$. Thus, there can be at most $\frac{n\tk}{k}$ disjoint balls of radius $\tk$. Consider a shortest path $P$ whose length equals the diameter $D$ of the graph $G$. We can find $\left\lceil(D+1) / (2\tk + 1)\right\rceil$ nodes in $P$ whose balls of radius $\tk$ are disjoint. Specifically, we can just pick the first node $s$ in $P$ and all nodes $v$ in $P$ such that $\dist(s,v)+1$ is an integer multiple of $2\tk + 1$.
    Since we always have $\tk \geq 1$, we infer that $\left\lceil(D+1) / (2\tk + 1)\right\rceil \leq \frac{n\tk}{k}$, which implies that   
    $D+1\leq{}\frac{n\tk}{k}\cdot{}(2\tk+1)\leq{} \frac{3n\tk^2}{k}$. Therefore, we obtain the first inequality $\tk > \sqrt{\frac{Dk}{3n}}$.

    For the second inequality, $\tk\leq{}D$ automatically holds by \cref{def:bcast_quality}. For the rest of the proof, we assume that $k < D$ and let $t = \sqrt{k}$,    
    so $\left|\calB_{t}(v)\right|\geq t = k/t$ for any node $v\in{V}$. 
    Therefore, by \cref{def:bcast_quality} $\tk(v)\leq{} t = \sqrt{k}$ for any node $v\in{V}$, so $\tk\leq{} t = \sqrt{k}$.
\end{proof}

	

\cref{lem:boundingTk} offers the first indication that $\NQ$ is a suitable parameter to describe a universal bound for various problems with parameter $k$, including broadcasting $k$ messages and computing distances to $k$ sources, is obtained by relating $\NQ$ to the previous \emph{existential} lower bound \smash{$\tilOm\big(\!\sqrt k\big)$} of for these problems~\cite{Kuhn2020, augustine2020shortest, CensorHillel2021}. 
\cref{lem:boundingTk} implies that the neighborhood quality $\NQ \in O(\sqrt{k})$ is always at most such an existential lower bound.
As we will later see, the bound is tight in that we have $\NQ \in \Theta(\sqrt{k})$ when $G$ is a path and $k \in O(n^2)$. 
Indeed, graphs that feature an isolated long path have frequently been used to obtain existential lower bounds for shortest paths problems in the \hybrid model~\cite{augustine2020shortest,Kuhn2020, Kuhn2022}.

Next, we show a statement that limits the rate of growth of $\tk$ as $k$ grows. The statement says that for $k'=\alpha{}k$, the value $\tk[k']$ can only be larger than $\tk$ by a factor of $O(\sqrt{\alpha})$. The idea behind the proof is that for any graph, all neighborhoods of radius $\tk$ can learn $\tildeBigO{k}$ messages in $\tk$ rounds. Therefore, if we increase $\tk$ by a factor of $O(\sqrt{\alpha})$, then all neighborhoods of radius $O(\sqrt{\alpha}) \cdot{}\tk$ can learn $\tildeBigO{k \cdot \alpha}$ messages in $O(\sqrt{\alpha}) \cdot{}\tk$ rounds, as the increase in both round complexity and radius contributes to the increase in bandwidth.

\begin{restatable}{lemma}{growthOfTk}
\label{lem:growth_of_tk}
For $\alpha \geq 1$, 
$\tk[\alpha{}k]\leq{}6\sqrt{\alpha}\cdot{}\tk$.
\end{restatable}
\begin{proof}
    If $6 \sqrt{\alpha} \cdot \tk \geq D$, then by definition $\tk[\alpha{}k] \leq D \leq 6 \sqrt{\alpha} \cdot \tk$. Thus, assume that $6 \sqrt{\alpha} \cdot \tk < D$. Let $v$ be any node. There must exist a node $u$ such that the hop distance between $v$ and $u$ is at least $D/2 > 3 \sqrt{\alpha} \cdot \tk$, since otherwise the diameter of the graph is less than $D$. Denote by $P = (v = v_1, v_2 \ldots, v_{\ell+1} = u)$ a shortest path between $v$ and $u$ of length $\ell \geq D/2 > 3 \sqrt{\alpha} \cdot \tk$. 
    
    Observe that for any two nodes on the path that are at least $3\tk$ edges apart on the path, their $\tk$-hop balls are disjoint. That is, we claim that for any two nodes $v_i$ and $v_{j}$ such that $j - i \geq 3\tk$, it holds that $\calB_{\tk}(v_i) \cap \calB_{\tk}(v_{j}) = \emptyset$. Assume for the sake of contradiction that $\calB_{\tk}(v_i) \cap \calB_{\tk}(v_{j}) \neq \emptyset$ and take $w \in \calB_{\tk}(v_i) \cap \calB_{\tk}(v_{j})$. It holds that $\hop(v_i, w) \leq \tk$ and $\hop(w, v_{j}) \leq \tk$, and thus $\hop(v_i, v_{j}) \leq 2\tk$. However, as $P$ is a shortest path, $\hop(v_i, v_{j}) = j-i \geq 3\tk > 2\tk$, where the last inequality holds since $\tk \geq 1$, and so we arrive at a contradiction.
    Therefore, we obtain that 
    \begin{align*}
        |\calB_{3 \sqrt{\alpha} \cdot \tk}(v)| &\geq \sum_{i \in [\sqrt{\alpha}]} |\calB_{\tk}(v_{1 + i \cdot 3\tk})| 
        \geq \sqrt{\alpha}\cdot{}\frac{k}{\tk} 
        > \frac{\alpha{}k}{3\sqrt{\alpha}\tk}.
    \end{align*}
    Thus, by \cref{def:neigh_qual}, $\tk[\alpha{}k] \leq 3\sqrt{\alpha} \cdot{} \tk < 6\sqrt{\alpha} \cdot{} \tk$, so we are done. 
\end{proof}

A useful property for our lower bounds is that we can always find a node that has a small neighborhood.\phil{I'm also putting this here since I use it at least twice in lower bounds proofs}

\begin{lemma}
    \label{lem:neighborhood_small}
    There exists a node $v \in V$ such that $|\calB_{r}(v)| < k/r$ for any $r < \NQ$.
\end{lemma}

\begin{proof}
    Recall \cref{def:neigh_qual}, which stipulates that $\NQ(G) = \max_{u \in V} \NQ(u)$. Let $v$ be the node that maximizes $\NQ(v)$, i.e., $\NQ(v) = \NQ(G)$.    
    
    The definition $\NQ(v) = \min \left(\left\{t : \size{\calB_t(v)} \geq k/t\right\} \cup{} \left\{D\right\}\right)$ implies that $|\calB_{r}(v)| < k/r$, since $|\calB_{r}(v)| \geq k/r$ would be a contradiction.
\end{proof}

\subsection{Special Graph Families}
The previous algorithms for weaker variants of \kdis all take $\tildeBigO{\sqrt{k}}$ rounds~\cite{augustine2020shortest, anagnostides2021deterministic},  so \cref{lem:boundingTk} implies that our algorithm for \kdis, which costs $\tildeBigO{\tk}$ rounds, is never slower than the previous works and supports a wider variety of cases. 
Conversely, we analyze certain graphs of families where $\tk \in o(\sqrt{k})$, to show that many such graphs exist, beyond just graphs with $D \in o(\sqrt{k})$. We proceed with estimating the value of $\tk$ for paths, cycles, and any $d$-dimensional grid graphs. 

\begin{definition}[Grid graphs]
    A $d$-dimensional grid graph $G=(V,E)$ with $n=m^d$ nodes is the $d$-fold Cartesian product $G= \underbrace{P_m \times \dots \times P_m}_{d \; \text{times}}$ of the $m$-node path $P_m$. 
\end{definition}

The calculation details for the proofs for the following results are deferred to \cref{sec:graph_families}.

\begin{restatable}[$\tk$ in paths and cycles]{theorem}{pathsTk}
\label{theorem:path_graphs_tk}
    For paths and cycles,  $\tk \in \min\left\{\Theta(\sqrt{k}), D\right\}$.
\end{restatable}

\begin{restatable}[$\tk$ in grids]{theorem}{gridGraphsTk}
\label{theorem:grid_graphs_tk}
    For $d$-dimensional grids with $d \in O(1)$, we have
    \[\tk  \in \min\left\{\Theta\left(k^{1/(d+1)}\right), D\right\}.\]
\end{restatable}

Since $D \in \Theta\left(n^{1/d}\right)$, we have $\tk \in \Theta\left(k^{1/(d+1)}\right)$ whenever $k \in O\left(n^{1+1/d}\right)$ in \cref{theorem:grid_graphs_tk}.
This shows a few interesting points. First, when $k = O\left(n^{1+1/d}\right)$ is not too big, we can broadcast $k$ messages in a much better round complexity $\tildeBigO{\tk} \subseteq \tildeBigO{k^{1/(d+1)}}$ than all the existing algorithms in previous works, which cost at least $\tildeBigO{\sqrt{k}}$ rounds. However, once $k$ crosses $\Omega\left(n^{1+1/d}\right)$, we cannot do much better than naively sending all messages via the local network in $O(D)$ rounds.

For any $d \in O(1)$, grid graphs have $O(dn) \subseteq O(n)$ edges, and thus in $\tildeBigO{\tn} = O\left(n^{1/(d+1)}\right)$ rounds, all the nodes can learn the entire graph topology using \cref{thm:optimal_dissemination}, and then they can locally compute APSP exactly. For any $d\geq 2$, this is polynomially faster than the existentially optimal algorithms of~\cite{augustine2020shortest, Kuhn2020, anagnostides2021deterministic} as well as the trivial $O(D)$-round algorithm. 

More generally, we have the following fact.

\begin{theorem}\label{thm:growth}
    Let $G = (V,E)$ be a graph satisfying $|\calB_r(v)| \in \Omega\left(r^d\right)$ for all $v \in V$ and $r \leq D$. We have $D \in O(n^{1/d})$ and $\tk \in \min\left\{D, O(k^{1/(d+1)})\right\}$.
\end{theorem}
\begin{proof}
The diameter bound $D \in O(n^{1/d})$ follows from the fact we may find a number $r' \in O(n^{1/d})$ such that $\calB_{r'}(v) \in \min\left\{n, \Omega(r^d)\right\} = n$ for all $v \in V$, so $D \leq r'$. For the neighborhood quality bound, if we assume that $\tk < D$, then  
 $\tk$ is the smallest number $t$ such that $\size{\calB_t(v)} \geq k/t$. The condition can be rewritten as $t \in O(k^{1/(d+1)})$, as  $|\calB_r(v)| \in \Omega\left(r^d\right)$ for all $v \in V$ and $r \leq D$. Therefore, indeed $\tk \in \min\left\{D, O(k^{1/(d+1)})\right\}$.
\end{proof}

Consequently, for the graph classes considered in \cref{thm:growth}, an $\tildeBigO{\tn}$-round algorithm offers a \emph{polynomial} advantage over the trivial diameter-round algorithm.

\section{Universally Optimal Multi-Message Broadcast}
\label{sec:optimal_broadcast}

In this section, we show our universally optimal broadcasting result. Using the analyses of $\tk$ and the tools developed in \cref{{sec:neigh_qual}}, we show universally optimal algorithms for \kdis and \kagg. 

\subsection{Basic Tools}

We start with introducing some basic tools.

\begin{lemma}[Uniform load balancing]
\label{lemma:load_balancing}
    Given a set of nodes $C$ with weak diameter $d$ and a set of messages $M$ with $|M|$ distributed across $C$, there is an algorithm that when it terminates, each $v\in{}C$ holds at most $\ceil{|M|/|C|}$ messages. The algorithm costs $2d=O(d)$ rounds deterministically in \hybridzero. 
\end{lemma}
\begin{proof}
    In $d$ rounds, all nodes flood the messages and identifiers of $C$. The minimum identifier node then computes an allocation such that each $v\in{}C$ is responsible for at most $\ceil{k/|C|}$ messages and floods the allocation for another $d$ rounds, so it reaches all $v\in{}C$.
\end{proof}

For any model with unlimited-bandwidth local communication, we use the following term throughout the paper, which we define formally.

\begin{definition} [Flooding]
\label{def:flooding}
    Flooding information through the local network is sending that information through all incident local edges of all nodes. In subsequent rounds, the nodes collect the messages they receive and continue to send them as well. After $t$ rounds, every node $v$ knows all of the information that was held by any node in its $t$-hop neighborhood before the flooding began.
\end{definition}

\paragraph{Overlay networks.} A basic tool in designing algorithms in the \hybrid model and its variants is to construct a \emph{overlay network}, which is a virtual graph on the same node set in the original graph $G$ that has some nice properties that can facilitate communication between nodes that are far away in $G$ using limited-bandwidth global communication. For example, it was shown in~\cite{Augustine2019} that a \emph{butterfly} network, which is a small-diameter constant-degree graph, can be constructed efficiently in \NCC. 
It was shown in~\cite[Theorem 2]{gmyr2017distributed} 
that given a connected graph $G$ of $n$ nodes and polylogarithmic maximum degree, there is an $O(\log^2{n})$-round deterministic algorithm that constructs a constant degree virtual tree $T$ of depth $O(\log{n})$ that contains all nodes of $G$ and is rooted at the node with the highest identifier in \hybridzero. Each node in $G$ knows the identifiers of its parent and children in $T$.
The round complexity can be further improved to $O(\log n)$ if randomness is allowed. Specifically, for any graph $G$, it was shown in~\cite{Goette2021} that in \ncczero, there is an $O(\log n)$-round algorithm that constructs a constant degree virtual rooted tree $T$ of depth $O(\log{n})$ w.h.p.
Both algorithms are not directly applicable to our setting, as we need a deterministic virtual tree construction that works for any graph $G$. We show how to adapt the deterministic overlay network construction of \cite{gmyr2017distributed} to this setting.

\begin{restatable}[Virtual tree]{lemma}{treeWithoutIDs}
\label{lem:construct_tree_without_ids}
In the \hybridzero model, there is a deterministic polylogarithmic-round algorithm that constructs a constant degree virtual rooted tree $T$ of depth $O(\log{n})$ over the set of all nodes in $G$. By the end of the algorithm, each node in $G$ knows the identifiers of its parent and children in $T$.
\end{restatable}
\begin{proof}
To adapt the virtual tree construction of \cite{gmyr2017distributed}, we just need to show that a virtual graph $H$ with a polylogarithmic maximum degree over the set of all nodes $V$ can be constructed in polylogarithmic rounds deterministically in such a way that each round of \LOCAL in $H$ can be simulated in polylogarithmic rounds in $G$.
Given $H$, we may apply \cite{gmyr2017distributed} to $H$ to obtain the desired tree $T$. 
To construct such a graph $H$, we use a \emph{sparse neighborhood cover}, which is a clustering of the node set of the original graph $G$ into overlapping clusters of polylogarithmic diameter such that for each node in $G$, its 1-hop neighborhood $\calB_1(v)$ is entirely contained in at least one of the clusters, and each node is in at
most polylogarithmic number of clusters. It was shown in~\cite[Corollary 3.15]{rozhovn2020polylogarithmic} that such a sparse neighborhood cover can be constructed deterministically in polylogarithmic rounds in the \local model. Given a sparse neighborhood cover, a desired overlay network $H$ can be obtained by letting each cluster compute a constant-degree virtual tree and then taking the union of all these trees. Since any two adjacent nodes in $H$ are within the same polylogarithmic-diameter, indeed each round of \LOCAL in $H$ can be simulated using polylogarithmic rounds in $G$.
\end{proof}


Using \cref{lem:construct_tree_without_ids}, we immediately obtain the following solution for \kagg for the special case of $k=1$. As we will later see, our \kdis and \kagg algorithms for general $k$ utilize this result.

\begin{lemma}[Basic deterministic aggregation]
\label{lem:hybrid_zero_aggregate}
    For $k = 1$, it is possible to solve \kagg and \kdis deterministically in polylogarithmic rounds in \hybridzero.
\end{lemma}
\begin{proof}
Construct a virtual tree $T$ using \cref{lem:construct_tree_without_ids}. By a BFS from the root, all nodes can learn its level in $T$. To solve \kagg, by reversing the communication pattern of BFS, the root of $T$ learns the solution for the aggregation task. By a BFS from the root again, all nodes learn the solution. The \kdis problem can be solved similarly.
\end{proof} 

For any $k \in \tilO(1)$,  \cref{lem:hybrid_zero_aggregate} implies that \kdis can also be solved in $\tilO(1)$ rounds by simply running $k \in \tilO(1)$ broadcast instances in parallel and using the unique identifiers of the initial senders to distinguish between different instances.

We need the following helper lemma on pruning virtual trees.

\begin{lemma}[Pruning]
\label{thm:ncc_tree_pruning}
    Let $T=(V,E_T)$ be a virtual tree rooted at $r$, maximum degree $c \in O(1)$ and depth $d$. Given some function $f: V \rightarrow \{0, 1\}$ such that $f(v)$ is known to $v$ for all $v \in V$, there is an algorithm that constructs a virtual tree $T'=(U,E_{T'})$, where $U=\{v \in V \mid f(v) = 1\}\subseteq{}V$, with maximum degree $c' \in O(cd)$ and depth $d' \leq d$. The construction takes $O(d^2)$ rounds deterministically in \hybridzero. 
\end{lemma}
\begin{proof}
    Observe that every $v\in V$ knows whether $v\in U$. For every $v \in V$, denote by $T(v)$ the subtree of $T$ rooted at $v$. In the first step, we let every $v \in V$ compute $|U \cap T(v)|$ by each node sending up the tree how many nodes in $U$ are in its subtree. This takes $O(d)$ rounds.

    We design a recursive algorithm $\mathcal{A}$ to construct the desired tree $T'$. 
    For each node $v \in V$, we write $\mathcal{A}(v)$ to denote the algorithm $\mathcal{A}$ applied to the subtree $T(v)$. Given a node $v$ in $T$, the algorithm $\mathcal{A}(v)$ works as follows.
    \begin{itemize}
        \item If $|U \cap T(v)| = 0$, then $v$ removes itself from the tree and notifies all its children $u$ to run $\mathcal{A}(u)$.
        \item If $v \in U$, then $v$ notifies all its children $u$ to run $\mathcal{A}(u)$.
        \item If $|U \cap T(v)| \neq 0$ and $v \notin U$, then $v$ performs a walk down the tree, each time choosing to go to a child that has some nodes of $U$ in its subtree until the walk reaches a node $u^\star$ that belongs to $U$. Let $P = (v, \ldots, u^\star)$ denote this walk. Contracting the path $P$ into $u^\star$ by removing all nodes in $P$ except for $u^\star$ and resetting the parent of all the children of the removed nodes to be $u^\star$. After the contraction,  $u^\star$  notifies all its children $u$ to run $\mathcal{A}(u)$.
    \end{itemize}

Each recursive call costs $O(d)$ rounds.
    Since all recursive calls are done in parallel, the algorithm $\mathcal{A}(r)$ takes $O(d^2)$ rounds.
    For the rest of the proof, we analyze the tree $T'$ constructed by the algorithm $\mathcal{A}(r)$. From the algorithm description, the node set of $T'$ is exactly $U$.
    In the algorithm, the only way that a node $u$ can gain more children is when $u = u^\star$ for some recursive call, in which case the degree of $u$ becomes $O(cd)$. Since this can happen at most once, we infer that the maximum degree of $T'$ is at most $c' \in O(cd)$. From the description of the algorithm, the depth of the current tree never increase, so the depth of $T'$ is at most $d' = d$.
\end{proof}

Combining \cref{lem:construct_tree_without_ids,thm:ncc_tree_pruning}, we obtain the following result.

\begin{lemma}[Virtual tree on a subset]
\label{lem:construct_tree_without_ids2}
Given some function $f: V \rightarrow \{0, 1\}$ such that $f(v)$ is known to $v$ for all $v \in V$, there is a deterministic polylogarithmic-round algorithm in \hybridzero that constructs a virtual rooted tree $T$ of maximum degree $O(\log n)$ and depth $O(\log{n})$ over the set $U=\{v \in V \mid f(v) = 1\}$. By the end of the algorithm, each node in $T$ knows the identifiers of its parent and children in $T$.
\end{lemma}
\begin{proof}
 We first create a virtual tree $T_1$ over $V$, with  maximum degree $c \in O(1)$ and depth $d \in O(\log{n})$ using \cref{lem:construct_tree_without_ids}.
 After that, we apply the pruning algorithm of \cref{lem:construct_tree_without_ids2} to $T_1$ and $f$ to obtain a virtual tree $T_2$ over $U$, with maximum degree $c' \in O(cd) \subseteq O(\log n)$ and depth $d' \leq d \in O(\log{n})$, satisfying all the requirements.
\end{proof}

\begin{figure}[ht!]
  \centering
\includegraphics[width=0.5\textwidth]{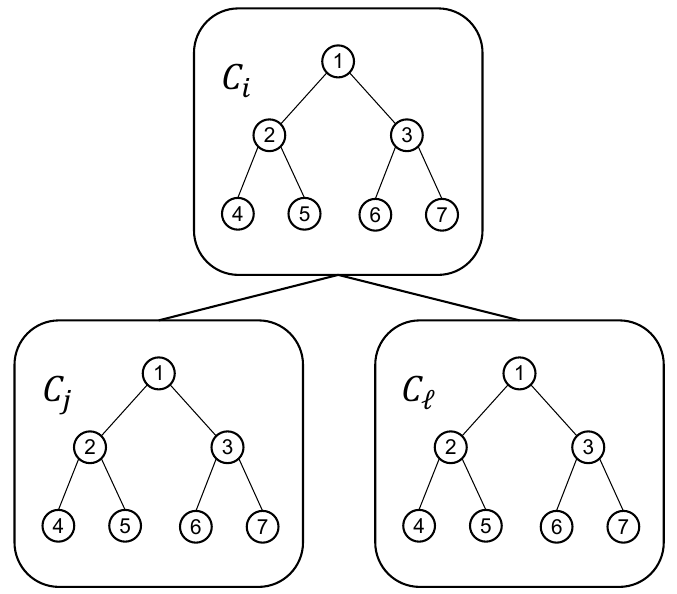}
  \caption{Overview of the proof of \cref{thm:optimal_dissemination}. We first create clusters with $\tildeBigO{\tk}$ weak diameter and roughly the same number of nodes, and then we construct a logical tree of the clusters, with logarithmic maximum degree and depth. Inside each cluster, we create a logical binary tree over the nodes of the cluster. As the clusters are roughly of the same size, we can ensure that the trees inside the clusters have the exact same shape. We ensure that for any two neighboring clusters in the cluster tree, the nodes in their internal trees know their respective nodes in the tree of the other cluster -- e.g.,~node $3$ in cluster $C_i$ knows the identifier of nodes $3$ in $C_j, C_\ell$ and can communicate directly with them using global communication. Once we are done constructing all of these trees, we propagate all the $k$ messages in the graph up to the top of the cluster tree, and then we propagate them back down to all the clusters to ensure that every node in the graph receives all the messages.} 
  \label{fig:notations}
\end{figure}

\subsection{\texorpdfstring{Solving \kdis and \kagg}{Solving k-dissemination and k-aggregation}}

We are now ready to prove \cref{thm:optimal_dissemination}.
\kdisUP*
\begin{proof}
\label{optimal_dissem_proof}
The algorithm consists of several phases: clustering, cluster-chaining, load balancing, and dissemination. See \cref{fig:notations} for an overview of our algorithm. The clustering phase ensures that we partition the nodes to disjoint clusters of similar size such that the weak diameter of each cluster is small. 
In the cluster-chaining phase, we order the clusters in a logical tree with logarithmic maximum degree and depth, and let the nodes of each cluster know the nodes of its parent and children clusters. 
In the dissemination phase, we trickle all the tokens up to the root cluster using the global network, and the chaining we devised in the cluster-chaining phase. Then, we trickle the tokens down the tree, such that each cluster learns all the tokens. 

We begin by computing $k$ by summing how many tokens each node holds using \cref{lem:hybrid_zero_aggregate} in $\tildeBigO{1}$ rounds, and then we compute $\tk$ in $\tildeBigO{\tk}$ rounds by \cref{lem:compute_neigh_qual}.

\paragraph{Clustering.} We wish to create a partition of the nodes in the graph into clusters of roughly the same size and with small weak diameter. We execute \cref{lem:clustering_with_weak_diam} with parameter $k$, which takes $\tildeBigO{\tk}$ rounds. In the subsequent discussion, let $R$ be the set of cluster leaders for the clustering computed by \cref{lem:clustering_with_weak_diam}.

We create a logical binary tree $\mathbb{T}_i$ of each cluster $C_i$, with   depth  $O(\log n)$, and with $r(C_i)$ being the root of $\mathbb{T}_i$.
Within each cluster $C_i$, we rank its nodes according to the in-order traversal of the logical binary tree $\mathbb{T}_i$. 
We desire for $\mathbb{T}_i$ to have exactly $2k/\tk$ nodes and have the same structure for all clusters $C_i$ in the clustering. As $k/\tk\leq{}|C_i|\leq{}2k/\tk$, this can be achieved by letting some of the nodes in $C_i$ to simulate two nodes in $\mathbb{T}_i$.

\paragraph{Cluster chaining.} This phase consists of two sub-phases. We first create a logical tree $\mathbb{T}_\diamond$ over the set of all clusters, with logarithmic maximum degree and depth. After that, for any two clusters $C_i$ and $C_j$ connected by an edge in $\mathbb{T}_\diamond$, we create a bijective mapping between the nodes in the two clusters with help from the 
 internal trees  $\mathbb{T}_i$ and $\mathbb{T}_j$.

\paragraph{Subphase 1: Building the cluster tree.}
In the first sub-phase, we build a virtual tree $\mathbb{T}_\diamond$ connecting all clusters, as follows.
Recall that each node knows whether it is a cluster leader or not. Thus, we define a function $f$ where every $v \in V$ sets $f(v) = 1$ if $v \in R$, and $f(v)=0$ otherwise. We now use \cref{lem:construct_tree_without_ids2} with $f$ to compute a virtual tree $\mathbb{T}_\diamond$, with maximum degree  $O(\log n)$ and depth $O(\log n)$, over the set of all cluster leaders. This step takes $\tildeBigO{1}$ rounds. From now on, we call $\mathbb{T}_\diamond$ the cluster tree.

\paragraph{Subphase 2: Matching parent and children cluster nodes.} 
Let $C_i$ and $C_j$ be two clusters whose leaders $r(C_i)$ and $r(C_j)$ are adjacent in the cluster tree $\mathbb{T}_\diamond$ such that $r(C_i)$ is the parent of $r(C_j)$ in $\mathbb{T}_\diamond$. As discussed earlier, it holds that $\mathbb{T}_i$ and $\mathbb{T}_j$ have the same structure. Let $v_i \in \mathbb{T}_i$  and $v_j \in \mathbb{T}_j$ be two nodes with the same rank in their trees. We now desire for $v_i$ and $v_j$ to be made aware of each other -- that is, we want them to learn the identifiers of each other so that they can communicate over the global network.

We begin with $r(C_i)$ and $r(C_j)$, who are at the root of $\mathbb{T}_i$ and $\mathbb{T}_j$, respectively. They already know the identifiers of each other, as that is guaranteed by the construction of $\mathbb{T}_\diamond$. Let $L_i$ and $R_i$ be the left and the right children of $r(C_i)$, respectively.  Let $L_j$ and $R_j$ be the left and the right children of $r(C_j)$, respectively. Node $r(C_i)$ sends to $r(C_j)$ the identifiers $L_i$ and $R_i$, and $r(C_j)$ sends the identifiers $L_j$ and $R_j$ to $r(C_i)$. Now, $r(C_i)$ sends to $L_i$ and $R_i$ the identifiers of $L_j$ and $R_j$, and likewise $r(C_j)$ communicates with $L_j$ and $R_j$. All of this takes $O(1)$ rounds using the global network, and can be done in parallel for all edges in the cluster tree $\mathbb{T}_\diamond$. 

Observe that now $L_i$ and $L_j$ know both their identifiers, and likewise $R_i$ and $R_j$. Thus, both of them can continue down their respective subtrees by applying the same algorithm, in parallel. As $\mathbb{T}_i$ and $\mathbb{T}_j$ have $O(\log n)$ depth, and each level of the trees takes $O(1)$ rounds to process, the process takes $O(\log n)$ rounds in total.

\paragraph{Load balancing.} Each cluster $C_i$ uniformly distributes the tokens of its nodes within itself by \cref{lemma:load_balancing}. There are $k$ tokens in the graph, and so at most $k$ tokens are in $C_i$. By \cref{lem:clustering_with_weak_diam},  $|C_i|\geq{}k/\tk$, so $C_i$ can load balance its tokens such that each node has at most $\tk$ tokens. In total, this phase takes $\tildeBigO{\tk}$ rounds, because the weak diameter of $C_i$ is at most $4\log{n}\tk$.

\paragraph{Dissemination.} We now aim to gather all the tokens in the root cluster $C_r$ of the cluster tree $\mathbb{T}_\diamond$. 
For $O(\log n)$ iterations, we send the tokens of each cluster up the cluster tree. Due to the load balancing step, each node holds at most $\tk$ tokens, so in $O(\tk)$ rounds we can send the tokens up by one level in the cluster tree  using the global communication network. In the beginning of each iteration, each cluster again load balances the tokens it received in the previous iteration, by \cref{lemma:load_balancing}. This is done to prevent the case of a node holding more than $\tk$ tokens.
We then continue to send the tokens up to the root, which takes at most $O(\log{n})$ iterations, by the depth of the cluster tree $\mathbb{T}_\diamond$. Considering the $\tildeBigO{\tk}$-round load balancing step at the beginning of each iteration, this step takes $O(\log{n}) \cdot \tildeBigO{\tk}=\tildeBigO{\tk}$ rounds.

Now, the root cluster holds all of the tokens. It again load balances the tokens within its nodes so that each node holds at most $\tk$ tokens. We now send down the tokens in the same manner in cluster tree. For $O(\log{n})$ iterations, each node sends its at most $\tk$ tokens to its matched nodes in the children clusters, through the global network. 
Now each cluster holds all the tokens. Each node floods all tokens through the local network for $4\tk\ceil{\log{n}}$ rounds. By \cref{lem:clustering_with_weak_diam}, the weak diameter of each cluster is at most  $4\tk\ceil{\log{n}}$, so all nodes in each cluster learn all the tokens. 
Similarly, this phase takes $\tildeBigO{\tk}$ rounds.

Now \kdis is solved, as every node in $G$ knows all $k$ tokens. Summing over all phases, the algorithm takes $\tildeBigO{\tk}$ rounds.
\end{proof}


We now extend the proof of \cref{thm:optimal_dissemination} to prove \cref{thm:kaggUB}.

\kaggUB*

\begin{proof}
    We show that we can indeed solve \kagg in $\tildeBigO{\tk}$ rounds, deterministically. Note that once \emph{only one} node learns the results of all $k$ aggregate functions, we can disseminate it in $\tildeBigO{\tk}$ rounds by \cref{thm:optimal_dissemination}. We use similar steps to the proof of \cref{thm:optimal_dissemination}. First, we cluster the nodes using the same procedure. We then compute inside each disjoint cluster $k$ intermediate aggregations, and load balance it inside the cluster with \cref{lemma:load_balancing}. That way, each node holds at most $\tk$ aggregation results. After that, we use the cluster tree and cluster chaining in the proof of \cref{thm:optimal_dissemination} to send the intermediate aggregation results up the cluster tree to the root cluster. In each step, we load balance again. As the depth of the constructed cluster tree is at most $\tildeBigO{1}$, this process finishes in $\tildeBigO{\tk}$ rounds. Once all the information is stored in the root cluster, we flood it inside it and compute locally the final $k$ aggregation results. This step takes $\tildeBigO{\tk}$ rounds by the weak diameter bound of our clusters, which is at most $4\tk\log{n}$.
    Finally, we disseminate the $k$ aggregation results from any node in the root cluster to the entire graph, using \cref{thm:optimal_dissemination} in $\tildeBigO{\tk}$ rounds.
\end{proof}

\section{Universally Optimal Multi-Message Unicast}\label{sec:optimal_unicast}

The \klrout (see \cref{def:klrout}) corresponds to the problem of each of $k$ source nodes sending a small message ($\tilO(1)$ bits) to each of $\ell$ dedicated target nodes. This is also known as \emph{unicasting} of individual point-to-point messages between dedicated pairs of nodes, which is in contrast to the previous section where every node is supposed to obtain one and the same copy of each message.

In this section we show that this message routing problem can be solved much faster 
than if one would simply broadcast each individual message. In particular, we show that $\tilO\left(\NQ\right)$ rounds suffice to solve the \klrout problem, under certain conditions. 
Note that broadcasting the $k\ell$ individual messages takes \smash{$\tilO\big(\NQ[k\ell]\big)$} rounds (\cref{thm:optimal_dissemination}) which is slower in general (by a factor up to $\sqrt{\ell}$, see \cref{lem:growth_of_tk}).
It is also better than the existential bound for the routing problem of \smash{$\tilT\big(\!\sqrt{k} + \sqrt{\ell}\big)$} (due to \cite{Kuhn2020, Schneider2023}), see \cref{lem:boundingTk}.


We start by introducing a distributed structure that assigns each source or target node a certain set of helper nodes in its local proximity. This structure can be used to enhance the global communication bandwidth of each node and is adaptive to the neighborhood quality $\NQ$ of the graph.



\subsection{Adaptive Helper Sets}
\label{sec:adaptive-helper-sets}

To obtain a solution for the \klrout problem we gear some techniques of \cite{Kuhn2020} for the parameter $\NQ$ by computing a structure of adaptive helper sets that accommodates the neighborhood quality of a given graph $G$.

\begin{definition}[Adaptive helper sets]
	\label{def:adaptive_helpers}
	Let $G = (V,E)$, $k \in [n]$ and let $W \subseteq V$. 
	A family $\{H_w \subseteq V \mid w \in W\}$ of \emph{adaptive helper sets} has the following properties. 

    \begin{enumerate}[(1)]
        \item Each $H_w$ has size at least $k/\NQ$. 
        
        \item For all $u \!\in\! H_w\!: \hop(w,u) \!\in\!  \tilO(\NQ)$.
 	
 	\item Each node is part of at most $\tilO(1)$ sets $H_w$.
    \end{enumerate}
\end{definition}

In contrast to the helper sets by \cite{Kuhn2020} we require a stronger guarantees for the helper sets, that leverages the graph parameter $\NQ$. In particular, the size of each helper set $H_w$ can be much larger and yet much closer to $w$ than in \cite{Kuhn2020}, which guarantees only $\tilT(\!\sqrt{k})$ for both of these parameters (and $\tilO(\NQ) \subseteq \tilO(\!\sqrt{k}) \subseteq \tilO(k/\NQ)$ by \cref{lem:boundingTk}).
Such a helper set does not necessarily exist under any circumstance, in particular, if $W$ is too large or too locally concentrated.
We prove that we can compute such a helper set in the following lemma under the following conditions. \cref{alg:compute-helpers} gives an overview of the steps.

\begin{algorithm}
    \caption{Compute Adaptive Helper Sets $\{H_w \subseteq V \mid w \in W\}$ (\cref{def:adaptive_helpers})}
    \label{alg:compute-helpers}
         Compute $\NQ$ using \cref{lem:compute_neigh_qual}
        
         

         Compute a clustering as in \cref{lem:clustering_with_weak_diam}

        
         For each cluster $C$ let \smash{$q_C = \min\big(1, \tfrac{k}{\NQ} \cdot \tfrac{1}{|C|} \cdot  8c \ln n\big)$}, \hfill {\small\textit{(const.\ $c$ controls success probability)}}
        
         For each cluster $C$ and each $w \in C \cap W$, each node $v \in C$ joins $H_w$ with $w \in C$ with probability $q_C$
\end{algorithm}


\begin{lemma}
	\label{lem:adaptive_helpers}
	Let $W \subseteq V$ be a subset of nodes of $V$ that join $W$ with probability at most $\frac{\NQ}{k}$. A family of \emph{adaptive} helper sets $\{H_w \subseteq V \mid w \in W \}$ for $W$ as specified in \cref{def:adaptive_helpers} can be computed w.h.p.\ in $\tilO(\NQ)$ rounds in \hybrid.
\end{lemma}

\begin{proof}
    First, nodes collaborate to compute the parameter $\NQ$ in $\tilO\left(\NQ\right)$ rounds, see \cref{lem:compute_neigh_qual}. 
    Second nodes compute the clustering with the conditions of \cref{lem:clustering_with_weak_diam}.
    In a sense, such a clustering locally partitions $G$ into ``areas of responsibility'', and allows us to assign each node $w \in W$ a helper set $H_w$ from the cluster $C$ that $w$ is located in, in a fair way such that no node in $C$ has to help too many nodes from $W$.
 
	

    For this to work, the nodes in $C$ first learn $C$ and $C \cap W$. Note that the weak diameter of $C$ is in $\tilO\left(\NQ\right)$
    thus we can compute the required information in the same number of rounds using the local network. 
    Then, for each node $w \in C \cap W$, each $v \in C$ joins the helper set $H_w$ with probability 
	 $$q_C = \min\big(1, \tfrac{k}{\NQ} \cdot \tfrac{1}{|C|} \cdot  8c \ln n\big).$$ 
    It is clear that this assignment of helper sets fulfills property (2) of \cref{def:adaptive_helpers} due to the weak diameter of $C$. Let us now look at the size of each helper set.
  
    The case $q=1$ implies that each node in the cluster $C$ that $w$ is located in, is drafted into $H_w$. Furthermore, using that each cluster has size at least $k/\NQ$ by \cref{lem:clustering_with_weak_diam}, we have $$|H_w| \geq k/\NQ.$$
    
    In case $q<1$, the expected size of $H_w$ is \smash{$\E(|H_w|) = q_C \cdot |C| = \tfrac{k}{\NQ} \cdot  8c \ln n$}, whereas $q$ is chosen such that we can conveniently apply a Chernoff bound (given in \cref{lem:chernoffbound} for completeness) to give a lower bound for the size of $H_w$ that holds w.h.p.:
	 \[
	 	\Prob\Big(|H_w| \leq  (1 - \tfrac{1}{2}) \E(|H_w|)\Big) \leq \exp \Big(\!-\!\tfrac{\E(|H_w|)}{8}\Big) \leq \exp \Big(\!-\!\tfrac{8c \ln n}{8}\Big) = \tfrac{1}{n^c}. 
	 \]
	 Therefore, the size of $H_w$ is at least $\E(|H_w|) /2 \geq k/\NQ$ w.h.p. 
	 
  For the third property, we exploit that $W$ was sampled uniformly at random with probability at most $\frac{\NQ}{k}$, thus the expected size of $C \cap W$ is $\E\big(|C \cap W|\big) \leq |C| \cdot  \tfrac{ \NQ}{k}.$  
  Let $J_v = |\{w \in W \mid v \in H_w\}|$ be the random number of helper sets that a node $v \in C$ joins. In expectation these are at most $\E(J_v) = q \cdot  \E\big(|C \cap W|\big) \leq 8 c \ln n$. Again, we employ a Chernoff bound (\cref{lem:chernoffbound}) to show $J_v \in \bigOa(\log n)$ w.h.p.:
	 \[
	 	\Prob\big(J_v \geq  (1+1) 8 c \ln n\big) \leq \exp \big(\!-\!\tfrac{8c \ln n}{3}\big) \leq  \tfrac{1}{n^c}. 
	 \]
    Finally, note that we consider the properties above for each node $v \in V$ and each set $H_w$. All these properties or ``events'' must hold or ``occur'' simultaneously. Since the number of these events is at most polynomial in $n$, all these events occur w.h.p.\ by a union bound (\cref{lem:unionbound}).
\end{proof}

\subsection{Solving the \texorpdfstring{\klrout}{(k,l)-routing} Problem}
\label{sec:adaptive_helper_sets}

We show how to solve the \klrout problem (\cref{def:klrout}) in \smash{$\tilO\left(\NQ\right)$} rounds under various conditions. Routing a total of $k \cdot \ell$ message pairs in $\tilO(\NQ)$ rounds is, however, subject to certain bottlenecks. For instance, it inherently requires that each source and target node has (almost) exclusive access to a sufficiently large set of nearby helper nodes, or that the number of messages each source or target has to send or receive is relatively small.

If, by contrast, both the set of sources or targets are locally highly concentrated and have a large number of messages to send, such that there are not enough helper nodes for each source or target, then local limitations for sending and receiving messages imply that $\tilO(\NQ)$ rounds are unattainable and
conducting a broadcast of all messages in $\tilO(\NQ[k\ell])$ rounds is in fact the best that can be achieved.

To realize the desired speedup, we require that the set of source nodes and the set of target nodes satisfy some condition of being ``well spread out'' in the network, which is fulfilled w.h.p.\ by sampling them randomly. Alternatively, we can also allow one of the sets (sources or targets) to be arbitrary, however, in this case, the opposite set has to be much smaller to make the workload of each node in the former set sufficiently light. The following theorem expresses this formally. See \cref{def:klrout} for the formal definition of the \klrout problem.

\UnicastUB*

The remaining part of this section is dedicated to the proof, of which we would first like to give a rough overview and then highlight the challenges that need to be resolved. First of all, we will show in the proof that the roles of sources and targets can be reversed, thus we will restrict ourselves to explain case (1) and case (3) of \cref{thm:routing}, the latter with the assumption $\ell \leq k$.

In case (1) of \cref{thm:routing} the idea is to draft a family of helper sets $\{H_t \subseteq V \mid t \in T \}$ for the target nodes $T$ using \cref{lem:adaptive_helpers}, exploiting the prerequisite that $T$ was randomly sampled. The rough idea is that each source can directly stream its messages for each $t \in T$ to one of the helper nodes $v \in H_t$ of $t$ using the global network. Afterwards, each target can collect its assigned messages from each helper over the local network.

In case (3) of \cref{thm:routing} we compute helper sets $\{H_t \subseteq V \mid t \in T \}$ and $\{H_s \subseteq V \mid s \in S \}$ for both sets of sources and targets, respectively. The very rough idea is to assign each helper $u \in H_s$ of $s \in S$ a helper of $H_t$ of $t \in T$ in a balanced way. Then we stream the messages of every source $s \in S$ to their respective helpers $H_s$ over the local network. Each $u \in H_s$ that has a message for $t \in T$ then sends its message to some node in $H_t$ over the global network, balancing out the workload over the nodes and over time. Finally, each $t \in T$ can collect its assigned messages from its helpers $H_t$ over the local network.

Under ideal conditions, this allows for all assigned messages to be send within the time bounds specified in \cref{thm:routing}. Important caveats are: first, that senders and receivers are not aware of each others helper sets, and worse still, exchanging the IDs of their respective helper sets between $S$ and $T$ can be seen as an instance of \klrout problem that we aim to solve in the first place. Second, large imbalances in the number of sources and targets $k \gg \ell$ in case (3) of \cref{thm:routing}, can inhibit the computation of suitable adaptive helper from \cref{lem:adaptive_helpers}.

We resolve both of these issues in the following. We start by showing how the required information to assign helpers of targets to senders can be condensed significantly by instead relaying messages via a pseudo-random set of intermediate nodes (cf.\ \cite{Kuhn2020}).
We construct a suitable assignment of such intermediate nodes that satisfies our conditions pertaining to the parameter $\NQ$ can be obtained with techniques of universal hashing. 

We require a function $h$ that assigns an intermediate node to each pair $(s,t) \in S \times T$ that the message from $s$ to $t$ should be routed over. The mapping $h$ must satisfy a few properties (see subsequent lemma) with the following intuition. 
First, no node is assigned as an intermediate node of too many source-target pairs.  Second, global communication with intermediate nodes can adhere to the constraints of the \hybrid model.
Third, this mapping can be computed easily.

\begin{lemma}
    \label{lem:assign_helpers}
    Consider the \klrout problem (\cref{def:klrout}) with $k \cdot \ell  \in \bigOa\big(\NQ \cdot n)$. There exists a mapping $h: [n] \times [n] \to [n]$ such that the following is true.
    \begin{enumerate}[(1)]
        \item For each $v \in V$ it holds that $\big|\big\{(s,t) \in S \times T : h\big(\ID(s),\ID(t)\big) = \ID(v) \big\}\big| \in \bigOa(\NQ)$, w.h.p.

        \item If each node $v \in V$ has a unique tuple $(i,j) \in [n]^2$ and sends at most one message to the node with identifier $\ID(h(i,j))$, then each node receives at most $\bigOa(\log n)$ messages, w.h.p.
        
        \item Such a mapping $h$ can be computed by every node in $\tilO(\NQ)$ rounds in \hybrid.     
    \end{enumerate}
\end{lemma}

\begin{proof}
    The idea is to use a hash function that is selected randomly from a suitable universal family $\calH$, consisting of functions $h : [n] \times [n] \to [n]$. That is, $h \in \calH$ maps a pair of identifiers $(i, j) \in [n]^2, j \in [\ell]$ to an intermediate node with identifier $h(i,j) \in [n]$.
    To define $\calH$ we utilize methods of universal hashing from the literature and summarized in \cref{def:hashfunctions} and \cref{lem:hashfunctions}. We assume that $\calH$ is $\kappa$-wise independent for some $\kappa \in \Theta(\NQ \log n)$ (see \cref{def:hashfunctions}), but other than that we leave $\calH$ unspecified for now.
    Further, assume that we randomly select a function $h \in \calH$. Note that for this it suffices to publish a relatively small random seed whose size depends on the independence of $h$ and the size of its domain and image. 
    
    We start with property (1) and bound the number of collisions. For $v \in V$ and $i,j \in [n]$ with $i > k, j > \ell$ let $X_v = |\{(s,t) \in S \times T \mid v = h\big(\ID(s),\ID(t)\big)\}|$.
    Bounding $X_v$ corresponds to a \emph{balls-into-bins} distribution, where we have $k \cdot \ell \in \bigOa(\NQ \cdot n)$ balls and $n$ bins. An analysis for this distribution is given in \cref{lem:balls_in _bins}, which guarantees that \whp $X_v \leq a \cdot \frac{k\ell}{n} \in \bigOa(\NQ)$ for all $v \in V$ for some constant $a  \in O(1)$.
    We conclude that no intermediate node $v$ is ``hit'' by $h\big(\ID(s),\ID(t)\big)$ more than $X_u \in \bigOa(\NQ)$ times w.h.p. 

    For property (2), observe that here we consider a balls-into-bins distribution with $|V| = n$ balls and $|V| = n$ bins using a  is a hash function $h$ with independence at least $\kappa \in \Omega(\log n)$. Therefore, by setting $m \in O(n \log n)$ in  \cref{lem:balls_in _bins}, we conclude that the number of messages received by a node is at most $O(\log n)$ for all nodes \whp.
    
    
    Property (3) is mainly about the size of the universal family $\calH$ that is required to ensure $\kappa$-wise independence for $\kappa \in \Theta(\NQ)$. 
    In particular, the size of such a family induces the size $\tilO(\NQ)$ of the random seed required to select a random function from it, see 
    \cref{lem:hashfunctions}. It is easy to create such a seed locally at a node, break it up into multiple messages and broadcast it using our methods from the previous section in $\tilO(\NQ)$ rounds, see \cref{thm:optimal_dissemination}.
\end{proof}

While \cref{thm:routing} allows flexibility in terms of the parameters $k$ and $\ell$ that dictate the size of the set of sources with the combined bound $k \cdot \ell \leq {n \cdot \NQ}$, this poses some challenges for the computation of helper sets. Specifically, if $k$ or $\ell$ is much larger then the other value, then the bound on the sampling probability of \cref{lem:adaptive_helpers} is violated. For instance, for $k=n$ and $\ell = \NQ$, helper sets for $S$ with $\tilO(1)$ workload per helper can only be of (too) small size \tilO(1).

On the positive side, we observe that if $k \gg \sqrt{n \cdot \NQ}$ then each source has to send only $\ell \ll k$ messages.
The next lemma shows instead how to reduce the \klrout problem under general conditions of \cref{thm:routing} case (3) to relatively few instances of the \klrout[k',\ell'] with $k', \ell' \leq \sqrt{n \cdot \NQ}$. This is sufficient to satisfy the condition of \cref{lem:adaptive_helpers}.

\begin{lemma}
    \label{lem:consolidate_sources}
    Given the \klrout problem for $\ell \leq k$ and $k \cdot \ell  \leq  \NQ \cdot n$ with $k > \sqrt{n \cdot \NQ}$ where source and target nodes are sampled with probabilities $\tfrac{k}{n}$ and $\tfrac{\ell}{n}$. This problem can be reduced to $\tilO(1)$ instances of \klrout[k',\ell'] problems with parameters $\ell' \leq k' \leq \sqrt{n \cdot \NQ}$ in $\tilO(\NQ)$ rounds, w.h.p., where source and target nodes are sampled with probabilities $\tfrac{k'}{n}$ and $\tfrac{\ell'}{n}$.
\end{lemma}

\begin{proof}      
    On one hand, the idea is to consolidate the sources $S$ into a sparser subset of ``super sources'' $S' \subseteq S$ that collect all messages of source nodes $S$ and will then take care of sending the collected messages. On the other hand, we compute a larger set of ``sub-targets'' $T' \subseteq V$, which takes care of receiving the messages from the super sources, such that still each super source has at most one message for each sub-target. The sets $S'$ and $T'$ are randomly sampled to maintain the conditions stated in \cref{thm:routing} (2).
    
    The first step is to compute $S'$. After sampling $S$ with probability $k/n$, each node from $S$ decides to join $S'$ with probability $p = \min\big(\frac{\NQ\cdot n}{k^2} \cdot 8c\ln n,1\big)$ for a constant $c >0$.
    Next, we compute a clustering according to \cref{lem:clustering_with_weak_diam} in $\tilO(\NQ)$ rounds.
    For $s \in S$ that is located in some cluster $C$, let $X = |C \cap S'|$ be the random number of super sources within the cluster $C$. We want to show that $X \geq 1$, i.e., each cluster $C$ contains at least one super source $s' \in S' \cap C$. If $p=1$, then $s \in S'$ and we are done. Else, by the size property of the clustering we have $\E(X) \geq \frac{k}{\NQ} \cdot \tfrac{k}{n} \cdot p \geq 8c\ln n$. Applying the Chernoff bound from \cref{lem:chernoffbound}, we obtain
    \[
	 	\Prob\Big(X \leq  (1 - \tfrac{1}{2}) \E(X)\Big) \leq \exp \Big(\!-\!\tfrac{\E(X)}{8}\Big) \leq \exp \Big(\!-\!\tfrac{8c \ln n}{8}\Big) = \tfrac{1}{n^c}. 
    \]
    Thus $X \geq \E(X)/2 \geq 1$ w.h.p.

    Next, within each cluster $C$ we assign all sources $s \in S \cap C$ to one of the (at least 1, w.h.p.) super sources $s' \in S' \cap C$ in a balanced way (difference between assigned sources between all $S' \cap C$ at most 1). This can be done $\tilO(\NQ)$ rounds, due to the weak diameter $\tilO(\NQ)$ of each cluster (\cref{lem:clustering_with_weak_diam}).
    For some $s' \in S' \cap C$ let $Y$ be the random number of assigned sources to $s$. In expectation, $Y$ will correspond to the inverse probability  with which we sampled $S'$ from $S$ (plus one to account for slight imbalances in the assignment of sources to super sources), i.e., 
    \[\E(Y) \leq \tfrac{1}{p} + 1 \leq \tfrac{k^2}{\NQ\cdot n \cdot 8c\ln n} +1\]
    Again, with a Chernoff bound we can limit the absolute number of assigned sources w.h.p., as follows $\Prob\big(Y \geq (1+ 3c\ln n) \E(Y)\big) \leq \exp\big(-c \ln n \cdot \E(Y)\big) \leq \exp(-c \ln n) = 1/n^{-c}$. Hence w.h.p.
    \[ Y \leq (1+ 3c\ln n) \cdot \E(Y) \leq \tfrac{k^2}{\NQ\cdot n},\] 
    where the last step holds for sufficiently large $n$.
    Since each source originally had to send $\ell \leq \frac{n \cdot \NQ}{k}$ messages, the number of messages per super source $s' \in S'$ is capped at $\ell \cdot Y \leq k$. We also want to limit the random number $|S'|$ of super sources that we obtain. The overall probability for a node to join $S'$ is $\tfrac{kp}{n}$, thus $\E(|S'|) = kp \leq \frac{\NQ \cdot n}{k} \cdot 8c\ln n$. We apply another Chernoff bound
    \[
    \Prob\Big(|S'| \leq  (1+1) \E(|S'|)\Big) \leq \exp \Big(\!-\!\tfrac{\E(|S'|)}{3}\Big) \leq \exp \Big(\!-\!\tfrac{8c \ln n}{8}\Big) = \tfrac{1}{n^c}. 
    \]
    Thus, w.h.p., $|S'| \leq 2 \E(k') \leq  \frac{\NQ \cdot n}{k} \cdot 16 c\ln n \leq \sqrt{\NQ \cdot n} \cdot 16 c\ln n \in \tilO(\sqrt{\NQ \cdot n})$.  Note that $|S|,|S'| \in \tilO(k)$, w.h.p., thus it can be made publicly known which super source $s' \in S'$ is responsible for which source $s \in S$ in $\tilO(\NQ)$ rounds using \cref{thm:optimal_dissemination}.

    The second step is to compute the sub-targets $T'$, which we sample from $V$ with probability $q = \min\big(\tfrac{k}{n} \cdot 8 c \ln n,1\big)$. In each cluster $C$ we will assign each target $t \in T \cap C$ roughly the same number of sub-targets $t' \in T' \cap C$ in the same cluster (with difference at most 1) in $\tilO(\NQ)$ rounds.  
    We also need certain bounds for the random number $Z$ of assigned sub-target nodes per target node within a cluster $C$ and the overall number $|T'|$. We obtain $Z \geq \tfrac{k}{\ell}$ and $|T'| \leq 2 p n \in \tilO(k) \subseteq \tilO(\sqrt{\NQ \cdot n})$ w.h.p. We skip the details for computing the expectations and applying Chernoff bounds, which are very similar to before.

    The choices of the sampling probabilities for $S'$ and $T'$ lead to the following. The number $Y$ of of source nodes $s \in S$ from which a super node $s' \in S'$ collected the messages is at most as large as the number $Z$ of assigned sub-targets $t' \in T$ per target node $t \in T$, w.h.p.:
    \[
        Y \leq \tfrac{k^2}{\NQ\cdot n} \leq \tfrac{k}{\sqrt{\NQ\cdot n}} \leq \tfrac{k}{\ell} \leq Z
    \]
    The fact that $Y \leq Z$ allows us to do the following steps. For each super source $s' \in S'$ and each $s \in S$ whose messages were collected by $s'$, the target node $t \in T$ will task one of its assigned sub-targets $t' \in H_t$ with receiving the message from $s$ to $t$ from $s'$. 

    We now have the sets $S'$ and $T'$ of roughly the correct sizes, which almost sets up the new routing problem, with one caveat: The super sources $S'$ do not yet know which sub-targets in $T'$ to send their messages to. We can approach this by first solving the routing problem the other way around. Each sub-target $t' \in T'$ prepares a message for each super source $s' \in S'$ from which it wants to receive a message, as assigned by the original target $T$. Note that this routing problem can be set up in a canonical way leveraging that all nodes have knowledge of the identifiers of $S'$ from our previous broadcast.
    After this step, the nodes in $S'$ have the identifiers of the nodes in $T'$ which they need to send their messages to.

    The last point is about the sampling rates $k',\ell'$ of $S',T'$. Note that the probability with which we sample each set is $\frac{k}{n} \cdot p \leq \E(|S'|)/n \in \tilO(\sqrt{\NQ/n})$ and $q \leq \E(|T'|)/n \in \tilO(\sqrt{\NQ/n})$ with $\bigOa(\log n)$ factors that we require to make the Chernoff bounds work. Thus our sampling rates $k',\ell' \in \tilO(\sqrt{\NQ\cdot n})$ are effectively off by a polylogarithmic factor.     
    The remedy is to do a random subdivision where each node of $S',T'$ randomly joins one of $\sigma$ groups for appropriately chosen $\sigma \in \tilO(1)$, such that the effective probability to end up in a certain group falls below $\sqrt{\NQ/n}$ thus satisfying our condition on the sampling probabilities $\frac{k'}{n}, \frac{\ell'}{n}$ as required in this lemma. We then solve the routing problem for each of the $\sigma^2 \in \tilO(1)$ pairings of such groups. The overall number \klrout[k',\ell'] problems to solve is $\tilO(1)$.
\end{proof}

The previous lemmas do most of the technical lifting for the proof of \cref{thm:routing}, whereas it remains to show how the communication can be coordinated efficiently via the helper and intermediate nodes in the \hybrid model such that the following holds. First, the number of global messages any node receives in any given round is bounded by $\bigOa(\log n)$, so that no messages are dropped and lost (see \cref{sec:models}). Second, assuming $\ell \leq k$, the parameters fit together in such a way that the communication for the \klrout problem can be completed in $\tilO(\NQ)$ rounds.
Third, we show that the roles of source nodes and target nodes can be reversed while keeping the same upper bounds, which reduces case (2) to case (1), and shows case (3) without the assumption $\ell \leq k$.
\cref{alg:routing} gives a very high-level overview of the involved steps for case (1) and (3) of \cref{thm:routing} under the assumption $\ell \leq k \leq \sqrt{n \cdot \NQ}$ (which generalizes to all cases in the \cref{thm:routing} as shown in the subsequent proof).

\begin{algorithm}
    \caption{Multi-Message Unicast (\cref{thm:routing} cases (1) and (3) with $\ell \leq k \leq \sqrt{n \NQ}$)}
    \label{alg:routing}
    
        Compute adaptive helper sets $H_t, H_s$ for each $t \in T, s \in S$ using \cref{alg:compute-helpers}

        In case (1) of \cref{thm:routing} define $H_s := \{s\}$ instead

        Each source $s \in S$ labels its message for receiver $\ID(t)$ with the tuple $(\ID(s),\ID(t))$

        Each source $s$ sends its messages to its helpers $H_s$ in a balanced way over the local network

        Compute function $h$ as in \cref{lem:assign_helpers}
        
        \While {\normalfont some helper $u \in H_s$ still has a message to send}{
            $u$ sends message with label $(i,j)$ to intermediate node $\ID\big(h(i,j)\big)$ one per round
        }

        Broadcast the set of identifiers of $S$ using \cref{thm:optimal_dissemination}
        
        Each target $t \in T$ creates a request labelled with $(\ID(s),\ID(t))$ for each $i \in [k]$

        Each target $t$ sends its requests to its helpers $H_t$ in a balanced way over the local network
    
        \While {\normalfont some helper $w \in H_t$ still has a request}{
            $w \in H_t$ sends one request with label $(i,j)$ to intermediate node $v$ with $\ID\big(h(i,j)\big)$

            $v$ sends the message with label $(i,j)$ back to $w$.
        }

        Each target $t$ collects its assigned messages from its helpers $H_t$ over the local network
\end{algorithm}

\begin{proof}[Proof of \cref{thm:routing}]
    We give a proof for case (1) and (3) of \cref{thm:routing} under the assumption $\ell \leq k \leq \sqrt{n \cdot \NQ}$.
    Note that $k \leq \sqrt{n \cdot \NQ}$ is w.l.o.g., since solving the problem under this assumption in $\tilO(\NQ)$ rounds induces a solution for the general case $k \cdot \ell \leq {n \cdot \NQ}$ in $\tilO(\NQ)$ rounds due to \cref{lem:consolidate_sources}. At the end of this proof, we will also show that restricting ourselves to case (1) and (3) and the assumption of $\ell \leq k$ are w.l.o.g.
    
    We first compute the parameter \NQ in \smash{$\tilO\left(\NQ\right)$} rounds according to \cref{lem:compute_neigh_qual}. Note that it is simple to sum up and publish the number of source and target nodes $k$ and $\ell$ in $\tilO(1)$ rounds in the global network using the aggregation routine by \cite{Augustine2019}.     
    We conduct this proof under the assumption that $k \leq \sqrt{n \cdot \NQ}$. 
    
    This assumption $k \leq \sqrt{n \cdot \NQ}$ helps us in the following way. The inequality $k \leq \sqrt{n \cdot \NQ}$ is equivalent to $\tfrac{k}{n} \leq \frac{\NQ}{k}$, where the left hand side corresponds to the probability for sampling source nodes $S$ and the right hand side is the bound that is allowed by \cref{lem:adaptive_helpers} for the computation of adaptive helper sets. Consequently, our assumption on $k$ allows us the computation of adaptive helper sets for the set of sources $S$ as in \cref{lem:adaptive_helpers}. Due to $\ell \leq k$ the same is true for the set targets $T$.
    
    Next, we compute adaptive helper sets $H_t$ for each target $t \in T$ in \smash{$\tilO\left(\NQ\right)$} rounds w.h.p., as shown in \cref{lem:adaptive_helpers}. In case (3) of the Theorem, we also compute $H_s$ for the sources $s \in S$. In case (1) we simply set $H_s = \{s\}$ meaning each source $s$ must send all of its messages itself.

    We continue distributing the workload from sources and targets to their respective helpers, for which we leverage the knowledge of nodes in $S$ and $T$ of the identifiers of the opposing set $T$ and $S$, respectively. Although the targets $T$ do initially not know the identifiers in $S$, this is not very much information and can simply be broadcast in $\tilO(\NQ)$ rounds using \cref{thm:optimal_dissemination} (and the observation that $|S| \in \bigO{k \log n}$, by a simple application of a Chernoff bound, see \cref{lem:chernoffbound}).    

    Each source $s \in S$ labels its messages destined to some $t \in T$ with a tuple $(\ID(s), \ID(t))$ and disseminates them among its helpers $H_s$ such that each helper has roughly the same number of messages (with a difference of at most one). Similarly, the set of targets $T$ creates requests, which are tuples $(\ID(s), \ID(t))$ which it balances among its helpers in $H_t$.    
    This is a local task that can be done using only the local network in $\tilO(\NQ)$ rounds due to property (2) of \cref{lem:adaptive_helpers}.    

    To transmit messages, each sender helper $u \in H_s$ relays a given message for some target node $t$ to some helper $w \in H_t$ via some intermediate node, such that no helper or intermediate node has too much ``workload'' in terms of overall messages to send (at most $\tilO(\NQ)$) and in terms of the number of messages received per round (at most $\bigOa(\log n)$). For this, we compute the function $h$ with the properties given in \cref{lem:assign_helpers}. This takes $\tilO(\NQ)$ rounds, according to property (3) of \cref{lem:assign_helpers}.

    Then each helper $u \in H_s$ that has a message that is labelled with the (unique!) tuple $(i,j)$ sends this message to the node with ID $h(i,j)$. Note that we throttle the send rate that the helpers $H_s$ send their messages to a single message per round, so that no node receives more than $\bigOa(\log n)$ per round by property (2) of \cref{lem:assign_helpers}.

    We compute the number of rounds until all messages of any helpers $u \in H_s$ for all $s\in S$ have been sent to intermediate nodes and make a case distinction. Case (1) of \cref{thm:routing} is simple, each sender has to send at most $\ell \in \tilO(\NQ)$ messages, which take $\ell$ rounds to send. In \cref{thm:routing} case (3) we exploit property (1) of \cref{def:adaptive_helpers}, that each source $s \in S$ has at least $|H_s| \geq k/\NQ$ helpers. This means that each helper $u \in H_s$ will be responsible for at most $\bigOa(\NQ)$ messages from $s$. By property (3) of \cref{def:adaptive_helpers}, each node is part of at most $\tilO(1)$ sets $H_s$. Therefore, each node has to send at most $\tilO(\NQ)$ messages in their role as helper of some source $s \in S$, which can be accomplished in as many rounds.

    On the side of the target helpers $H_t$ we do similar steps to retrieve messages from the intermediate nodes. In alternating rounds, each helper $w \in H_t$ will send one(!) request $(i,j)$ to the helper $h(i,j)$ to which it attaches its own identifier $\ID(w)$. In the subsequent round, each intermediate node $v$ that received a request $(i,j)$ from helper $w$ will send the message with label $(i, j)$, to the requesting helper $w$ using the handle $\ID(w)$. By property (3) of \cref{lem:assign_helpers} each intermediate node receives and has to reply to at most $\bigOa(\log n)$ requests in this way.

    Finally, we want to reduce case (2) of \cref{thm:routing} to case (1) and show that the assumption $\ell \leq k$ in case (3) can be made w.l.o.g.
    Specifically, we aim to solve case (2) and case (3) of \cref{thm:routing} for $k \leq \ell$, in $\tilO(\NQ[\ell])$ rounds using \cref{alg:routing} that solves \cref{thm:routing} case (1) and case (3) for $\ell \leq k$ in $\calT$ rounds.
    
    First each $t \in T$ creates a logging message. Note that the ``reverse'' routing problem where each node $t \in T$ sends a logging message to each $s \in T$ conforms to case (1) and case (3) for $\ell \leq k$, which we can solve in $\calT \in \tilO(\NQ[\ell])$ rounds using \cref{alg:routing}.    
    While the logging messages are in transit, we attach the information in which rounds these messages got forwarded between helpers and intermediate nodes. Note that this information has size $\bigOa(\log n)$ and can indeed be attached to the logging message. After delivering all logging messages, each $s \in S$ attaches the logging message from $t \in T$ to its own message to $t$. 
    
    This allows the message from $s$ to $t$ to retrace the path of the logging message in $\calT \in \tilO(\NQ[\ell])$ rounds as follows. The assignment of helper nodes remains in place as in the reverse problem, i.e., each helper $v \in H_s, s\in S$ that was assigned a request for some logging message from $t \in T$, is also responsible for forwarding the message from $s$ to $t$. All corresponding steps can easily be done in the local network in $\in \tilO(\NQ[\ell])$ rounds.     
    Then a transition of a messages from $s \in S$ to some $t \in T$ among helpers and intermediate nodes occurs in round $\tau$, when the reverse transition was made by the logging message in round $\calT_\calA - \tau$. By symmetry, it is clear, that since all the logging messages were delivered w.h.p.\ the same is true for the actual messages.
\end{proof}

\section{Universally Optimal Shortest Paths}
\label{sec:upper_bound}

In this section, using the broadcast and unicast tools developed in \cref{sec:optimal_broadcast,sec:optimal_unicast}, we design various universally optimal algorithms for approximating distances and cut sizes. As we will later see, our \emph{existentially optimal} shortest paths algorithms in \cref{thm:almost_shortest_sssp,thm:k-ssp} are key technical ingredients of some of our universally optimal shortest paths algorithms.

In \cref{subsect:universalklsp}, we show  $(1+\eps)$-approximation algorithms for the $(k,\ell)$-SP problem for a certain range of parameters of $k$ and $\ell$ in weighted graphs.
In \cref{subsect:APSPu}, we show  $(1+\eps)$-approximation algorithms for APSP in unweighted graphs.
In \cref{subsect:APSPw}, we show  $O(1)$-approximation algorithms for APSP in weighted graphs.
In \cref{subsect:cuts}, we combine our broadcast tool with a spectral sparsifier to approximate all cut sizes.

\subsection{\texorpdfstring{$(1+\eps)$}{(1+eps)}-Approximate \texorpdfstring{$(k,\ell)$-SP}{(k,l)-SP} in Weighted Graphs}\label{subsect:universalklsp}

In this section, we give $(1\p\eps)$-approximation algorithms for the $(k,\ell)$-SP problem for certain parameters of $k$ and $\ell$. Our algorithms take $\tilO\left(\NQ\right)$ rounds for any graph $G=(V,E)$, which is competitive with the best algorithm optimized for $G$ up to $\tilT(1)$ factors, due to a matching lower bound that holds even for the $(k,1)$-SP problem in \cref{sec:lower_bound}, whereas our solution can handle up to $\ell \leq \NQ^2$ random target nodes.

\uniKLSP*

\begin{proof}
    For the first claim we solve the $\ell'$-SSP problem for the set of target nodes, which has size at most $\ell' \in \bigOa(\ell \log n) \subseteq \tilO(\NQ)$ w.h.p.\ (the bound on $\ell'$ is a simple application of the Chernoff bound in \cref{lem:chernoffbound}). For this we set up $\ell'$ instances of the SSSP problem which we solve with \cref{thm:almost_shortest_sssp}. Since we consider $\varepsilon$ constant and $\ell' \in \tilO(\NQ)$, the round complexity for this step is at most $\tilO\left(\NQ\right)$.    
    For the second claim, where $\ell \leq \NQ^2$ and $\ell \cdot k \leq \NQ \cdot n$ we first solve the shortest paths problem for the set of target nodes in $\tilO(\NQ)$ rounds with \cref{thm:k-ssp}.

    Afterwards, in either claim, each source $s$ knows its distance to each target $t$, but we require this the other way around in order to solve the problem. This is a task that we can solve with \cref{thm:routing}, where each source $s$ puts its distance to each target $t$ into a message, and the goal is to deliver all these messages in $\tilO\left(\NQ\right)$ rounds.
    For the first claim we can apply \cref{thm:routing} (1) due to $\ell \leq \NQ$. For the second claim we can apply \cref{thm:routing} (3) due to the condition $\ell \cdot k \leq \NQ \cdot n$.    
\end{proof}

\subsection{\texorpdfstring{$(1+\eps)$}{(1+eps)}-Approximate APSP in Unweighted Graphs} 
\label{subsect:APSPu}
We now proceed to prove \cref{thm:unweighted_apsp_approx} by utilizing our $(1+\eps)$-approximate SSSP result from \cref{thm:almost_shortest_sssp}. See \cref{alg:unweighted_apsp} for an overview of our algorithm.

\unweightedAPSPApprox*



\begin{proof}
    We begin by computing $\tk[n]$ in $\tildeBigO{\tk[n]}$ rounds using \cref{lem:compute_neigh_qual}, so from now on we can assume all nodes know this value.
    We proceed by broadcasting the identifiers of all the nodes in $\tildeBigO{\tn}$ rounds using \cref{thm:optimal_dissemination}. 
    We execute \cref{lem:clustering_with_weak_diam} with $k = n$, in $\tildeBigO{\tk[n]}$ rounds, to cluster the nodes such that we know a set $R$ of cluster leaders, each cluster $C$ has weak diameter at most $4\tn\ceil{\log{n}}$, and $n/\tk[n] \leq |C| \leq 2n/\tk[n]$.

    Now, observe that as the clusters are disjoint and each has size at least $n/\tn$, then we have at most $\tn$ clusters and as such $|R| \leq \tn$. Using \cref{thm:almost_shortest_sssp}, it is possible to compute $(1+\eps)$-approximate distances from all the nodes in $R$ to the entire graph in $\tildeBigO{|R|/\eps^2} \subseteq \tildeBigO{\tn/\eps^2}$ rounds. Denote the computed approximate distances by $\hat{d}$.

    Each node $v$ learns its neighborhood of radius $x = (4\tn\ceil{\log{n}})/\eps$, denoted $\calB_{x}(v)$. This takes $O(x) \subseteq \tilO(\tn/\eps)$ rounds. After that, every node $v$ broadcasts its closest node in $R$, denoted $c_v \in R$, and the unweighted distance $d(v, c_v)$. As each node broadcasts $O(1)$ messages,  using \cref{thm:optimal_dissemination}, this requires $\tildeBigO{\tn}$ rounds.

    Finally, each node $v$ approximates its distance to each other node $w$ as follows. If $w \in \calB_x(v)$, then $v$ already knows its exact unweighted distance to $w$, as the local communication links have unlimited bandwidth, so $v$ sets $\delta(v, w) = d(v, w)$. Otherwise, $v$ sets $\delta(v, w) = \hat{d}(v, c_w) + d(w, c_w)$, where $c_w$ is the closest node in $R$ to $w$. Note that $v$ knows both $c_w$ and $d(w, c_w)$, as $w$ broadcasts these values in the previous step.

    \paragraph{Analysis.} We conclude the proof by showing that $\delta$ is a $(1+\eps')$ approximation of $d$ for some $\eps' \in \Theta(\eps)$.
    If $w\in{}B_{x}(v)$, then we already have $\delta(v,w)=d(v,w)$. Otherwise, $d(v,w)>x=4\tn\ceil{\log{n}}/\eps$, and $\delta = \hat{d}(v, c_w) + d(w, c_w)$. We begin by showing that $\delta(v, w) \geq d(v, w)$. As $\hat{d}(v,c_w)$ is a valid $(1+\eps)$-approximation, $\hat{d}(v,c_w)\geq{}d(v,c_w)$, and thus $\delta(v, w) \geq d(v, c_w) + d(w, c_w) \geq d(v, w)$, where the last inequality is due to the triangle inequality. We now bound $\delta(v,w)$ from above. It holds that $d(v,w)>4\tn\ceil{\log{n}}/\eps$ and $d(w,c_w)\leq{}4\tn\ceil{\log{n}}$, as the weak diameter of each cluster is at most $4\tn\ceil{\log{n}}$. Therefore, $d(w,c_w)\leq{}4\tn\ceil{\log{n}}<\eps\cdot{}d(v,w)$. As such, the following holds.
    \begin{align*}
        \delta(v,w)&=\hat{d}(v,c_w)+d(w, c_w)\\
        &\leq{}(1+\eps)\cdot d(v,c_w)+d(w,c_w) \\
        &\leq{}(1+\eps)\cdot \left(d(v,w)+d(w,c_w)\right)+d(w,c_w) \\
    &=(1+\eps)\cdot d(v,w)+(2+\eps)\cdot d(w,c_w) \\
    &<(1+\eps)\cdot d(v,w)+\left(2\eps+\eps^2\right)\cdot d(v,w) \\
    &=\left(1+3\eps+\eps^2\right)\cdot d(v,w)\\
   &=(1+\eps')\cdot d(v,w) & \eps' := 3\eps+\eps^2
    \end{align*}
 Therefore, we achieve a $(1+\eps')$ approximation, where $\eps' = 3\eps+\eps^2$. As $\eps \in (0, 1)$,  $\eps' < 4\eps$,  so by a simple change of parameter $\widetilde{\eps} = \eps/4$ we achieve the desired $(1+\eps)$ approximation result with our algorithm.
\end{proof}

\begin{algorithm}
\caption{$(1+\eps)$-Approximate Unweighted APSP} \label{alg:unweighted_apsp}
        Compute $\tn$

        Broadcast the identifiers of all the nodes.
        
        Run \cref{lem:clustering_with_weak_diam} with $k = n$ to get the set of cluster leaders $R$.
        
        Run $(1+\eps)$-approximate SSSP from each node in $R$. Denote computed distances by $\hat{d}$.
        
        Each node learns the $x = (4\tn\ceil{\log{n}})/\eps$-hop neighborhood.

        Using \cref{thm:optimal_dissemination}, each node $v$ broadcasts its closest cluster leader $c_v \in R$ and $d(v,c_v)$.
        
        Each node $v$ approximates its distance to any $w\in{V}$ by: 
        \[
        \delta(v,w)=
            \begin{cases}
                d(v,w), & w\in{}B_{x}(v)\\
                \hat{d}(v, c_w)+d(c_w,w), & \text{otherwise}
            \end{cases}
        \]
\end{algorithm}

\subsection{\texorpdfstring{$O(1)$}{O(1)}-Approximate APSP in Weighted Graphs}
\label{subsect:APSPw}

In this section, we show a simple deterministic APSP algorithm by broadcasting the spanner and a more complicated randomized APSP algorithm that involves skeleton graphs. We begin with proving \cref{thm:firstWeightedAPSPApprox}. 
A spanner is a sparse spanning subgraph that maintains a good approximation of distances in the original graph. More formally, a subgraph $H$ of a \emph{weighted} $G=(V,E)$ is a \emph{$\alpha$-spanner} if $d_H(u,v) \leq \alpha \cdot d_G(u,v)$ for all $u \in V$ and $v \in V$.

\begin{lemma} [\hspace{-0.01cm}{\cite[Corollary 3.16]{rozhovn2020polylogarithmic}}]
\label{thm:spanner_construction}
    Let $G=(V,E,\omega)$ be a weighted graph. For any integer $k$, there exists a deterministic algorithm in the \congest model that computes a $(2k-1)$-spanner with $O(kn^{1+1/k}\log n)$ edges in $\widetilde{O}(1)$ rounds.
\end{lemma}

To prove \cref{thm:firstWeightedAPSPApprox}, we execute \cref{thm:spanner_construction} and then broadcast the resulting spanner.

\firstWeightedAPSPApprox*
\begin{proof}
    We run the algorithm of \cref{thm:spanner_construction} with $k= \lceil \eps\log{n}/2 \rceil$ to receive a $(2k-1)$-spanner with \[m^\ast \in O\left(k\cdot{}n^{1+1/k}\cdot \log n\right) \subseteq \widetilde{O}\left(4^{1/\eps}\cdot{}n\right)\] edges. By \cref{lem:growth_of_tk}, $\tk[m^\ast] \in O(2^{1/\eps}\tn)$. Using the algorithm of \cref{thm:optimal_dissemination}, we can make the spanner globally known in $\tilO(\tk[m^\ast]) \subseteq \widetilde{O}(2^{1/\eps}\cdot{}\tn)$ rounds by broadcasting $m^\ast$ messages. This allows all nodes to approximate APSP by a factor of $2k-1 = 2 \lceil \eps\log{n}/2 \rceil - 1 < \eps\log{n} + 1$.
\end{proof}

For the rest of the section, we turn to the randomized setting and show an improved APSP algorithm.
Before we do so, we must introduce the well-known concept of \emph{skeleton graphs}, first observed by Ullman and Yannakakis~\cite{Ullman1991}. We define the $h$-hop distance $d_{G}^h(u,v)$ between two nodes $u$ and $v$ in $G$ as the minimum path length over all paths connecting $u$ and $v$ with at most $h$ hops in $G$. In case $\hop(u,v) > h$, we let $d_{G}^h(u,v) = \infty$.

\begin{definition}[Skeleton graphs]\label{def:skeleton_graph}
Given a parameter $x$, a skeleton graph $\calS = (V_\calS, E_\calS, \omega_\calS)$ of $G=(V,E,\omega)$, is obtained by sampling each node of a given weighted graph $G$ to $V_\calS$ independently with a probability that is at least $\frac{1}{x}$. The edge set of $\calS$ is defined as \[E_\calS = \left\{ \{u,v\} \in V_\calS \times V_\calS \mid u \neq v \text{ and } \hop_G(u,v) \leq h\right\},\] 
where $h:=\xi{}x\ln{n}$ for some sufficiently large constant $\xi > 0$. 
For each skeleton edge $e = \{u,v\} \in E_\calS$, its weight  $\omega_\calS(u,v)$ is defined as the $h$-hop distance $d_{G}^h(u,v)$ between $u$ and $v$ in $G$.
\end{definition}




Such a skeleton graph has many useful properties, the main ones being that it preserves distances of $G$ and can be computed efficiently in a distributed fashion, where each skeleton node knows all its incident skeleton edges.  From its definition, it is clear that a skeleton graph can be constructed in the \local model, and thus \hybridzero, in $h \in \tilO(x)$ rounds.  The following lemma summarizes the well-known properties of a skeleton graph.



\begin{lemma}[\hspace{-0.01cm}{\cite[Lemmas 4.2 and 4.3]{augustine2020shortest}}]
	\label{lem:skeleton-graph}
	A skeleton graph $\calS = (V_\calS, E_\calS)$  can be constructed in $h\in \tilO(x)$ rounds in the \LOCAL model with the following properties.
\begin{enumerate}[(1)]
    \item For any $u,v \in V$ with $\hop(u,v) \geq h$, there is a shortest path $P$ from $u$ to $v$ such that any subpath $Q$ of $P$ with at least $h$ nodes contains a node in $V_\calS$ w.h.p.
    \item For any $u,v \in V_\calS$, $d_\calS(u,v) = d_G(u,v)$ w.h.p.
\end{enumerate} 
\end{lemma}

We are now ready to prove \cref{theorem:secondWeightedAPSPApprox}. See \cref{alg:polynomial_approx_apsp} for an overview of our algorithm.

\begin{algorithm}
\caption{$(4\alpha - 1)$-Approximate Weighted APSP} \label{alg:polynomial_approx_apsp}
            Broadcast the identifiers of all the nodes.

            Compute $\tn$. Denote $t = n^{1/(3\alpha + 1)} \cdot \left(\tn\right)^{2/(3 + 1/\alpha)}$.
            
            Compute a skeleton graph $G_{\mathcal{S}}=(V_{\mathcal{S}}, E_{\mathcal{S}}, \omega_{\mathcal{S}})$ with sampling probability $1/t$.
            
            Compute a $(2\alpha-1)$-stretch spanner for $G_{\mathcal{S}}$, denoted $K$.

            Broadcast the edges of $K$. Locally compute a $(2\alpha - 1)$ approximation of distances between all nodes in $V_\mathcal{S}$, denoted $\hat{d}$.

            Each node learns the $h = \xi{}t\ln{n}$-hop neighborhood, where $\xi$ is the constant from \cref{def:skeleton_graph}.
            
            Each node $v \in V$ finds a skeleton node $v_s \in V_\mathcal{S}$ in its $h$-hop neighborhood with minimum $d^h(v, v_s)$ and broadcasts $v_s$ and $d^h(v, v_s)$.
            
            Each node $v$ approximates its distance to any $w\in{V}$ by:
            \[
            \delta(v,w)= \min\left\{d^h(v,w),d^h(v,v_s)+\hat{d}(v_s,w_s)+d^h(w_s,w)\right\}
            \]
\end{algorithm}

\secondWeightedAPSPApprox*

\begin{proof}
We may assume that $\alpha \in O(\log n)$, since otherwise we can solve the problem by simply running the algorithm of \cref{thm:firstWeightedAPSPApprox} with some $\eps \in \Theta(1)$ to attain the desired approximation ratio $\alpha$ in $\tilO(\tn)$ rounds. Hence we may hide any $\alpha$-factor in any function stated in $\tilO(\cdot)$ notation.

   We begin by broadcasting the identifiers of all the nodes in $\tildeBigO{\tn}$ rounds using \cref{thm:optimal_dissemination}.
    We then compute $\tk[n]$ in $\tildeBigO{\tk[n]}$ rounds using \cref{lem:compute_neigh_qual} and then denote \[t = n^{1/(3\alpha + 1)} \cdot \left(\tn\right)^{2/(3 + 1/\alpha)}.\] In this proof, we strive to achieve a round complexity of $\tildeBigO{t +  \tn}$.
    
    Using \cref{lem:skeleton-graph}, we compute a skeleton graph $G_{\mathcal{S}}=(V_{\mathcal{S}}, E_{\mathcal{S}}, \omega_{\mathcal{S}})$ with sampling probability $1/t$ in $\tildeBigO{t}$ rounds. Observe that $|V_{\mathcal{S}}|=\widetilde{\Theta}(n/t)$ \whp. We now create a $(2\alpha - 1)$ spanner, denoted $K$, of $G_{\mathcal{S}}$ using the algorithm of \cref{thm:spanner_construction}. Each round of the algorithm of \cref{thm:spanner_construction} is simulated over $G_{\mathcal{S}}$ using the local edges of $G$, and thus takes $\tildeBigO{t}$ rounds. As \cref{thm:spanner_construction} takes $\tildeBigO{1}$ rounds, our entire simulation takes $\tildeBigO{t}$ rounds in $G$.
    
    Due to \cref{thm:spanner_construction}, the spanner $K$ has $\tildeBigO*{\alpha \cdot |V_{\mathcal{S}}|^{1+1/\alpha}} \subseteq \tildeBigO*{(n/t)^{1+1/\alpha}}$ edges. Set $x = \max\{(n/t)^{1+1/\alpha}, n\}$ and compute $\tk[x]$ in $\tildeBigO{\tk[x]}$ rounds using \cref{lem:compute_neigh_qual}. Using \cref{thm:optimal_dissemination}, we can broadcast the entire spanner $K$ in $\tildeBigO{\tk[x]}$ rounds.
    
    We desire to show that $\tk[x] \in O(t + \tn)$. If $x = n$, then trivially $\tk[x] = \tn$. Otherwise, 
    \begin{align*}
        x &= (n/t)^{1+1/\alpha} \\
        &= n^{1+1/\alpha} \cdot t^{-(1 + 1/\alpha)} \\
        &= n^{1+1/\alpha} \cdot \left(n^{1/(3\alpha + 1)} \cdot \left(\tn\right)^{2/(3 + 1/\alpha)}\right)^{-(1+1/\alpha)} \\
        &= n^{1+1/\alpha - (1+1/\alpha)/(3\alpha + 1)} \cdot \left(\tn\right)^{-(2+2/\alpha)/(3+1/\alpha)} \\
        &= n^{((\alpha + 1)/\alpha)\cdot(1-1/(3\alpha + 1))} \cdot \left(\tn\right)^{-(2\alpha + 2)/(3\alpha + 1)} \\
        &= n^{((\alpha + 1)/\alpha)\cdot(3\alpha/(3\alpha + 1))} \cdot \left(\tn\right)^{-(2\alpha + 2)/(3\alpha + 1)} \\
        &= n^{3\cdot(\alpha + 1)/(3\alpha + 1)} \cdot \left(\tn\right)^{-(2\alpha + 2)/(3\alpha + 1)} \\
        &= n^{1 + 2/(3\alpha + 1)} \cdot \left(\tn\right)^{-(2\alpha + 2)/(3\alpha + 1)}. 
    \end{align*}
    
    Due to \cref{lem:growth_of_tk}, $\tk[x] \in O\left(\sqrt{x/n} \cdot \tn\right)$, and so
    \begin{align*}
        \tk[x] 
        &\in O\left(\sqrt{x/n} \cdot \tn\right) \\
        &\subseteq O\left(\sqrt{n^{2/(3\alpha + 1)} \cdot \left(\tn\right)^{-(2\alpha + 2)/(3\alpha + 1)}} \cdot \tn\right) \\
        &\subseteq O\left(n^{1/(3\alpha + 1)} \cdot \left(\tn\right)^{-(\alpha + 1)/(3\alpha + 1)} \cdot \tn\right) \\
        &\subseteq O\left(n^{1/(3\alpha + 1)} \cdot \left(\tn\right)^{2\alpha/(3\alpha + 1)}\right) \\
        &\subseteq O\left(n^{1/(3\alpha + 1)} \cdot \left(\tn\right)^{2/(3 + 1/\alpha)}\right) \\
        &\subseteq O(t). 
    \end{align*}

    Thus, in either case, $\tk[x] \in O(t + \tn)$. Therefore, using \cref{thm:optimal_dissemination} we can broadcast all the edges in $K$ in $\tildeBigO{\tk[x]} \subseteq \tildeBigO{t + \tn}$ rounds. Using this information, each node locally computes a $(2\alpha - 1)$ approximation to the distances in $G_{\mathcal{S}}$.

    Next, every node learns its $h = \xi{}t\ln{n}$-hop neighborhood in $\tildeBigO{t}$ rounds, where $\xi$ is the constant from \cref{lem:skeleton-graph}. Due to \cref{lem:skeleton-graph}, every node $v$ sees at least one skeleton node in its $h$-hop neighborhood. Thus, $v$ can find a skeleton node $v_s \in V_\mathcal{S}$ in its $h$-hop neighborhood with minimum $d^h(v, v_s)$, and then $v$ broadcasts $v_s$ and $d^h(v, v_s)$. This step takes $\tildeBigO{\tn}$ rounds, due to \cref{thm:optimal_dissemination}, as every node broadcasts $O(1)$ messages.

    Finally, each node $v$ approximates its distance to any node $w$ by \[\delta(v,w)= \min\{d^h(v,w),d^h(v,v_s)+\hat{d}(v_s,w_s)+d^h(w_s,w)\}.\] It remains to show that this is a $(4\alpha-1)$-approximation.

    \paragraph{Analysis.} Let $v,w\in{V}$. If there exists a shortest path between them of less than $h$ hops, then $\delta(v,w)=d^h(v,w)=d(v,w)$. Otherwise, all $v$-$w$ shortest paths are longer than $h$ hops, and by \cref{lem:skeleton-graph}, there exists a skeleton node $s$ on one of $v$-$w$ shortest paths. Further, \cref{lem:skeleton-graph} guarantees that we can find such a skeleton node $s$ in the $h$-hop neighborhood of $v$. As $s$ sits on a shortest path from $v$ to $w$, it also holds that $d^h(v, s) = d(v, s)$. Finally, it holds that $d^h(w, w_s) \leq d(w, s)$ -- this is true as either there is a shortest path from $w$ to $s$ with at most $h$ hops, in which case $d^h(w, w_s) \leq d^h(w, s) = d(w, s)$, or there is a path from $w$ to $s$ with a skeleton node $s'$ on it that is also in the $h$-hop neighborhood of $w$, in which case $d^h(w, w_s) \leq d^h(w, s') \leq d(w, s)$. Using all of these, we now show that $\delta$ is a $(4\alpha-1)$-approximation.
    \begin{align*}
        \delta(v,w) &= d^h(v,v_s)+\hat{d}(v_s,w_s)+d^h(w_s,w)\\
    &\leq{} d^h(v,v_s) +(2\alpha-1)d(v_s,w_s) + d^h(w_s,w)\\
    &\leq{} d^h(v,v_s) + (2\alpha-1)(d(v_s,v)+d(v,w)+d(w,w_s)) + d^h(w_s,w) \\
    &\leq{} d^h(v,v_s) + (2\alpha-1)(d^h(v_s,v)+d(v,w)+d^h(w,w_s)) + d^h(w_s,w) \\
    &\leq{} (2\alpha-1)d(v,w) + 2\alpha(d^h(v,v_s)+d^h(w_s,w)) \\
    &\leq{} (2\alpha-1)d(v,w) + 2\alpha(d^h(v,s)+d^h(w_s,w)) \\
    &=(2\alpha-1)d(v,w) + 2\alpha(d(v,s)+d^h(w_s,w)) \\
    &\leq{} (2\alpha-1)d(v,w) + 2\alpha(d(v,s)+d(s,w)) \\
    &=(2\alpha-1)d(v,w) + 2\alpha\cdot{}d(v,w) \\
    &=(4\alpha-1)d(v,w)
    \end{align*}

    Observe that the approximation never underestimates, i.e.,~$\delta(v, w) \geq d(v, w)$, so we are done. To see this, consider the formula $\delta(v,w) = d^h(v,v_s)+\hat{d}(v_s,w_s)+d^h(w_s,w)$: Both $d^h(v,v_s)$ and $d^h(w_s,w)$ correspond to actual path lengths, and we have $\hat{d}(v_s,w_s) \geq d(v_s,w_s)$.
\end{proof}

\subsection{Approximating Cut Sizes}\label{subsect:cuts}

Our universally optimal broadcasting tool can be used in combination with different sparsification tools. In particular, we prove \cref{theorem:minCutApprox} using the following sparsifier.

\begin{lemma}[\hspace{-0.01cm}{\cite[Theorem 3.4]{koutis2016simple}}]
\label{thm:congest_cut_sparsifier}
    There is an $\tildeBigO{1/\eps^2}$-round \congest algorithm that given a weighted graph $G=(V,E,\omega)$ and any $\eps>0$, computes a weighted graph $H=(V, \hat{E}, \hat{\omega})$ such that for any cut $\emptyset \subsetneq S  \subsetneq V$ it holds that $(1-\eps)\operatorname{cut}_H(S)\leq{}\operatorname{cut}_G(S)\leq{}(1+\eps)\operatorname{cut}_H(S)$ and $|\hat{E}|=\tildeBigO{n/\eps^2}$ \whp. 
\end{lemma}
Now we prove \cref{theorem:minCutApprox}.

\minCutApprox*

\begin{proof}
    We run the algorithm of \cref{thm:congest_cut_sparsifier} in $\tildeBigO{1/\eps^2}$ rounds and get a cut sparsifier with $k \in \tildeBigO{n/\eps^2}$ edges \whp. By \cref{lem:growth_of_tk} and \cref{thm:optimal_dissemination} we can broadcast the sparsifier in $\tildeBigO{\tk} \subseteq \tildeBigO{\tn/\eps}$ rounds, and then each node can compute any cut approximation locally.
\end{proof}

\section{Lower Bounds}
\label{sec:lower_bound}

In this section, we give matching lower bounds for the problems considered in the previous sections, thus proving their universal optimality. Our lower bounds are divided into two categories: information dissemination and shortest path problems. 


\subsection{Preliminaries}

The general framework for our lower bounds in the \hybrid model
is based on a reduction of an intermediate problem called the \textit{node communication problem}, which describes the complexity of communicating information between distinct node sets in the \hybrid model. 

More specifically, we are given two sets of nodes $A$ and $B$, where nodes in $A$ ``collectively know'' the state of some random variable $X$ and need to communicate it to $B$. The formal definition of the node communication problem is given in \cref{sec:node_comm}. 
The amount of information that can be conveyed from $A$ to $B$ via global communication is fundamentally restricted by the nodes within $h$ hops that each node set can rely on for global communication, as shown in the following lower bound. We emphasize that this lower bound holds even if the topology of $G$ is known to all nodes.



\begin{lemma}[cf. \cite{Kuhn2022, Schneider2023}]
	\label{lem:lower_bound_node_comm}
	Let $A,B$ be disjoint node sets and let $h \leq \hop(A,B)$ and $N := |\calB_{h-1}(A)|$.
	Any algorithm that solves the {node communication problem}, formally defined in \cref{def:node_comm_problem}, with $A,B$ in the \hybridpar{\infty}{\gamma} model with success probability at least $p$ for any random variable $X$, takes at least \smash{$\min\!\big(\frac{pH(X) -1}{N \cdot \gamma}, \frac{h}{2} \m 1\big)$} rounds in expectation.
\end{lemma}

On a very high level, the proof of the lemma shows that a \hybrid algorithm that solves the node communication problem induces a transcript of the global communication that happened during its execution that can be used to obtain a uniquely decodable code for $X$, which implies a lower bound for the size of the transcript by Shannon's source coding theorem. This, in turn, implies a lower bound for the number of rounds it must have taken to communicate this transcript between the two node sets. 

To bound the amount of communication that can occur between the two node sets $A$ and $B$ per round. two ideas are used. First, the two node sets $A$ and $B$ need to be at a sufficiently large hop distance to limit the amount of information that can be transmitted in the local network. Second, the number of nodes that $A$ and $B$ have in their proximity determines the amount of global communication that can occur. This idea was also behind the definition of graph parameter neighborhood quality.

\subsection{Universal Lower Bounds for Information Dissemination}

We show how to use the lower bound for the node communication problem to derive matching lower bounds for the information dissemination problems considered in previous sections. We start by proving the following lemma that directly combines the properties of the graph parameter $\NQ$ with the lower bound for the node communication problem. This lemma forms the foundation for further reductions to our various information dissemination and graph problems.

One of the challenges for making this lower bound universal is that we cannot assume anything about the local communication graph $G=(V,E)$ other than its neighborhood quality. In particular, we cannot completely tailor a graph to prove a certain lower bound, as is customary in existential lower bounds.

\begin{lemma}
    \label{lem:receive_lower_bound}
    Given any local communication graph $G = (V,E)$, then for any assignment of $k$ tokens of size at least $\lceil\log k\rceil+1$ bits to the set $V$, where a single node may hold up $k$ tokens, there exists a node $v \in V$ such that it takes $\tilOm(\NQ)$ rounds in expectation in the \hybrid model, until $v$ knows all tokens with constant probability, even if $v$ is given knowledge of $G=(V,E)$.
\end{lemma}

\begin{proof}
    Let $r := \NQ -1$ and let $v$ be a node with $\calB_r(v) \leq k/r$, whose existence is guaranteed in \cref{lem:neighborhood_small}. We assume that $\NQ(v) \geq 6$, since otherwise the lower bound is trivial.

 Define $h = \lfloor \frac{r}{3} \rfloor -1$. We have $h \geq 1$, due to $\NQ(v) \geq 6$.
    It is possible to identify another node $w \in \calB_{r}(v)$ such that the following two properties are fulfilled: (i) $\calB_{h}(w) \subseteq \calB_{r}(v)$ and (ii) $\calB_{h}(v) \cap \calB_{h}(w) = \emptyset$.
    We pick $w$ as follows. Since $r \leq \NQ \leq D$, there must be an unweighted shortest path $P \subseteq \calB_{r}(v)$ of length $r$ that starts at node $v$. We define $w \in P$ as the node that is at a distance exactly $2h+1$ from $v$.
    Then property (ii) follows directly from distance $\hop(v,w) = 2h+1$. Property (i) is because all other nodes in $\calB_{h}(w)$ are at hop distance at most $3h + 1 \leq r$.

    The reason to construct $w$ is the following. Since $\calB_{h}(v) \cap \calB_{h}(w) = \emptyset$ (property (ii)), at least one of $\calB_{h}(v)$ and $\calB_{h}(w)$ must contain at most $k/2$ tokens, otherwise we would have more than $k$ tokens. This implies that for one of the two nodes $v,w$ at least $k/2$ tokens must be located outside the respective $h$-hop neighborhood. In the following discussion, we assume w.l.o.g.\ that this is the case for $v$.
    
    
    Now we are ready to set up the node communication problem. Let $A := V \setminus \calB_{h}(v)$ and let $B := \{v\}$. We have $h = \hop(A,B)$. Let $N := |\calB_{h}(v)| \leq |\calB_{r}(v)| \leq k/r = k/(\NQ-1)$. Let $X$ be a random $(k/2)$-bit string, where each bit is assigned a value 0 or 1 with equal probability independently. Let the content of the $k/2$ messages outside $\calB_{h}(v)$ each contain one bit and the corresponding index in the bitstring $X$, for which $\lceil\log k\rceil+1$ bits suffice. The entropy of $X$ is \smash{$H(X) = -2^{k/2} \cdot 2^{-k/2} \cdot \log\big(2^{-k/2}\big) = \tfrac{k}{2}$} by \cref{def:entropy}.
    Note that the initial state of node $v$ is independent of $X$, in particular, even if we give initial knowledge of $G$.
    This sets up the conditions of the node communication as specified in \cref{def:node_comm_problem}.
    
    Furthermore, solving the problem that $v$ learns all $k/2$ messages also solves this node communication problem since $v$ can reconstruct the outcome of $X$ from these messages. By this reduction, a lower bound for the node communication problem must be one for the problem of $v$ learning all $k/2$ messages.
    Recall that $p$ is a constant, $\gamma \in \Theta(\log^2 n)$ in the \hybrid model (see \cref{sec:models}), and $h \in \Theta(\NQ)$, so the expected number of rounds is at least 
    \[
        \min\!\left(\frac{pH(X)-1}{N \cdot \gamma}, \frac{h}{2} \m 1\right) \leq \min\!\left(\frac{(\NQ-1)(pk/2 - 1)}{k \gamma}, \frac{h}{2} \m 1\right) \in \tilOm(\NQ). \qedhere
    \]    
\end{proof}

With this lemma on hand, we can prove lower bounds for \kdis and \klrout. To also obtain a lower bound for the \kagg problem, we first show a simple reduction from \kdis to \kagg.

\begin{lemma}[\kdis $\leq$ \kagg]\label{lem:reduction}
If there is a deterministic algorithm that solves \kagg in $t$ rounds in \hybrid or \hybridzero, then \kdis can be solved in  $t + \tildeBigO{1}$ rounds deterministically in the same model.
\end{lemma}
\begin{proof}
Intuitively, since we have $k$ tokens to disseminate and $k$ aggregation results that can be made globally known in $t$ rounds, we would like to place those $k$ tokens in different indices of the values and have the rest of the nodes send the unit element in the rest of the indices. The only problem is that all of the nodes holding tokens need to coordinate in which indices each node should put its tokens, so they match the $k$ tokens to $k$ indices. This can be done by the following algorithm, operating in $\tildeBigO{1}$ rounds.

    First, we use \cref{lem:construct_tree_without_ids2} to construct a virtual tree on the set of all nodes with at least one token to disseminate. After we get this tree $T$, we can in $\tildeBigO{1}$ rounds compute for each node $v$ how many tokens its subtree, including itself, holds. We denote it by $\ell(v)$. This is done by sending the information up from the last level of the tree, aggregating the number in each node, and sending it to the parent node. After that, we begin allocating the indices, starting from the root.
    If the root holds $m$ tokens, it reserves the first $m$ indices for itself, and tells its first child $v_1$ that it should start allocating from the index $m+1$, and to its second child $v_2$, if it exists, that it should start allocating from the index $(m+1)+\ell(v_1)$. The root continues in this fashion for all its remaining children. The nodes lower in the tree also continue in the same fashion. This creates a bijection of the $k$ tokens across the graph to the $k$ indices of aggregation and correctly allocates all indices to each token-holding node, so we may now run the \kagg algorithm in $t$ rounds to solve the \kdis problem to let all nodes learn the $k$ tokens.
    In total, the \kdis algorithm takes $t+\tildeBigO{1}$ rounds.
\end{proof}

Now we prove the lower bounds for \kdis, \kagg, and \klrout.

\broadcastunicastLB*
\begin{proof}
    In both the \kdis and the \klrout problems with arbitrary target nodes, we require that any given node is able to receive any $k$ distinct messages. Therefore, both mentioned problems solve the problem described in \cref{lem:receive_lower_bound} which implies the stated lower bound by reduction.

    In the case of randomized target nodes, there is a probability of $\tfrac{\ell}{n}$ that we pick the node described in \cref{lem:receive_lower_bound} as a target. Hence, when we are given one round less than what is needed to solve the problem from \cref{lem:receive_lower_bound}, this will result in failure with probability at least $1-p$ in case we pick that node as a target, which gives us an overall failure probability of at least $\frac{(1-p)\ell}{n}$. Therefore, to ensure a failure probability less than {$\frac{(1-p)\ell}{n}$} any algorithm must take at least the $\tilOm(\NQ)$ rounds from \cref{lem:receive_lower_bound}.

    By \cref{lem:reduction}, the lower bound for \kdis also applies to \kagg.
\end{proof}


The lower bounds for \kdis, \kagg, and \klrout are almost tight, up to $\tilO(1)$ factors, as they match our upper bounds given in \cref{thm:optimal_dissemination,thm:kaggUB,thm:klsp_UB}. 
 For the case of \kdis and \kagg, our upper bounds even apply to the deterministic setting. 
For \klrout, our algorithms, which work for certain ranges of parameters $(k,\ell)$ with random target nodes, have a much higher success probability than that in the lower bound.\phil{Routing is relatively simple for $\ell \in \bigO(1)$ target nodes, so maybe this should be added to have some range where the the lower bound for arbitrary targets is matched.}
\yijun{Added back this.}



\subsection{Universal Lower Bounds for Shortest Paths Problems}

We show that the neighborhood quality parameter constitutes a lower bound for certain shortest paths problems.
Since we consider graph problems, we have to precisely state the information about the graph and the problem parameters that we grant or withhold from nodes initially, i.e., as input before any communication starts, as this has significant implications on lower bounds. We emphasize that any assumption on the initial knowledge of nodes is only used to demonstrate the universality of our lower bounds and is not required for our algorithmic upper bounds!
In all of the following lower bounds, we always grant knowledge of the topology of the unweighted graph $G = (V,E)$ to all nodes as initial input, which proves the competitiveness of our algorithms with \emph{any} potential algorithm that is optimized for a given graph.

Same as any other universal lower bounds, including \cref{thm:broadcast_unicast_LB} and the ones in \cite{haeupler2021universally}, we must specify clearly which part of the input is considered known to the algorithm for which we desire to prove a lower bound.
Here we take into consideration the initial knowledge of the set of all node identifiers, which highly depends on the computational model (\hybrid or \hybridzero) for which we prove the lower bound. 
For the $k$-SSP problem, we also consider the initial knowledge about the set of source nodes in our computational models. 
In \hybrid, nodes are aware of $n$ and thus the set of identifiers $[n]$. We further assume that nodes are aware of the set of identifiers $\{\ID(s) \in [n] \mid s \in S\}$ in the set of source nodes in the $k$-SSP problem, which only makes the lower bound stronger. However, we assume that nodes do not know the concrete mapping $\ID : V \to [n]$ from the nodes in the graph to their identifiers other than their immediate neighbors. In other words, nodes know $G$ and which identifiers belong to the sources $S$, but not necessarily where nodes with the identifiers of the sources are located in $G$.

In contrast, the definition of \hybridzero (\cref{sec:models}) implies that any node $v\in V$ does initially not even know the set $\{\ID(s) \in [n^c] \mid s \in S \setminus \calB_1(v)\}$.
We start with a simple lower bound that exploits this particular property of \hybridzero with a reduction from the lower bound in \cref{lem:receive_lower_bound}. Note that this lower bound is tight for $k=n$ due to \cref{thm:apsp_hybridzero_unweighted}.

\unweightedLB*

\begin{proof}
    In the \hybridzero model, the set of identifiers of $V$ comes from a polynomial range $[n^c]$, see \cref{sec:models}. We define the set $\{\ID(s) \mid s \in S\}$ as our set of tokens that are distributed among the set of source nodes. Each token is any distinct value in $[n^c]$, which can be written in $c \log n \geq \lceil \log k \rceil +1$ bits.

    By the definition of the $k$-SSP problem \cref{sec:problems}, to solve the shortest path problem, not only must each node $v \in V$ learn a set of approximate distance labels $\tild(v,s)$ (see \cref{sec:problems}) for each source $s \in S$, but must also be able to match identifiers $\ID(s)$ to the corresponding distance label $\tild(v,s)$.

    Hence, by solving the $k$-SSP problem, each node $v \in V$ assigns identifiers $\{\ID(s) \mid s \in S\}$ to the distance labels and must therefore necessarily learn the $k$ tokens. This reduces the problem from \cref{lem:receive_lower_bound} to $k$-SSP in \hybridzero which therefore takes at least $\tilOm(\NQ)$ rounds.
\end{proof}

Next, we turn to lower bounds in the \hybrid model, where we cannot rely on the fact that nodes have to learn source identifiers, since they are known in the first place. We show that even a solution to the $(k,1)$-SP problem where a single target node must learn its distance to $k$ sources requires $\tilOm(\NQ)$ rounds. We proceed with the following lemma.

\begin{lemma}
    \label{lem:node_partition}
    Let $G = (V,E)$, $k \leq n$, $r = \NQ-1$ and assume $r \geq 2$. Then there exists a node $v \in V$ such that $\calB_r(v) \leq k/r$ and there is a partition of $V' := V \setminus \calB_r(v)$ into $V'  = V_1 \cup V_2 = V', V_1 \cap V_2 = \emptyset$ satisfying the following conditions. 
    \begin{enumerate}[(1)]
        \item $|V_1| \geq n/4$ and $|V_2| \geq n/4$,
        \item there is a weight assignment $w:E \to [W]$ such that for any $v_1 \in V_1,v_2 \in V_2$, we have $p(n) d(v,v_1) \leq  d(v,v_2)$ for any (fixed) polynomial $p(n)$.
    \end{enumerate}
\end{lemma}

\begin{proof}
    Let $v \in V$ be a node that satisfies the condition $\calB_r(v) \leq k/r$, which exists according to \cref{lem:neighborhood_small}. 
    We first lower bound the size of $V'$:
    \[
        |V'| = |V\setminus \calB_r(v)| \geq n-\frac{k}{r} = n- \frac{k}{2} \geq n- \frac{n}{2} = \frac{n}{2}
    \]
    Consider a breadth first search tree $T$ with root $v$ (note that we do not actually compute the breadth first tree 
    $T$, we only use it to exploit the properties to show the existence of the objects required in the lemma).
    We add nodes from $V$ to $V_1$ by the breadth first rule in $T$ until $|V_1 \cap V'| = n/4$ (and then remove the nodes from $V\setminus V'$ from $V_1$, so $V_1 \subseteq V'$).
    Then we define $V_2 := V' \setminus V_1$, which has size $|V_2| = |V'|-|V_1| \geq n/4$. This satisfies the size requirements (1) of the partition $V_1,V_2$.

    It remains to assign the weights $w : E \to [W]$ suitably.	
    First, all edges that are not part of $T$ obtain weight $n \cdot p(n)$ (which corresponds to the maximum polynomial weight $W$ we require for the lower bound this lemma is used in).
    Second, all edges that have exactly one endpoint in $V_1 \cup \calB_r(v)$ and one in $V_2$ obtain weight $n \cdot p(n)$.
    Third, all remaining edges obtain weight 1.

    Due to the breadth-first rule, nodes in $V_1$ have no ancestors in $V_2$ in $T$. Therefore all nodes in $V_1$ can be reached from $v$ using at most $n$ tree edges of weight $1$, i.e., $d(v,v_1) \leq n$ for all $v_1 \in V_1$. Nodes in $V_2$ must have an ancestor $V_1 \cup \calB_r(v)$, thus any path from $v$ to some node in $V_2$ must either use an edge of weight $n\cdot p(n)$ or an edge outside the tree of weight $n\cdot p(n)$, thus $d(v,v_2) \geq n \cdot p(n)$ for all $v_2 \in V_2$. This implies \smash{$\frac{d(v_2,v)}{d(v_1,v)} \geq p(n)$} for all $v_1 \in V_1,v_2 \in V_2$ and thus property (2).   
\end{proof}

The previous lemma allows us to reduce the node communication problem from \cref{lem:lower_bound_node_comm} to the $(k,\ell)$-SP problem and prove the following lower bound. Note that this lower bound is almost tight for the parameter ranges given in the corresponding \cref{thm:klsp_UB}.

Fundamental to this proof is that nodes do initially not know where the source nodes are located in $G$. However, the proof works even if the nodes are given initial knowledge of all edge weights in the graph.

\weightedLB*

\begin{proof}
    By \cref{lem:growth_of_tk} we have that $\NQ[n] = \NQ[4 \cdot n/4] \leq 12 \cdot \NQ[n/4]$, thus $\NQ[n] \in \Theta(\NQ[n/4])$ so it suffices to show the claim for $k \leq \frac{n}{4}$ with the difference for larger $k$ being at most a constant. Furthermore, we assume $\NQ \geq 3$, else the claim is trivial.

    We set up the node communication problem using the node $v$ and two node sets $V_1,V_2$ from \cref{lem:node_partition} with $|V_1|,|V_2| \geq n/4 \geq k$. The (weighted) distances from $v$ to the nodes in $V_1$ and $V_2$, respectively differ by a factor of at least $q(n)$ (with those in $V_2$ having the larger distance). Moreover, $V_1$ and $V_2$ are at hop distance $h := \NQ-1$ from $v$.
    Consider a random bit string $X = (x_i)_{i \in [k]}, x_i \in \{0,1\}$ of length $k$ with $\Prob(x_i\!=\!0)=\frac{1}{2}$. We enumerate $k$ nodes from each set $v_{1,i}\in V_1, v_{2,i} \in V_2$ with indices  $i \in [k]$. Then, a set of $k$ source nodes $S = \{s_1, \ldots, s_k\} \subseteq V$ is assigned to nodes from $V_1 \cup V_2$ as follows:  if $x_i =0$ we assign $v_{i,1}$ as source $s_i$, else we assign $v_{i,2}$ as source $s_i$ (meaning we assign the identifier $\ID(s_i)$ to that node).
    
	Assume that $v$ is given the knowledge of the number $k$, the set of identifiers of the nodes in $V_1,V_2$ (but not the mapping $\ID: V \to [n]$!) and which index $i$ is associated with which pair of identifiers in $V_1,V_2$ (which can only make the problem simpler). This is in addition to the knowledge of $G = (V,E,\omega)$.	
	Note that $v$ does not have any initial knowledge of whether $v_{1,i}$ or $v_{2,i}$ was assigned as source $s_i$ and thus has no knowledge about $X$. Or to express this differently: the initial state of $v$ is independent from $X$. To finish setting up the node communication problem let us define $A :=\{v\}$ and $B := V_1 \cup V_2$. We have $h = \hop(A,B)$ and $N := |\calB_{h-1}| \leq |\calB_{h}| \leq k/h = k/(\NQ-1)$.
	
	Presume that we solve the $(k,\ell)$-SP problem with approximation factor at most $q(n)-1$ (where the polynomial $q$ can be chosen freely) with probability at least $p$. Thus the node $v$ learns its distance to each source $s_i$ up to factor $q(n)-1$. Since $v$ knows that its distances to nodes in $V_2$ are a factor of $q(n)$ larger than to those in $V_1$, this approximation lets $v$ determine which nodes were selected as source and thus $v$ will be able to determine the state of $X$ with probability at least $p$.
	
	We have therefore solved the node communication problem with the following parameters. The node $v$ has learned,  with a probability of success $p$, the state of a random variable $X$ which it previously had zero knowledge of, which has Shannon entropy $H(X) = k$ (see \cref{def:entropy}). Plugging those parameters into \cref{lem:lower_bound_node_comm}, the number of rounds for solving $(k,1)$-SP must have been at least 
    \[
        \min\!\left(\frac{pH(X)-1}{N \cdot \gamma}, \frac{h}{2} \m 1\right) \leq \min\!\left(\frac{(\NQ-1)(pk-1)}{k \gamma}, \frac{h}{2} \m 1\right) \in \tilOm(\NQ).
    \]    
    Like in the proof of \cref{thm:broadcast_unicast_LB}, for randomized target nodes we have to factor in the probability that we randomly pick the node $v$ as a target, which implies that for an algorithm to achieve a probability of failure of less than \smash{$\frac{(1-p)\ell}{n}$} with otherwise the same conditions, it has to run for at least $\tilOm(\NQ)$ rounds.
\end{proof}

While the theorem above works even if we grant $G=(V,E,\omega)$ and source node identifiers $\{\ID(s) \mid s \in S\}$ as input to every node, it exploits the fact that by definition nodes cannot connect the source node identifiers $\{\ID(s) \mid s \in S\}$ to the set of nodes in $G$ since the mapping $\ID: V \to [n]$ is unknown. 

That can be turned around, in the sense that if we withhold the weight function $\omega$ from each node initially, then we can instead grant $G=(V,E)$, $\{\ID(s) \mid s \in S\}$ and the complete mapping $\ID : V \to [n]$ to each node and still obtain a lower bound for the shortest paths problem as we show in the theorem below.

This is interesting in the sense that, if nodes are know combined information $G=(V,E)$, the weight function $\omega$, $\{\ID(s) \mid s \in S\}$ and $\ID : V \to [n]$, then clearly each node can locally solve any shortest path problem in a single round. Consequentially, withholding either edge weights $\omega$ or the assignment of node identifiers $\ID : V \to [n]$ is strictly necessary.

\weightedLBSourcesKnown*


\begin{proof}
    The idea for the proof is to encode a random variable with large entropy in the set of weights. For this, we show a reduction from the problem of receiving $k$ tokens given in \cref{lem:receive_lower_bound}. Assume $\NQ \geq 4$, otherwise the lower bound is trivial. Further, since $\NQ[n] \in \Theta(\NQ[n/2])$ by \cref{lem:growth_of_tk}, it suffices to show the claim for $k \leq \frac{n}{2}$
    
    Let $r := \NQ-2$. By construction of $v$ (which maximizes $\NQ(v)$), $\calB_r(v) \leq k/r \leq k/2$ by \cref{lem:neighborhood_small}. Hence, there are at least $|V\setminus \calB_r(v)| \geq n/2 \geq k$ nodes in hop distance at least $r = \NQ -2 \geq 2$.    
    Note that $v$ does not share an edge with any of the nodes $V' := V\setminus \calB_r(v)$.     
    Consider any tree $T$ with root $v$. Since $G$ is connected and $T$ is a tree, each node in $u \in V'$ has a parent edge $\pi(u) \in E$ in $T$ and $\pi(u) \neq \pi(w)$ for all $u,w \in V'$ with $u \neq w$.

    Let $u_1, \ldots , u_k$ be $k$ distinct nodes in $V'$. Consider $k$ messages $m_1, \ldots, m_k$ of length $\sigma := \lceil \log k\rceil +1$ bits, which we interpret as values $m_i \in [2^\sigma-1]$. Note that $2^\sigma \in \bigOa(n)$, therefore, it is possible to assign weights $\omega(\pi(u_i)) = n \cdot  m_i$. For all edges in $E(T) \setminus \{\pi(u_1), \ldots , \pi(u_k)\}$ we assign weight 1, and for all other edges $E \setminus E(T)$ we assign weight $n \cdot 2^\sigma$. 
    
    Every node has a path in the tree $T$ of weight at most $n \cdot (2^\sigma \m 1)$, thus the edges outside $T$ are not viable. By learning the distance $d(v,u_i)$ for each $i \in [k]$, the node $v$ can recover the original messages with a simple arithmetic operation $m_i = d(v,u_i) \mod n$. We say that both endpoints of some edge $\pi(u_i)$ know $m_i$. Note that by construction initially $v$ does not know any of the edge weights $\omega(\pi(u_i))$.

    This concludes the reduction from the problem in \cref{lem:receive_lower_bound} to the $k$-SSSP problem even if the set of sources are completely known, which implies that the lower bound of $\tilOm(\NQ)$ also applies to $k$-SSP under these conditions.
    As before, we would like to point out that the proof holds for randomized target nodes by factoring in the probability to randomly pick $v$ as a target, implying that achieving a failure probability less than \smash{$\frac{(1-p)\ell}{n}$} has the same lower bound $\tilOm(\NQ)$.    
\end{proof}

\section{Existentially Optimal Shortest Paths}
\label{sec:sssp_logtime}

In this section, we show how to obtain a near-optimal deterministic $\tilO(1)$-round SSSP algorithm in the \hybridzero model by combining the techniques of \cite{rozhovn2022undirected} with a new algorithm to efficiently compute an Eulerian orientation in the \hybridzero model.

It was shown in~\cite{rozhovn2022undirected} that a $(1+\eps)$-approximation of SSSP can be computed deterministically in the $\mathsf{PRAM}$ model with linear work and $\tilO(1)$ depth. 
Their main machinery is the simulation of an interface model, called \minor model, defined as follows. Let $G=(V,E)$ be an undirected input graph. 
Different from most of the distributed models, here both nodes and \emph{edges} are assumed to be individual computational units that communicate in synchronous rounds and conduct arbitrary local computations in each round. Initially, nodes know their $O(\log n)$-bit IDs, and edges know the IDs of their endpoints. In each round, communication occurs by conducting the following three operations in that order.

\begin{description}
	\item[Contraction:] Each edge e chooses a value $c_e \in \{\top, \bot\}$, which defines a minor network $G' = (V', E')$ where all edges with $c_e = \top$ are contracted. This forms a set of supernodes \smash{$V'\subseteq 2^{V}$}, where a supernode $s \in V'$ consists of nodes connected by contracted edges. During the contraction, parallel edges are kept and self-loops are removed, so $G'$ is a multi-graph, and for each $e \in E$ that connects nodes of distinct supernodes, there is a distinct corresponding edge in $E'$. 
	\item[Consensus:] Each $v \in V$ chooses an $\tilO(1)$-bit value $x_v$. Each supernode $s \in V'$ sets \smash{$y_s := \bigoplus_{v \in s} x_v$}, where $\bigoplus$ is some pre-defined aggregation operator. All
	$v \in s$ learn $y_s$.
	\item[Aggregation:] Each edge $e \in E'$ connecting supernodes $a \in V'$ and $b \in V'$ learns $y_a$
	and $y_b$, and chooses two $\tilO(1)$-bit values $z_{e,a}$, $z_{e,b}$, where $z_{e,a}$ is intended for $a$ and $z_{e,b}$ is intended for $b$. Every node $v \in s$ in each supernode $s \in V'$ learns the aggregate of its incident edges in $E'$, which is
	\smash{$\bigotimes_{e \in \text{incidentEdges}(s)} z_{e,s}$}, where $\bigotimes$ is some pre-defined aggregation operator. 
\end{description}

The work of \cite{rozhovn2022undirected} shows that SSSP can be solved with a stretch of $(1 \p \eps)$ in $\tilO(1/\eps^2)$ rounds of the \minor model, \textit{if we can additionally} call on an oracle \oracle that solves the \euler problem once per round. For the sake of presentation, the formal definition of \oracle is deferred to \cref{sect:oracle_impl}.

\begin{lemma}[\hspace{-0.01cm}{\cite[Theorem 3.19]{rozhovn2022undirected}}] 
	\label{lem:sssp_minor_aggregation}
	A $(1 \p \eps)$-approximation of SSSP on $G$ can be computed with a total of $\tilO\big(1/\eps^2\big)$ rounds of \minor model and calls to the oracle \oracle.
\end{lemma}

Our goal in this section is to apply \cref{lem:sssp_minor_aggregation} to prove \cref{thm:almost_shortest_sssp}. In \cref{sect:minor_sim}, we show that we can efficiently simulate the \minor model in \hybridzero. In \cref{sect:oracle_impl}, we show that we can efficiently implement the oracle \oracle in \hybridzero.

\subsection{Simulation of \texorpdfstring{\minor in \hybridzero}{Minor-Aggregation in Hybrid0}}\label{sect:minor_sim}





We show that the \minor model can be emulated efficiently in the \hybridzero model. 

\begin{lemma}
	\label{lem:simulate_minor_agregation}
	A round of the \minor model can be simulated in $\tilO(1)$ rounds in the \hybridzero model deterministically.
\end{lemma}

\begin{proof}
Recall from the definition of the \minor model that $V'$ is the set of connected components of the subgraph of the local network $G$ consisting of the edges $e$ that select $c_e = \top$. 
Pretending that all the edges  $e$ that select $c_e = \bot$ do not exist, we may run the overlay network construction algorithm of \cref{lem:construct_tree_without_ids}, in parallel for all connected components $s \in V'$, to compute a virtual rooted tree $T_s$ with depth $O(\log n)$ and maximum degree $O(1)$ in \smash{$\tilO(1)$} rounds on each connected component $s \in V'$.

 Since nodes in the distributed model are computationally unlimited, each edge $e \in E$ can be simulated by both of its endpoints simultaneously.
 As discussed in the proof of \cref{lem:hybrid_zero_aggregate},  the small diameter and maximum degree of $T_s$ enable efficient computation of aggregation and broadcast operations on $T_s$ in $\tilO(1)$ rounds, allowing us to implement the consensus step of \minor in $\tilO(1)$ rounds, in parallel for all $s \in V'$.
	
Recall from the definition of \minor that each edge $e = \{s_1,s_2\} \in E'$ corresponds to an edge $\{v_1, v_2\} \in E$ with $v_1 \in s_1$ and $v_2 \in s_2$, so similarly we may simulate $e = \{s_1,s_2\}$ by \textit{both} of nodes $v_1$ and $v_2$. The values $y_{s_1}$ and $y_{s_2}$ from the consensus step can be exchanged between $v_1$ and $v_2$ via the local edge $\{v_1, v_2\} \in E$ in a single round, so both of them can continue the correct simulation of $e$ in the aggregation step. This is required to choose values $z_{e,a}$ and $z_{e,b}$, which depend on both $y_{s_1}$ and $y_{s_2}$. The aggregation of the values $z_{e,s}$ of all edges $e \in E'$ incident some $s\in V'$ can be done using $T_s$ in $\tilO(1)$ rounds  as before.
\end{proof}

\begin{remark}
\label{rem:minor_aggregation_simulation}
We would want to point out that for the efficient simulation of the \minor model it is insufficient to rely solely on the \NCC model. In particular the aggregation techniques from \cite{Augustine2019} for the \NCC model are not directly applicable since the local communication of the \hybridzero model is strictly necessary.
The reason for this lies in the Aggregation step of the \minor model combined with the fact that the \minor model considers non-contracted edges $E'$ as computational entities that must be simulated by the nodes in the \hybridzero model.

In the aggregation step these ``edge entitities'' $e\in E'$ first have to learn the consensus values $y_a,y_b$ from the incident clusters $a,b \in V'$ based on which they propose new values $z_{e,a},z_{e,b}$, which means that a node $v \in V$ that simulates $e$ has to learn $y_a,y_b$. By the pigeon-hole principle there has to be a node $v \in V'$ that simulates at least $|E'|/|V|$ edges in $E'$, and in general $|E'|/|V| \in \widetilde\omega(1)$. Further, $v$ has to learn the consensus values of all incident clusters of the edges it simulates, which might be all distinct, hence this is $\widetilde\omega(1)$ information in general. Therefore, the communication of all ``edge entitities'' cannot be simulated in $\tilO(1)$ rounds with the $\tilO(1)$ bandwidth that \NCC offers.
\end{remark}

\subsection{Implementation of the Oracle \texorpdfstring{\oracle}{}}\label{sect:oracle_impl}
An \emph{Eulerian graph} is a graph that contains an  Eulerian cycle, It is well-known that a graph is Eulerian if all nodes have even degree.
The goal of the \euler problem is to orient all edges in a given Eulerian graph $H$ in such a way that the indegree and the outdegree are equal for each node.
The oracle \oracle solves the \euler problem on a certain Eulerian graph $H$ that involves \emph{virtual nodes}. The precise specification of \oracle is given as follows.

\begin{definition}[Oracle \oracle]
	\label{def:euler_oracle}
	Let $H$ be an Eulerian graph that is a subgraph of some graph $H'$ that is obtained by adding at most $\tilO(1)$ arbitrarily connected virtual nodes to the communication network $G$.	
	Virtual edges, i.e., the edges incident to at least one virtual node, are given in a distributed form: An edge between a virtual and a real node is known by the real node, and edges between two virtual nodes are known by all real nodes. The \oracle outputs an orientation of edges of $H$ such that the indegree and the outdegree of each node in $H$ are equal.
\end{definition}

\cref{def:euler_oracle} combines the definition of \oracle in~\cite[Definition 3.8]{rozhovn2022undirected} and the definition of virtual nodes in~\cite[Appendix B.4]{rozhovn2022undirected}.

Our strategy for implementing \oracle is to reduce the problem to the case where the arboricity of $H$ is $\tilO(1)$ using the power of unlimited-bandwidth local communication in the \hybridzero model. We begin with considering the \euler problem for graphs with arboricity $\alpha \in \tilO(1)$. 

\begin{lemma}\label{lem:euler_arboricity}
    The \euler problem on an Eulerian graph $H=(V,E)$ with arboricity $\alpha \in \tilO(1)$ that contains $\tilO(1)$ virtual nodes can be solved deterministically in $\tilO(1)$ rounds in the \hybridzero model. 
\end{lemma}
\begin{proof}
    For the sake of presentation, we first focus on the case where there are no virtual nodes. In this case, we show that given any Eulerian graph $H=(V,E)$ with arboricity $\alpha \in \tilO(1)$, the \euler problem can be solved deterministically in polylogarithmic rounds in the \hybridzero model. We emphasize that our algorithm for this task does not exploit the unlimited-bandwidth local communication in the sense that in each round we only communicate $O(\log n)$ bits of information along each edge in $H$.    
    
    We start by running the \emph{forests decomposition} algorithm of Barenboim and Elkin~\cite{BE10}, using $O(\log n)$ rounds in the local network with $O(\log n)$-bit messages, we can orient the edges in such a way that the outdegree of each node is $O(\alpha) \subseteq \tilO(1)$. 

    Let $H^*$ be the result of splitting each node $v \in V$ into $\deg(v)/2$ nodes of degree two. Using the edge orientation, we may assign each node in $H^*$ to a node in $H$ in such a way that the load of each node in $H$ is $O(\alpha) \subseteq \tilO(1)$, as follows. Let $v^*$ be any node in $H^*$. Let $u^*$ be any one of the two neighbors of $v^*$ in $H^*$. Let $u$ and $v$  be the nodes in $H$ corresponding to  $u^*$ and $v^*$, respectively. If the orientation of $\{u,v\}$ is $u \rightarrow v$, then we assign $v^\ast$ to $u$. Otherwise, we assign $v^\ast$ to $v$. 

    Using the above assignment, we can simulate one round of \hybridzero in $H^*$ using $O(\alpha) \subseteq \tilO(1)$ rounds of \hybridzero in $H$, given that each node in $H^\ast$ only sends and receives $O(1)$ messages of $O(\log n)$ bits from each of its neighbors in $H^\ast$ in a round. 

    Since $H^*$ is $2$-regular, $H^*$ is a disjoint union of cycles. To solve the \euler problem on $H$, all we need to do is to orient each cycle of $H^*$. We show that such an orientation can be computed in $O(\log n \log^\ast n)$ rounds of $H^*$. To do so, we utilize a well-known fact that a maximal independent set can be computed in $O(\log^\ast n)$ rounds in the \congest model~\cite{barenboim2018locally}.

    From now on, we focus on orienting one cycle $C$ of $H^\ast$ in a consistent manner. We construct a sequence of cycles $C_0, C_1, C_2, \ldots$ starting from $C_0 = C$, as follows. Given $C_i$, we compute $C_{i+1}$ in two steps.
   \begin{itemize}
       \item Computing a maximal independent set $I_i$ of $C_i$ in $O(\log^\ast n)$ rounds.
       \item Each node $v \in I_i$ removes itself from the cycle and notifies its two neighbors $u$ and $w$ to add an edge $\{u,w\}$.
   \end{itemize} 
  Since $|C_i|/3\leq |I_i| \leq |C_i|/2$, we have $|C_{i+1}| = |C_i| - |I_i| \in \left[|C_i|/2, 2|C_i|/3\right]$. Therefore, after $t \in O(\log n)$ iterations, we obtain that $|C_t| \in O(1)$, so we may orient $C_t$ in $O(1)$ rounds. From $i = t$ down to $i = 1$, given an orientation of $C_i$, an orientation of $C_{i-1}$ can be computed in $O(1)$ rounds. Therefore, in $O(t) \subseteq O(\log n)$ rounds, we obtain the desired orientation of $C = C_0$.

  \paragraph{Virtual nodes.} For the rest of the proof, we show how to adapt the above algorithm to the case where there can be $\tilO(1)$ virtual nodes. Again, for the sake of presentation, we start with a high-level overview of our approach and defer the technical details for the implementation of some steps of the algorithm.
  For each virtual node $v$, let $S_v$ denote the set of real nodes adjacent to $v$. We split $v$ into  $\lfloor |S_v|/2 \rfloor$ new degree-2 nodes and one \emph{residual} virtual node $v_\diamond$, as follows. We pair up the nodes in $S_v$ arbitrarily, with at most one unpaired node when $|S_v|$ is odd.  For each pair $\{u,w\}$, we create a new virtual node $v_{\{u,w\}}$ and connect both $u$ and $w$ to $v_{\{u,w\}}$. 
  Let $v_\diamond$ be the result of removing all incident edges to the paired nodes in $S_v$ from $v$. Since the total number of virtual nodes is $\tilO(1)$, the degree of $v_\diamond$ is also $\tilO(1)$, as $v_\diamond$ can only be adjacent to at most one real node. Therefore, using $\tilO(1)$ rounds, we can broadcast all the edges incident to all residual virtual nodes to the entire network.

  \paragraph{Simulation.}
  Let $H'$ denote the graph resulting from the above splitting procedure. Since all nodes in $H'$ still have even degrees, $H'$ is Eulerian. We now show how to simulate our \euler algorithm on $H'$.
  For all the new degree-2 node $v_{\{u,w\}}$ that we create, both $u$ and $w$ are aware of each other and know the presence of $v_{\{u,w\}}$, so we may let either $u$ or $w$ simulate $v_{\{u,w\}}$. For all the residual virtual nodes, since their adjacency list is already global knowledge, we may let all real nodes in the network jointly simulate them by running \kdis with $k \in \tilO(1)$ to simulate one round of communication, as we only have $\tilO(1)$ such residual virtual nodes and each of them is adjacent to only $\tilO(1)$ nodes. The cost of \kdis is $\tilO(1)$ rounds due to \cref{lem:hybrid_zero_aggregate}.

  \paragraph{Creation of new degree-2 nodes.} Now we discuss the remaining technical detail for how we split a virtual node $v$ into  $\lfloor |S_v|/2 \rfloor$ new degree-2 nodes and one {residual} virtual node $v_\diamond$ and broadcast the adjacency list of $v_\diamond$ to the entire network. We show that this task can be done in  $\tilO(1)$ rounds, so we can afford to do this in parallel for all virtual nodes.

Recall that the set of all real nodes $S_v$ adjacent to $v$ know that they are in $S_v$, so we may use the overlay network construction of \cref{lem:construct_tree_without_ids2} to build a virtual rooted tree $T$ of maximum degree $O(\log n)$ and depth $O(\log{n})$ over $S_v$ in $\tilO(1)$ rounds deterministically. By the end of the construction, each node in $T$ knows the identifiers of its parent and children in $T$. Now, we use $T$ to pair up the nodes in $S_v$ by first letting each node $u \in S_v$ create a token that contains $\ID(u)$. Each node in $T$ sends up its token along the tree edges in $T$ towards the root. At any time, if a node $u \in S_v$ holds more than one token, then $u$ locally pairs up its tokens, leaving at most one unpaired token. All the token pairs stay in $u$. Only the unpaired token is sent up to the parent. After repeating this process for $O(\log n)$ steps, we have paired up all tokens, except for at most one unpaired token for the case where $|S_v|$ is odd. For each pair $\{u,w\}$, to let both $u$ and $w$ be aware of each other, we just need to reverse the communication pattern to send all the tokens back to their initial owners, with the added information about the pairing.

We let the root $r$ of $T$ be responsible for broadcasting the adjacency list of $v_\diamond$ to the entire network. Since the set of all the edges between virtual nodes is already global knowledge, $r$ just needs to broadcast at most one edge that connects $v_\diamond$ to a real node. Since the total number of virtual nodes is $\tilO(1)$, this can be done by running \kdis with $k \in \tilO(1)$, which costs $\tilO(1)$ rounds due to \cref{lem:hybrid_zero_aggregate}.
\end{proof}

Using \cref{lem:euler_arboricity}, we now show that we can efficiently implement a call of \oracle in \hybridzero.

\begin{lemma} 
	\label{lem:sim_euler_oracle}
	A call of the oracle \oracle can be implemented in $\tilO(1)$ rounds deterministically in the \hybridzero model.
\end{lemma}
\begin{proof}
Given \cref{lem:euler_arboricity}, we just need to reduce the problem to the case where the arboricity of the graph is $\tilO(1)$. In this proof, we may completely ignore the $\tilO(1)$ virtual nodes, as their presence can affect the arboricity by at most an additive $\tilO(1)$. Therefore, for the sake of presentation, in the following discussion, we assume that there are no virtual nodes in the considered graph.

Let $G=(V,E)$ be the graph under consideration.
Our idea is to reduce the arboricity of $G$ by greedily orienting disjoint cycles on small-diameter subgraphs in parallel, which will ultimately leave us with a remaining graph of yet unoriented edges with arboricity $O(\log n)$. Observe that by orienting disjoint cycles consistently in one direction, the remaining graph of unoriented edges retains even node degree. In other words, it is still an Eulerian graph. To identify disjoint cycles efficiently, we first compute a $(\chi,\mathcal{D})$-\emph{network decomposition}, which is a partition of the node set $V$ into clusters of diameter $\mathcal{D}$ such that the cluster graph, where two clusters are adjacent if a pair of nodes from the two clusters are adjacent, is properly $\chi$-colored. A network decomposition with $\chi \in O(\log n)$ and $\mathcal{D} \in \tilO(1)$ can be computed in $\tilO(1)$ rounds deterministically in the \LOCAL model~\cite{rozhovn2020polylogarithmic}.  
	
We compute such a network decomposition not on $G$ but on the power graph $G^{2} = (V,E')$, where any two nodes in $V$ at distance at most $2$ hops in $G$ are connected by an edge in $E'$. The round complexity of computing a network decomposition of $G^2$ is asymptotically the same as computing a network decomposition of $G$, since one round of $G^{2}$ can be simulated in two rounds in the local network of the \hybridzero model. Let $C$ be a cluster of the network decomposition of $G^{2}$. Since two nodes in different clusters of the same color cannot be adjacent in $G^2$, the distance in the \emph{original network} $G$ between the two nodes must be greater than two.	
	Let $C'$ be the extended cluster of $C$, defined as the union of $C$ and all nodes within one hop of $C$. Due to the aforementioned distance property, any two \textit{extended} clusters of the same color are node-disjoint.
	
	For each color in $[\chi]$ and for each cluster $C$ with that color in parallel, all nodes in the extended cluster $C'$ learn $C'$ in $\mathcal{D} \in \tilO(1)$ steps via the local network. This knowledge enables all nodes in $C'$ to consistently orient and greedily remove  cycles w.r.t.~$G$ in $C'$, until $C'$ becomes a forest. Note that edges that are nominally removed will still be used for communication in $C'$. Further, we can do this for each extended cluster $C'$ of the current color in parallel since they do not overlap. Each extended cluster is now cycle-free, so all edges in all extended clusters of that color form a single forest.
	Since each node in $G$ is in some cluster $C$ of some color, each edge in $G$ belongs to some extended cluster $C'$, so in the end all edges in $G$ are covered by some forest. The number of forests equals the number of colors $\chi \in \bigO{\log n}$, so the remaining graph has arboricity $\bigO{\log n}$ and can be oriented by \cref{lem:euler_arboricity} in $\tilO(1)$ rounds, even with the presence of $\tilO(1)$ virtual nodes.
\end{proof}

We are now ready to prove \cref{thm:almost_shortest_sssp}.

\optSSSP*

\begin{proof}
By \cref{lem:sssp_minor_aggregation}, to compute a $(1 \p \eps)$-approximation of SSSP on $G$, we just need to perform $\tilO\big(1/\eps^2\big)$ rounds of the \minor model and calls to the oracle \oracle. By \cref{lem:simulate_minor_agregation,lem:sim_euler_oracle}, each of these tasks can be done in polylogarithmic rounds in \hybridzero, so the overall round complexity is $\tilO\big(1/\eps^2\big)$, as required.
\end{proof}

\section{Existentially Optimal Shortest Paths with Multiple Sources}
\label{sec:kssp}

In this section, we consider the $k$-SSP problem and prove \cref{thm:k-ssp} by combining the SSSP algorithm from  \cref{thm:almost_shortest_sssp} with a framework to efficiently schedule multiple algorithms in parallel on a skeleton graph. Since our algorithmic solutions scale with the amount of global communication that we permit, we also give solutions for the more general \hybridpar{\infty}{\gamma} model, where each node in each round can communicate $\lambda$ bits of information using global communication. 
In particular, we show that \smash{$\tilO\big(\!\sqrt{k/\gamma}\big)$} rounds suffice to compute $O(1)$-approximations for $k$-SSP. The claim for the standard \hybrid model is implied by substituting $\gamma = \log^2 n$. Our upper bound for the more general \hybridpar{\infty}{\gamma} model is also tight, up to $\tilO(1)$ factors, as the lower bound of \smash{$\tilOm\big(\!\sqrt{k}\big)$} by \cite{Kuhn2020} can be generalized to \smash{$\tilOm\big(\!\sqrt{k/\gamma}\big)$}, see~\cite{Schneider2023} for details. 


\subsection{Parallel Scheduling of Algorithms on a Skeleton Graph}
\label{ssec:scheduling}

Besides the scheme for SSSP introduced in \cref{sec:sssp_logtime}, our main technical component is to show that we can efficiently run multiple instances of independent algorithms in parallel on a skeleton graph. See \cref{def:skeleton_graph} for a formal definition of a skeleton graph and see \cref{lem:skeleton-graph} for its useful properties. 

For efficient parallel scheduling of algorithms, we also require the concept of so-called \textit{helper sets} introduced by Kuhn and Schneider~\cite{Kuhn2020}. 
The rough idea is as follows. When given a set of nodes $W \subseteq V$ such that each node in $V$ joins $W$ with probability $1/x$ independently, one can assign each sampled node a set of \emph{helper nodes} of a certain size within some small distance, which essentially marks the spot where distance and neighborhood size are in a balance. We parameterize the definition of helper sets from~\cite{Kuhn2020} as follows.

\begin{definition}[Helper sets~\cite{Kuhn2020}]
	\label{def:helpers}
	Let $G = (V,E)$ be a graph, let $x < n$, and let $W \subseteq V$ be selected by letting each node in $V$ join $W$ with probability $1/ x$ independently. 
	A family $\{H_w \subseteq V \mid w \in W\}$ of {helper sets} fulfills the following properties for all $w \in W$ and some integer $\mu \in \smash{\tilT (x)}$.
\begin{enumerate}[(1)]
    \item Each set $H_w$ has size at least $\mu$. 
    \item For each $u \in H_w$, it holds that $\hop(w,u) \leq \mu$.
    \item Each node is member of at most $\tilO(1)$ sets $H_w$.
\end{enumerate}
\end{definition}

The computation of helper sets by \cite{Kuhn2020} works out of the box and we are going to utilize it in the following form.

\begin{lemma}[\hspace{-0.01cm}\cite{Kuhn2020}]
	\label{lem:helpers}
 Let $W \subseteq V$ be selected by letting each node in $V$ join $W$ with probability $1/ x$ independently. 
	A family of helper sets $\{H_w \subseteq V \mid w \in W \}$
 can be computed w.h.p.\ in $\tilO\big(x\big)$  rounds in the \local model.
\end{lemma}

We prove the following scheduling lemma.

\begin{lemma}[Scheduling]
	\label{thm:scheduling}
	Let $\gamma < k < n$ and let $\calS$ be a skeleton graph of the local communication graph $G$ with sampling probability $\sqrt{\gamma/k}$. Let $\calA_1, \dots, \calA_k$ be \hybrid algorithms operating on $\calS$ with round complexity at most $T$. Then $\calS$ can be constructed and  $\calA_1, \dots, \calA_k$ can be executed on $\calS$ in $\tilO\big(\sqrt{k/\gamma}\cdot T\big)$ rounds of the \hybridpar{\infty}{\gamma} model w.h.p.
\end{lemma}

\begin{proof}
	The first step is to compute the skeleton graph $S =(V_\calS, E_\calS)$ with sampling probability \smash{$\sqrt{\gamma/ k}$}, which takes \smash{$\tilO\big(\!\sqrt{k/\gamma}\big)$}  rounds and guarantees $d_\calS(u,v) = d_G(u,v)$ for all $u.v \in V_\calS$ w.h.p., by \cref{lem:skeleton-graph}.	
	As the second step, we compute helper sets for $V_\calS$, which assign each $u \in V_\calS$ a set $H_u \subseteq V$ such that (1) \smash{$|H_u| \in \tildeBigOmega{\sqrt{k/\gamma}}$}, (2) \smash{$\max_{v \in H_u}\! \hop(u,v) \in \tilO\big(\!\sqrt{k/\gamma}\big)$}, and (3) each $v \in V$ is member of $\tilO(1)$ helper sets. This is accomplished w.h.p.\ in \smash{$\tilO\big(\!\sqrt{k/\gamma}\big)$} rounds by \cref{lem:helpers}.
	
	Note that $\calS$ is established in a distributed sense, where each skeleton edge is known by its incident skeleton nodes. The third step is to transmit the incident edges and the input w.r.t.\ each algorithm $\calA_i$ of $u \in \calS$ to its helpers $H_u$ using the local network, which takes \smash{$\tilO\big(\!\sqrt{k/\gamma}\big)$} rounds due to (2).	This gives each helper $v \in H_u$ all the required information to simulate any $\calA_i$ on $\calS$ on $u$'s behalf if the helper is also provided with the messages that $u$ receives during the execution of $\calA_i$.

 \paragraph{Assigning jobs and pairing up helpers.}
	In the fourth step, the simulation of the algorithms $\calA_i$ for node $u\in \calS$ is distributed among all helpers $H_u$ in a balanced fashion, that is, each helper $v\in H_u$ is assigned at most $\ell \leq \lceil k/ |H_u| \rceil \in \tilO\big(\!\sqrt{k\gamma}\big)$ algorithms $\calA_i$ to simulate. We let all helper sets select the same number $\ell$. We enumerate the set of helpers $H_u$ with indices $\{v_j\}$ and assign the algorithms $\calA_{\ell (j-1) +1}, \dots, \calA_{\ell j}$ to $v_j$.	
	Furthermore, for each edge $\{u,u'\} \in E_\calS$ and each $\calA_i$, we pair up the helper $v_j \in H_u$ that simulates $\calA_i$ for $u$ with the helper $v_j' \in H_{u'}$ that simulates $\calA_i$ for $u'$.
	By \cref{def:helpers,lem:skeleton-graph} and the triangle inequality, it holds that $$\hop(v_j,v_j') \leq \hop(v_j,u) + \hop(u,u') + \hop(u',v_j') \in \smash{\tilO\big(\!\sqrt {k/\gamma}\big)}.$$ Hence this step can be coordinated in \smash{$\tilO\big(\!\sqrt {k/\gamma}\big)$} rounds via the unlimited-bandwidth local network.	
 
	\paragraph{Message exchange.} It remains to coordinate the message exchange among helpers that simulate some $\calA_i$ for some skeleton node $u\in \calS$. In the first round, all outgoing messages can be computed locally by helpers $v_j \in H_u$ based on the input they received from $u$. In general, we presume that the simulation and computation of outgoing messages were correct up to the current round. It suffices to show the correct transmission of all messages for all simulated algorithms for the next round. 
 We consider the messages via the local network $\calS$ and the messages via the global network separately.

\begin{description}
    \item[Local network:] The hop distance between helpers $v_j,v_j'$ that simulate $\calA_i$ for two skeleton nodes that are adjacent in $\calS$ nodes is \smash{$\hop(v_j,v'_j) \in \tilO\big(\! \sqrt {k/\gamma}\big)$}. Hence, the local messages can be transmitted in the same number of rounds  $\tilO\big(\! \sqrt {k/\gamma}\big)$ using the unlimited-bandwidth local network. 
    \item[Global network:] Each helper is assigned at most \smash{$\ell \in \tilO\big(\!\sqrt{k\gamma}\big)$} algorithms to simulate. Since each helper has a capacity of $\gamma$ bits for global communication, all global messages can be sent sequentially in \smash{$\tilO\big(\!\sqrt{k\gamma} / \gamma\big) \subseteq \tilO\big(\!\sqrt{k/\gamma} \big)$} rounds. Due to the way the algorithms $\calA_i$ are assigned to helpers $v_j$, no helper receives more messages than it sends in each round.
\end{description}
 Consequently, each helper can simulate a single round of all the assigned algorithms $\calA_i$ in \smash{$\tilO\big(\!\sqrt{k/\gamma} \big)$}  rounds, thus it takes \smash{$\tilO\big(\!\sqrt{k/\gamma} \cdot  T\big)$} rounds to complete the simulation of all $\calA_i$. Finally, any output that was computed by some helper $v_j \in H_u$ in the simulation of $\calA_i$ can be transmitted back to $u$ in \smash{$\tilO\big(\!\sqrt{k/\gamma} \big)$} rounds via the local network.
\end{proof}

\subsection{Solving the  \texorpdfstring{$k$}{k}-SSP Problem}

We now employ the procedure from \cref{ssec:scheduling} to schedule $k$ instances of the SSSP algorithm from \cref{thm:almost_shortest_sssp} on an appropriately sized skeleton graph $\calS$, which allows us to solve $k$-SSP for the case that the sources are part of the skeleton $\calS$.

\begin{lemma}
	\label{lem:k-SSP_on_skeleton}
	Let $\gamma < k < n$ and let $\calS = (V_\calS, E_\calS)$ be a skeleton graph of the local graph $G$ with sampling probability \smash{$\sqrt{\gamma/k}$}. Provided that all sources are in $V_\calS$, we can solve the $k$-SSP problem with stretch $1 \p \eps$ in $\tilO\big(\!\sqrt{k/\gamma}\cdot \tfrac{1}{\eps^2}\big)$ rounds of the \hybridpar{\infty}{\gamma} model w.h.p.	
\end{lemma}

\begin{proof}
	Let $K \subseteq V_\calS$ be the set of $k$ source nodes. Let algorithms $\calA_s$ be an instance of the SSSP algorithm from \cref{thm:almost_shortest_sssp} for some source $s \in K$. By \cref{thm:almost_shortest_sssp}, the algorithms $\calA_s$ have a native round complexity of $\tilO(\tfrac{1}{\eps^2})$ to compute an approximation with stretch $1 \p \eps$. We apply \cref{thm:scheduling} to schedule all $\calA_s, s \in K$ on $\calS$ in parallel in \smash{$\tilO\big(\!\sqrt{k/\gamma} \cdot \frac{1}{\eps^2} \big)$} rounds.	
	After conducting this procedure, every skeleton node $u \in V_\calS$ knows a stretch $1 \p \eps$ distance approximation $\tild_\calS(u,s)$  for every $s \in K$. 
	
	In a post-processing step, each node $v \in V$ that is not part of the skeleton $\calS$ computes its distance estimate to each source as follows. For this, we utilize the fact that skeleton edges $E_\calS$ correspond to paths of at most \smash{$h \in \tilO\big(\!\sqrt{k/\gamma} \big)$} hops in $G$.
	First, node $v$ obtains its $h$-hop distance $d_G^h(v,u)$ to each skeleton node $u$ that is within $h$ hops to $v$ in \smash{$\tilO(\sqrt{k/\gamma})$} rounds using unlimited-bandwidth local communication links. For each source $s \in K$, node $v$ computes $\tild(v,s) := \min_{u \in \calS} d_{G}^h(v,u) + \tild_\calS(u,s)$. By \cref{lem:skeleton-graph} (a) there is a shortest path from $v$ to $s$ which has a skeleton node $u \in V_\calS$ within $h$ hops of $v$ w.h.p. For this particular skeleton node $u$, we have 
	\begin{align*}
		\tild(v,s) & \leq  d_{G}^h(v,u) + \tild_\calS(u,s) \\
		& \leq d_{G}^h(v,u) + (1 \p \eps)d_\calS(u,s) & \text{\cref{thm:almost_shortest_sssp}}\\
		& = d_{G}^h(v,u) + (1 \p \eps)d_G(u,s)  & \text{\cref{lem:skeleton-graph} (b)}\\
		& \leq  (1 \p \eps)\left(d_{G}^h(v,u) + d_G(u,s)\right)\\ 
		& \leq (1 \p \eps)d_{G}(v,s)  & \text{\cref{lem:skeleton-graph} (a)}\tag*{\qedhere}
	\end{align*}
\end{proof}

Finally, we show that solving $k$-SSP on the skeleton $\calS$ can also be used to obtain good approximations for more general cases of the $k$-SSP problem on the local graph $G$ with the different trade-offs as stated in \cref{thm:k-ssp}.

\kssp*

\begin{proof} 
	The case $k \leq \gamma$ is the simplest. Here we have enough available global capacity in $\hybridpar{\infty}{\gamma}$ to execute the $k$ \hybrid SSSP algorithms in parallel in $\tilO(\tfrac{1}{\eps^2})$ rounds (\cref{thm:almost_shortest_sssp}), since each such algorithm requires only $\tilO(1)$ global capacity per round and the local capacity is unlimited.
	
	For the rest of the proof, we assume $\gamma < k$.		
	Let $\calS$ be a skeleton graph with skeleton nodes $V_\calS$ sampled with probability $\sqrt{\gamma/k}$ and skeleton edges $E_\calS$ that correspond to paths of at most \smash{$h \in \tilO\big(\!\sqrt{k/\gamma} \big)$} hops in $G$, which can be computed in  \smash{$\tilO\big(\!\sqrt{k/\gamma} \big)$} rounds, see \cref{lem:skeleton-graph}. 
	Let $K \subseteq V_\calS$ be the set of source nodes for the considered $k$-SSP instance. There are two cases. 

 \paragraph{Random sources.}
	Consider the case where the set of soures $K$ is selected in such a way that each node in $V$ joins $K$ with probability $k/n$ independently. In this case, we are only interested in the standard \hybrid model, i.e., $\gamma \in \tilO(1)$.	By a simple application of a Chernoff bound (\cref{lem:chernoffbound}), we infer that $|K| \in \tilO(k)$ \whp.
 
	We make a distinction on the parameter $k$. If $k \geq n^{2/3}$, then we apply the known $k$-SSP algorithm by~\cite{censor2021sparsity} whose round complexity is \smash{$\tilO(n^{1/3}+\sqrt{k})$}, which equals \smash{$\tilO(\!\sqrt{k})$} when $k \geq n^{2/3}$. Their algorithm can even provide an exact solution, without any approximation error.
		
	Next, we consider the case where $k \leq n^{2/3}$. We may assume that $K$ is a subset of the set of skeleton nodes $V_\calS$ because $k/n$ is at most the sampling probability $\sqrt{\gamma / k}$ when $k \leq n^{2/3}$, as
 \[\frac{k}{n} \leq n^{-1/3} \leq \sqrt{k} < \sqrt{\gamma / k}.\]
 
 Therefore, we may sample the set of skeleton nodes $V_\calS$ by first letting all nodes in $K$ join $V_\calS$ and then letting each node in $V \setminus K$ join $V_\calS$ with probability $\sqrt{\gamma / k} - k/n$.
 Hence we can apply \cref{lem:k-SSP_on_skeleton} to obtain a $(1 \p \eps)$-approximation of $k$-SSP on $G$ using $\tilO\left(\sqrt{k/\gamma} \cdot \frac{1}{\eps^2}\right) \subseteq \tilO\left(\sqrt{k} \cdot \frac{1}{\eps^2}\right)$ rounds, as required.

\paragraph{Arbitrary sources.}
	Finally, consider the case where $K$ consists of $k$ sources that are chosen arbitrarily. We have to deal with the possibility that these sources might be highly concentrated in a local neighborhood and cannot simply be added to the skeleton graph, so this leads to difficulties with scheduling the SSSP algorithms. Even assuming $k \leq n^{2/3}$ does not help here. 
 
 We employ a technique similar to~\cite{Kuhn2022}: Let each source $s \in K$ tag its closest skeleton node $u_s :=  \argmin_{w \in V_\calS} d_{G}^h(s,w)$ as its \textit{proxy source} in $\calS$. Recall that $u_s$ is within \smash{$h \in \tilO\big(\!\sqrt{k/\gamma}\big)$} hops of $s$ by \cref{lem:skeleton-graph} (a). 
	
	Subsequently, we have $k' \leq k$ \textit{skeleton} nodes that are tagged as proxy sources, which again allows us to compute $1 \p \eps$ approximations for $k'$-SSP on the proxy sources in $\tilO\big(\!\sqrt{k/\gamma}\cdot \tfrac{1}{\eps^2}\big)$ rounds using \cref{lem:k-SSP_on_skeleton}. To construct decent approximations to the real source nodes we have to first make the distance $d_{G}^h(u_s,s)$ from each source $s\in K$ to its closest skeleton $u_s$ public knowledge. 
 
	This task is an instance of the \kdis problem, so we may apply our broadcast algorithm in \cref{thm:optimal_dissemination} to solve it.
In the \hybridpar{\infty}{\gamma} model, we can exploit the higher global capacity by running $y \in \tilO(\gamma)$ such algorithms in parallel, where each algorithm is responsible for broadcasting at most $x \in \tilO(k/\gamma)$ tokens, so the round complexity is \[\tilO\big(\mathcal{NQ}_x\big) \subseteq \tilO(\sqrt{x}) \subseteq \tilO\big(\sqrt{k/\gamma})\] by \cref{thm:optimal_dissemination,lem:boundingTk}.

 
 
	
	Now each node $v \in V$ constructs its distance estimate $\hat d$ to each source $s \in K$ by \[\hat d(v,s) := \min\left\{d_G^h(v,s), \tild(v,u_s) + d_{G}^h(u_s,s)\right\}.\] 
 Recall that $v$ already knows $\tild(v,u_s)$ and $d_{G}^h(u_s,s)$ by the previous steps. We can let all $v \in V$ learn $d_G^h(v,s)$ for all $s \in K$ in $h \in \tilO\big(\!\sqrt{k/\gamma}\big)$ rounds using unlimited-bandwidth local communication.

 \paragraph{Analysis.} For the rest of the proof, we analyze the approximation ratio of the distance estimate $\hat d$.  If there is a shortest path between $v \in V$ and $s \in K$ with at most $h$ hops, then we already have $\hat d(v,s) = d_G^h(v,s) =  d_G(v,s)$.
 Otherwise, there must be a node $w \in V_\calS$ within $h$ hops of $s$ on some shortest path from $v$ to $s$ by \cref{lem:skeleton-graph} (1). Due to our choice of $u_s$ and $w$, we have \[d_{G}(u_s,s) \leq d_{G}^h(u_s,s) \leq d_{G}^h(w,s) = d_{G}(w,s) \leq d_{G}(v,s).\] We combine this with the triangle inequality and obtain
	\begin{align*}
		\hat d(v,s) & \leq  \tild(v,u_s) + d_{G}^h(u_s,s)\\
		&  \leq  (1 \p \eps) d_{G}(v,u_s) + d_{G}^h(u_s,s) & \text{\cref{thm:almost_shortest_sssp}}\\
		& \leq  (1 \p \eps)\left(d_{G}(v,w) + d_{G}(w,u_s)+ d_{G}^h(u_s,s)\right) &\text{Triangle inequality}\\ 
		& \leq  (1 \p \eps)\left(d_{G}(v,w) + d_{G}(w,u_s)+ d_{G}^h(w,s)\right) & d_{G}^h(u_s,s) \leq d_{G}^h(w,s)\\ 
		& =  (1 \p \eps)\left(d_{G}(v,s) + d_{G}(w,u_s)\right) & \text{Our choice of $w$}\\ 
		&  \leq  (1 \p \eps)\left(d_{G}(v,s) + d_{G}(w,s) + d_{G}(s,u_s)\right) &\text{Triangle inequality}\\ 
		&  \leq  (1 \p \eps)\left(3d_{G}(v,s)\right) & d_{G}(s,u_s) \leq d_{G}(w,s) \leq d_{G}(v,s)\\
  & = (3 \p \eps') d_{G}(v,s). & \eps' := 3 \eps 
	\end{align*}
  Therefore, we achieve a $(1+\eps')$ approximation in $\tilO\big(\!\sqrt{k/\gamma}\cdot \tfrac{1}{\eps^2}\big)$ rounds, where $\eps' = 3\eps$. By a simple change of parameter $\widetilde{\eps} = \eps/3$, we achieve the desired $(1+\eps)$ approximation result with our algorithm with the same asymptotic round complexity.
\end{proof}

\bibliographystyle{alpha}
\bibliography{ref/hybrid}

\appendix
 
\section{Basic Probabilistic Tools}
\label{apx:generalnotations}
\label{apx:Pseudorandom}
In this section, we introduce a few basic probabilistic tools that we use throughout this paper.
We use the following common forms of Chernoff bounds in our proofs.

\begin{lemma}[Chernoff bound]
	\label{lem:chernoffbound}
Let $X = \sum_{i=1}^n X_i$ be the sum of $n$ independent
binary random variables $X_i \in \{0,1\}$.  
Let $\mu_H \geq \E(X)$ be an upper bound of $\E(X)$.
Let $\mu_L \geq \E(X)$ be a lower bound of $\E(X)$.
We have the following tail bounds.
\begin{align*}
\Prob(X \geq (1+\delta)\mu_H) &\leq \exp\left(\frac{- \delta \mu_H}{3}\right)    & & & & \text{for any } \, \delta > 1,\\
\Prob(X \leq (1-\delta)\mu_L) &\leq \exp\left(\frac{- \delta^2 \mu_L}{2}\right)    &  & & & \text{for any } \, \delta \in [0,1].
\end{align*}
\end{lemma}

The first inequality in \cref{lem:chernoffbound} even holds if we have $k$-wise independence among the random variables $\{X_1, X_2, \ldots, X_n\}$ for $k \geq \lceil \mu_H \delta \rceil$. The following tail bound can be obtained by a change of variable $\delta' = \delta \cdot \mu_H / \E(X)$ in~\cite[Theorem 2]{Schmidt1995}. 

\begin{lemma}[Chernoff bound with $k$-wise independence]
	\label{lem:chernoffbound2}
Let $X = \sum_{i=1}^n X_i$ be the sum of $n$ $k$-wise independent
binary random variables $X_i \in \{0,1\}$.  
Let $\mu_H \geq \E(X)$ be an upper bound of $\E(X)$.
If $k \geq \lceil \mu_H \delta \rceil$, then we have the following tail bound.
\begin{align*}
\Prob(X \geq (1+\delta)\mu_H) &\leq \exp\left(\frac{- \delta \mu_H}{3}\right)    & & & & \text{for any } \, \delta > 1.
\end{align*}
\end{lemma}

To keep the paper self-contained we state the following standard technique known as \emph{union bound}. Roughly speaking it shows that if we have a polynomial number of events occurring with high probability, then they all occur with high probability.

\begin{lemma}[Union Bound]
	\label{lem:unionbound}
	Let $E_1, \ldots ,E_k$ be events, each taking place with probability $1-\tfrac{1}{n^{c+d}}$, for constants $c,d$ and sufficiently large $n$. If $k \leq n^d$, then $E \coloneqq \bigcap_{i=1}^{k} E_i$ takes place with probability $1-\tfrac{1}{n^{c}}$, for sufficiently large $n$.
\end{lemma}


	If only a constant number of events is involved, we typically use the union bound without explicitly mentioning it. Typically, when we use a union bound, we can afford to make $c$ an arbitrarily large constant, so it is possible to use the above lemma in a nested fashion as long as the number of applications is constant, as each application of a union bound decreases the actual constant that appears in the definition of ``w.h.p.'' This constant usually only appears as a factor in the round complexity and can be absorbed by the $\bigOa$-notation, so we can make the initial constants $c$ large enough to account for decreases through applications of the union bound.

An application of the tools above is to analyze the classic \emph{balls-into-bins} distribution, where $m$ tasks are assigned to $\ell$ nodes randomly.
We show the following bound on the maximum number of tasks per node.
 
\begin{lemma}
	\label{lem:balls_in _bins}
	Given $\ell$ bins and \emph{at most} $m$ balls with $m  = \xi \ell \ln n$ for some constant $\xi \in \Theta(1)$, assign each ball to a bin uniformly at random and $k$-wise independently for some $k \geq c \ln n$ with $c \geq \xi/3$. No bin contains more than $a \cdot \tfrac{m}{\ell}$ balls with probability at least $1 - \ell n^{-c}$, where $a = \left(1 + 3c/\xi\right)$.
\end{lemma}

\begin{proof}
	Let $X_{v}$ be the number of balls in bin $v$. We have $\mathbb{E}\big(X_{v}\big) \leq \tfrac{m}{\ell}$. 
 Then a Chernoff bound with $k$-wise independence (\cref{lem:chernoffbound2}) yields
	\[\Prob\Big(X_{v} \geq  a \cdot \tfrac{m}{\ell} \Big) = \Prob\Big(X_{v} \geq (1 \!+\! \tfrac{3c}{\xi})\tfrac{m}{\ell} \Big) \leq \exp\Big(\!-\! \frac{3\xi c \ln n}{3\xi}\Big) =  \frac{1}{n^{c}}.\]
	By the union bound (\cref{lem:unionbound}) over all $\ell$ bins, the event \smash{$\bigcap_{v \in V} X_{v} \!\leq\! \gamma $}  takes place with probability at least $1 - \ell n^{-c}$.
\end{proof}

To construct random variables from a $k$-wise independent balls-into-bins distribution, we use families of $k$-wise independent hash functions, which are defined as follows.

\begin{definition}
	\label{def:hashfunctions}	
	For finite sets $A,B$, let $\mathcal H$ be a family consisting of hash functions $h : A \to B$. Then $\mathcal H$ is called $k$-wise independent if for a random function $h \in \mathcal H$ and for any $k$ distinct keys $a_1, \ldots, a_k \in A$ we have that $h(a_1), \ldots, h(a_k) \in B$ are independent and uniformly distributed random variables in $B$.
\end{definition}

It is known that such a family of hash functions exists in the following form~\cite{Vadhan2012}.

\begin{lemma}
	\label{lem:hashfunctions}
	For $A := \{0,1\}^a$ and $B := \{0,1\}^b$, there is a family of $k$-wise independent hash functions $\mathcal H := \{h : A \to B\}$ such that selecting a random function from $\mathcal H$ requires $k \!\cdot\! \max(a,b)$ random bits and computing $h(x)$ for any $x \in A$ can be done in $poly(a,b,k)$ time.
\end{lemma}

%

	In the proof of \cref{lem:assign_helpers}, we use a random member $h \in \mathcal H$ of a family as described in \cref{def:hashfunctions} to limit the number of messages any node receives in a given round. 
 In particular, assuming that $\calH$ is $k$-wise independent with $k \in \Omega(\log n)$ being sufficiently large, if all nodes send at most $\bigOa(\log n)$ messages to targets that are determined using the random hash function $h$ with a distinct key for each message, then any node receives at most $\bigOa(\log n)$ messages per round \whp. This can be seen by modeling the distribution as a balls-into-bins distribution and using $m \in O(n \log n)$ in \cref{lem:balls_in _bins} to analyze the maximum load per bin.

\section{Estimation of \boldmath\texorpdfstring{$\tk$}{Neighborhood Quality} on Special Graph Classes}
\label{sec:graph_families}
In this section, we estimate the value of $\tk$ for paths, cycles, and any $d$-dimensional grid graphs.

Arguably, the simplest graph to consider is the path graph on $n$ nodes $P_n$. By definition, $\tk=\max_{v\in{}V}\{\tk(v)\}$, so it suffices to identify $\argmax_{v\in{}V}\tk(v)$ and compute its $\tk(v)$ value. The $\tk(v)$ value roughly corresponds to the size of the neighborhood of $v$ that has the bandwidth to communicate $k$ messages with the rest of the graph in $\tk(v)$ rounds. If the neighborhood is sparser, then $\tk(v)$ will be higher.  We formalize this observation as follows.

\begin{observation}
\label{thm:tk_nodes_ineq}
    Let $v,w\in{V}$. If $|\calB_t(v)|\leq{}|\calB_t(w)|$ for all $t\in{}\mathbb{N}^+$, then $\tk(v)\geq{}\tk(w)$.
\end{observation}
\begin{proof}
    It follows directly from \cref{def:bcast_quality}. 
    If for any $t$, $|\calB_t(v)|\leq{}|\calB_t(w)|$, then the minimum $t$ which satisfies the equation in the definition $|\calB_t(v)|\geq{}k/t$ is higher for $v$, so $\tk(v)\geq{}\tk(w)$.
\end{proof}



By \cref{thm:tk_nodes_ineq}, in a path graph, the two endpoints attain the highest $\tk(v)$ value. Therefore, to compute $\tk$ of a path graph, we just need to focus on the endpoints.

\begin{lemma}\label{lem:tkpath}
    For the path graph on $n$ nodes, \[\tk=\begin{cases}
        \Theta{}(\sqrt{k}) & \text{if } k\leq{}D^2+D \in O(n^2), \\
        D=n-1 & \text{otherwise. }
    \end{cases}\]
\end{lemma}
\begin{proof}
Let $r$ be an endpoint of the path.
 We have   $|\calB_t(r)|=t+1$ for all $t \leq D = n-1$, so the condition $|\calB_t(v)| \geq k/t$ in \cref{def:bcast_quality} can be rewritten as $t+1\geq{}k/t$ or equivalently $t^2+t\geq{}k$. If $k\leq{}D^2+D$, then the smallest $t$ meeting the condition $t^2+t\geq{}k$ is at most $D$ and satisfies $t \in \Theta(\sqrt{k})$, so $\tk \in \Theta(\sqrt{k})$.
 Otherwise, the smallest $t$ meeting the condition $t^2+t\geq{}k$ is at least $D$, so $\tk = D$.
\end{proof}


\pathsTk*
\begin{proof}
The case where $G$ is a path follows from \cref{lem:tkpath}.
For the case where $G$ is a cycle, the value $\tk$ is asymptotically the same as paths.
More formally, let $v$ be any node in an $n$-node cycle graph $G$. we have   $|\calB_t(v)|=2t+1$ for all $t \leq D$, so the condition $|\calB_t(v)| \geq k/t$ in \cref{def:bcast_quality} can be rewritten as $2t+1\geq{}k/t$ or equivalently $2t^2+t\geq{}k$. If $k\leq{}2D^2+D \in O(D^2)$, then the smallest $t$ meeting the condition $2t^2+t\geq{}k$ is at most $D$ and satisfies $t \in \Theta(\sqrt{k})$, so $\tk \in \Theta(\sqrt{k})$.
 Otherwise, the smallest $t$ meeting the condition $2t^2+t\geq{}k$ is at least $D$, so $\tk = D$.
\end{proof}


We now turn to compute $\tk$ for $n$-dimensional grids. 
In this section, we write $S_r(v)=\{v \mid d(u,v)=r\}$ to denote the set of nodes at a distance exactly $r$ from $v$.
By \cref{thm:tk_nodes_ineq}, to compute $\tk$, it suffices to focus on the corner nodes, since they have the smallest neighborhoods compared to all other nodes. As a warm-up, when $d=2$, $|\calB_r(v)|=|\calB_{r-1}(v)|+r+1$, since we add another sub-diagonal with every radius expansion. By induction, we obtain that \[|\calB_r(v)|=\sum_{i=0}^{r}i+1=\sum_{i=1}^{r+1}i=\frac{(r+1)(r+2)}{2}=\frac{r^2+3r+2}{2} = \Theta(r^2).\]
Substituting this into \cref{def:bcast_quality}, we obtain that $\tk=\min\{\Theta(k^{1/3}), D\}$.
We generalize this argument in the following statements.

\begin{lemma}
\label{lemma:square_grid_ring}
    Let $G=(V,E)$ be a $d$-dimensional grid, with $|V|=n=m^d$. Let $w$ be a corner node.  If $r \leq m-1 = n^{1/d} - 1$, then $|S_r(w)|=\binom{r+d-1}{d-1}$.
\end{lemma}
\begin{proof}
As long as $r \leq m-1$, the set $S_r(w)$ can be represented by $S_r(w)=\left\{v\in{}\mathbb{N}^d \mid \sum_{i=1}^{d}v_i=r\right\}$, as the distances in a grid obey the $L_1$ metric, and without loss of generality, the corner node $w$ corresponds to the $\vec{0}$ vector in $\mathbb{N}^d$. Therefore, the number of elements in $S_r(w)$ is exactly $\binom{r+d-1}{d-1}$.
\end{proof}

\begin{lemma}
\label{lemma:square_grid_neighborhood}
    If $r \leq n^{1/d} - 1$, then $|\calB_r(v)|=|S_0(w)\cup{}\cdots\cup{}S_r(w)|=\binom{r+d}{d}$.
\end{lemma}
\begin{proof}
   The lemma follows from the following calculation.
\begin{align*}
|\calB_r(v)|
&=|S_0(w)\cup{}\dots\cup{}S_r(w)|\\
&=\sum_{i=0}^{r}|S_i(w)|\\
&=\sum_{i=0}^{r}\binom{i+d-1}{d-1}\\
&=\sum_{i=0}^{r}\binom{i+m}{m} & m=d-1\\
&=\sum_{j=m}^{m+r}\binom{j}{m} & j=i+m\\
&=\sum_{j=m}^{m+r}\left[\binom{j+1}{m+1}-\binom{j}{m+1}\right] & \text{Pascal's Triangle}\\
&=\binom{m+r+1}{m+1}=\binom{r+d}{d} & \text{Telescoping Sum} &&& \qedhere
\end{align*}   
\end{proof}

This leads us to the following theorem for $d$-dimensional grids.

\gridGraphsTk*

\begin{proof}
As discussed earlier, we just need to focus on corner nodes. Let $v$ be any corner node in an $n$-node $d$-dimensional grid $G$. From \cref{lemma:square_grid_neighborhood}, we have   $|\calB_t(v)|=\binom{t+d}{d}$ for all $t \leq m-1 = n^{1/d} - 1$, in which case the condition $|\calB_t(v)| \geq k/t$ in \cref{def:bcast_quality} can be rewritten as $t \cdot \binom{t+d}{d}\geq{} k$. 
First, we consider the case where $k \leq (m-1) \cdot \binom{(m-1)+d}{d} \in O\left(n^{1 + 1/d}\right)$.
Since $t \cdot \binom{t+d}{d} \in \Theta\left(t^{d+1}\right)$, we infer that the smallest $t$ meeting the condition $t \cdot \binom{t+d}{d}\geq{} k$ satisfies $t \in \Theta\left(k^{1/(d+1)}\right)$, so $\tk \in \Theta\left(k^{1/(d+1)}\right)$.

For the rest of the proof, consider the remaining case where $k > (m-1) \cdot \binom{(m-1)+d}{d} \in \Omega\left(n^{1 + 1/d}\right)$. Since $\tk$ is non-decreasing in $k$, we know that $\tk \in \Omega\left(k^{1/(d+1)}\right) \subseteq \Omega\left(n^{1/d}\right) \subseteq \Omega(D)$. Combining this with the trivial upper bound $\tk \leq D$, we infer that $\tk  \in \min\left\{\Theta\left(k^{1/(d+1)}\right), D\right\}$, as required.
\end{proof}

\section{The Node Communication Problem}
\label{sec:node_comm}

In this section, we review the \emph{node communication problem}~\cite{Kuhn2022, Schneider2023}, which is a useful tool in proving lower bounds.
The Shannon entropy of a random variable $X$ can be thought of as the average information conveyed by a realization of $X$ and is defined as follows.

\begin{definition}[Entropy, cf., \cite{Shannon1948}]
	\label{def:entropy}
	The Shannon entropy of a random variable $X\!:\! \Omega \!\to\! S$ is defined as $H(X) := - \!\sum_{x \in S} \mathbb{P}(X \!=\! x) \log \big(\mathbb{P}(X \!=\! x) \big)$. For two random variables $X,Y$ the \textit{joint entropy} $H(X,Y)$ is defined as the entropy of $(X,Y)$. The \textit{conditional entropy} is $H(X|Y) = H(X,Y) - H(Y)$. 
\end{definition}

Based on the concept of entropy we can express what it means that nodes in a hybrid network ``know'' or are ``unaware'' of the outcome of some random variable $X$ (see also \cite{Kuhn2022}).

\begin{definition}[Knowledge of Random Variables]
	\label{def:node_knowledge}
	Let $V$ be the set of nodes in a distributed network. Let $A \subseteq V$ and let $S_A$ be the state (including inputs) of all nodes in $A$ (we interpret $S_A$ as a random variable). Then the nodes in $A$ \textit{collectively know} the state of a random variable $X$ if $H(X|S_A) = 0$ (see Definition \ref{def:entropy}), i.e., there is no new information in $X$ provided that $S_A$ is already known.		
	Similarly, we say that $X$ is \textit{unknown} to $B \subseteq V$, if $H(X|S_B) = H(X)$, meaning that all information in $X$ is new even if $S_B$ is known. Equivalently, we can define this as $S_B$ and $X$ being independent. We can extend these definitions to communication parties Alice and Bob with states $S_{\text{Alice}}$ and $S_{\text{Bob}}$ where Alice knows $X$ if $H(X|S_{\text{Alice}}) = 0$ and $X$ is unknown to Bob if $H(X|S_{\text{Bob}}) =  H(X)$.
\end{definition}

The node communication problem was introduced as an abstraction for the problem where a part of the network has to learn some information that only another distant part of the network knows.

\begin{definition}[Node Communication Problem]
	\label{def:node_comm_problem}
	Let $G=(V,E)$ be some graph. Let $A,B \subseteq V$ be disjoint sets of nodes and $h := \hop(A,B)$. Furthermore, let $X$ be a random variable whose state is collectively known by the nodes $A$ but unknown to any set of nodes disjoint from $A$. An algorithm $\mathcal A$ solves the \emph{node communication problem} if the nodes in $B$ collectively know the state of $X$ after $\mathcal A$ terminates. We say $\calA$ has success probability $p$ if $\calA$ solves the problem with probability at least $p$ for \emph{any} state $X$ can take.
\end{definition}

\end{document}